\documentclass[11pt]{article}

\usepackage{psfrag}
\usepackage{amsmath}
\usepackage{amssymb}
\usepackage{amscd}
\usepackage{picinpar}
\usepackage{times}
\usepackage{pb-diagram}
\usepackage{graphicx}
\usepackage{wrapfig}
\usepackage{xspace}
\usepackage{algorithmic, algorithm}   
\usepackage{multirow}

\def\T{\mathcal{T}}

\newtheorem{theorem}{Theorem}[section]
\newtheorem{lemma}[theorem]{Lemma}

\newtheorem{definition}[theorem]{Definition}

\newenvironment{proof}[1][Proof]{\begin{trivlist}
\item[\hskip \labelsep {\bfseries #1}]}{\end{trivlist}}
\newcommand{\qed}{\nobreak \ifvmode \relax \else
      \ifdim\lastskip<1.5em \hskip-\lastskip
      \hskip1.5em plus0em minus0.5em \fi \nobreak
      \vrule height0.75em width0.5em depth0.25em\fi}

\begin{document}

\title{Discrete Conformal Deformation: Algorithm and Experiments}

\author{ 
Jian Sun \thanks{Mathematical Sciences Center, Tsinghua University, Beijing, 100084, China. \textit{Email: jsun@math.tsinghua.edu.cn.}}
\and Tianqi Wu \thanks{Mathematical Sciences Center, Tsinghua University, Beijing, 100084, China. \textit{Email: mike890505$@$gmail.com.}}
\and Xianfeng Gu \thanks{Department of Computer Science, Stony Brook University, New York 11794, USA. \textit{Email:  gu$@$cs.stonybrook.edu}}
\and Feng Luo \thanks{Department of Mathematics, Rutgers University, New Brunswick, NJ 08854, USA.\textit{Email: fluo$@$math.rutgers.edu}}
}


\date{}

\maketitle

\begin{abstract}
In this paper, we introduce a definition of discrete conformality for
triangulated surfaces with flat cone metrics and describe an algorithm for solving the problem
of prescribing curvature, that is to deform the metric discrete conformally 
so that the curvature of the resulting metric coincides with the 
prescribed curvature. We explicitly construct a discrete conformal map 
between the input triangulated surface and the deformed triangulated surface.  
Our algorithm can handle the surface with any topology with or without boundary, 
and can find a deformed metric for any prescribed curvature  satisfying the 
Gauss-Bonnet formula. In addition, we present the numerical examples to show 
the convergence of 
our discrete conformality and to demonstrate the efficiency and the robustness
of our algorithm. 
\end{abstract}

\newpage
\section{Introduction}
In this paper, we introduce a definition of discrete conformality for
triangle meshes and describe an algorithm for solving the problem
of prescribing curvature, that is to deform the metric discrete conformally 
so that the curvature of the resulting metric coincides with the 
prescribed curvature. In addition, we explicitly construct a 
discrete conformal map between the original triangle mesh 
and the deformed triangle mesh. The problem of prescribing curvature
has many applications in various engineering fields including
computer vision, image processing, and computer graphics. 
For instance, by setting the curvature to be zero, one can discrete conformally
flatten a triangle mesh into the plane and thus obtain a discrete conformal 
parametrization of the mesh. 

Our discrete conformal deformation consists of two basic operations:
{\it vertex scaling} and 
{\it cocircular diagonal switch} (see Figure~\ref{fig:basic-operations}). 
Assume a closed surface $S$ 
is equipped with a triangulation $T=(V, E, F)$ where $V$, $E$ and $F$ are the vertex set, 
the edge set and the triangle set, respectively. An edge length assignment
$l: E\rightarrow \mathbb{R}^+$ assigns any edge $e\in E$ with the length
$l(e)$, which determines a metric on $S$ provided that
the triangle inequalities are satisfies for all triangles in $T$. 
The operation of vertex scaling is a special way of changing the edge lengths. Specifically, 
the vertex scaling of the edge length assignment $l$ by a 
function $w: V\rightarrow \mathbb{R}$ is another edge length assignment, 
denoted $w *_T l$, so that for any edge $e\in E$ with the endpoints $u, v\in V$ 
\begin{equation}
w*_T l(e) = e^{w(u) + w(v)}l(e). 
\label{eq:vertexscaling}
\end{equation}
We call the function $w$ the discrete conformal factor. A discrete conformal factor $w$
is legitimate if the edge length assignment $w *_T l$ satisfies the triangle inequalities
for all the triangles in $T$. By a simple dimension counting,  the vertex scalings of $l$ 
will not in general cover all possible edge length assignments on $E$.
For an edge $e$ in $T$, 
denote $f$ and $f'$ the two triangles in $T$ incident to $e$, 
and $e_1, e_2$, respectively $e'_1, e'_2$, are two other edges of $f$ respectively $f'$ listed
counterclockwise, as shown in Figure~\ref{fig:basic-operations}. Define the length cross ratio
of the edge $e$ under the edge length assignment $l$ as $c_l(e) = (l(e_1)l(e'_1))/(l(e_2)l(e'_2))$. 
Then it is easy to verify that an edge length assignment $\tilde{l}$ is a vertex scaling of $l$ if 
and only if the length cross ratio is preserved, i.e., $c_l(e) = c_{\tilde{l}}(e)$, for any edge $e$ 
in $T$. 

The vertex scaling operatoion was introduced by	Ro$\check{c}$ek and Williams in physics~\cite{Rocek} 
and independently by Luo in mathematics~\cite{luo}. Luo established a (convex) variational principle 
associated to the vertex scaling operation.  This variational principle has many nice properties.
The one most relevant to the applications in engineering fields
is that there is an efficient algorithm to solve the problem of prescribing curvature, 
and thus the problem of discrete conformal parametrization as a special case. 
The main observation is that given a triangulation $T$ and an edge length assignment $l$
over its edges $E$, 
the conformal factor $w$ so that the metric determined by $w *_T l$ achieves 
the prescribed curvature is the unique minimizer of a convex energy, 
whose gradient and Hessian can be explicitly estimated. 
Thus, the minimizer can be efficiently computed by Newton's method. 
The convex energy discovered by Luo~\cite{luo} takes the form of path integral of 
a differential one-form.  An explicit formula of this convex energy based on Lobachevsky function
was found later by Springborn et al.~\cite{ssp}. 
However, there are the cases where the discrete conformal factor $w$ solving the prescribing 
curvature problem does not exist. In fact, in those cases, the minimizer $w$ of the above convex 
energy is not legitimate. 

\begin{figure}[t]
\begin{center}
\begin{tabular}{cc}
\includegraphics[width=0.5\textwidth]{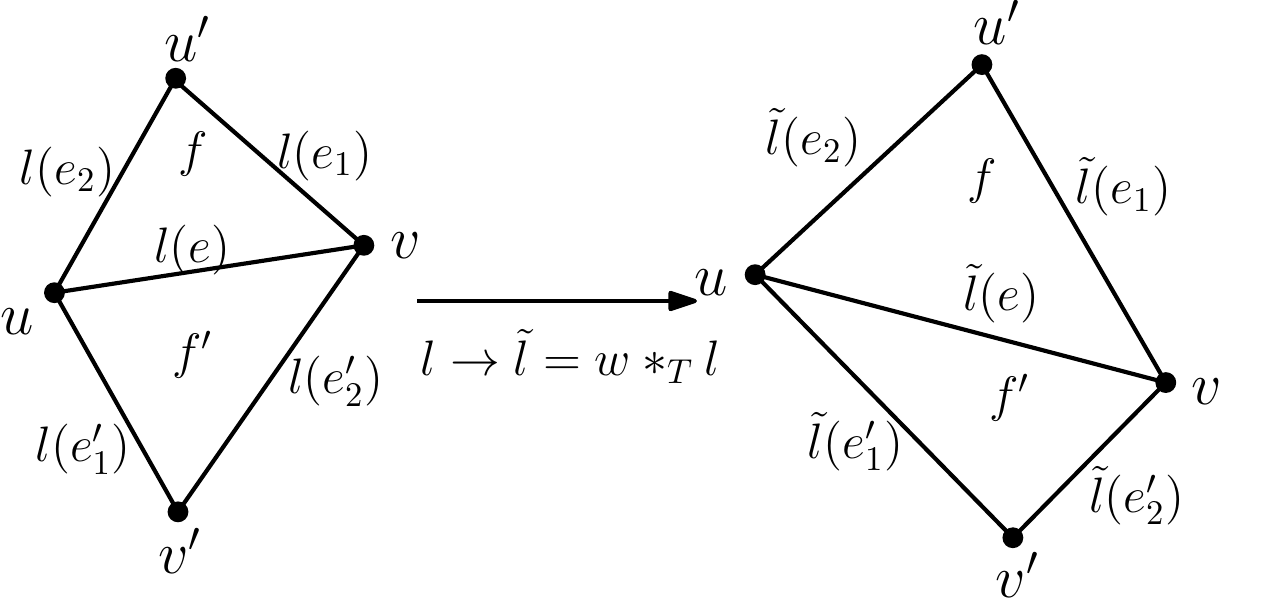} &
\includegraphics[width=0.5\textwidth]{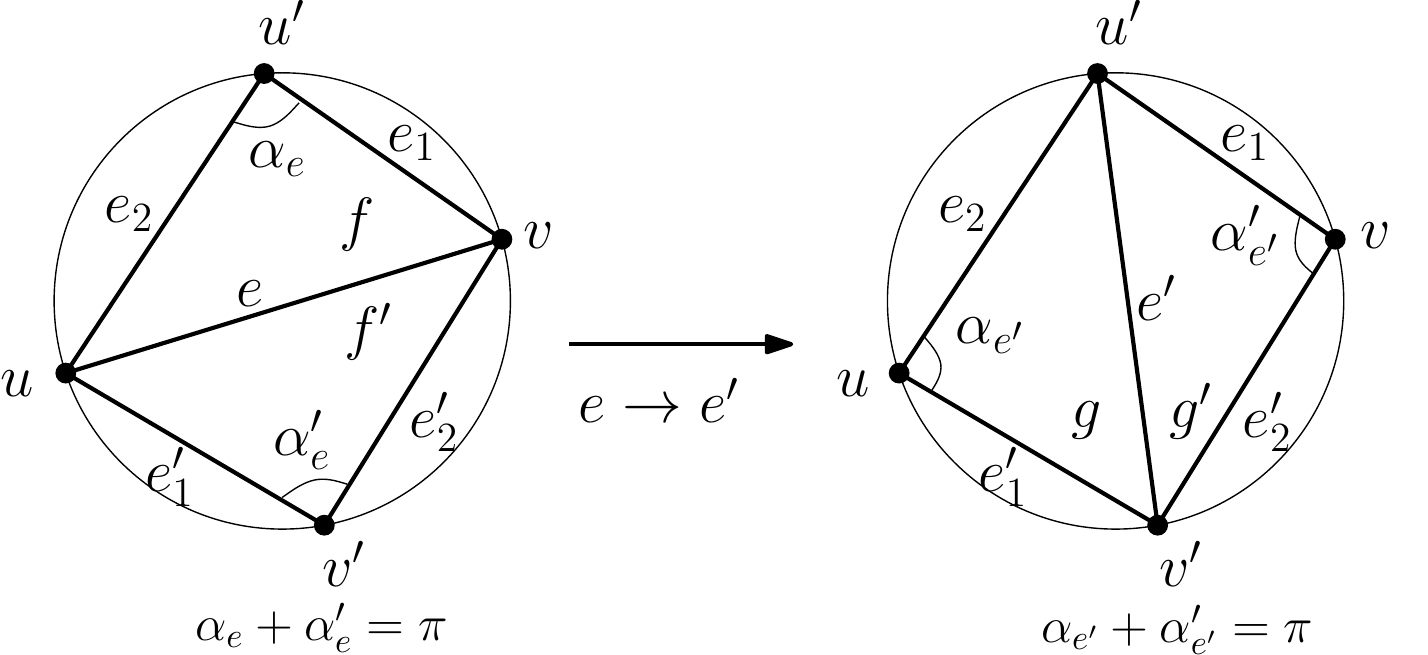}\\
Vertex scaling & Cocircular diagonal switch
\end{tabular}
\end{center}
\vspace{-0.1in}
\caption{Two basic operations in discrete conformal deformation.
\label{fig:basic-operations}}
\end{figure}

To tackle the issue of existence, we introduce the second operation: diagonal switch. 
Let $e$ be an edge in $T$ adjacent to two distinct triangles $f$ and $f'$ in $T$, 
the diagonal switch of the edge $e$ replaces $e$ by the other diagonal $e'$
of the quadrilateral $f \cup f'$. This also replaces the triangles $f, f'$ by two new triangles
$g, g'$, as shown in Figure~\ref{fig:basic-operations}, and produces a new triangulation $T'=(V, E', F')$ 
on $S$.  With the diagonal switch operation, 
we can extend the domain of legitimate discrete conformal factors. To see this, we start with a Euclidean
triangulation $T$ and an initial edge length assignment $l$ over the edges in $T$, and 
then we vertex scale $l$ by continuously changing the function $w$ along the gradient of the 
above convex energy. At some point, some triangle in $T$ may become degenerate under the new 
edge length assignment $w*_T l$, that is the triangle inequality becomes equality. 
It was shown by Luo~\cite{luo} that in any degenerated triangle one of its inner angle must equal $\pi$. 
By diagonally switching the edge opposite to that angle, the degenerated triangle is removed. 
In this way, one may make the conformal 
factor $w$ legitimate. However, the diagonal switch operation brings up many complicated issues. 
For instance, with diagonal switch, a priori, the energy depends on not only the discrete conformal 
factor $w$, but also the triangulations on $S$, which are combinatorial structures. 
{\it Can the energy with combinatorial variables still be convex?}
In addition, the new edges are emerging with the diagonal switch operation. 
{\it What is the assignment of the lengths for these edges which are not in the 
initial triangulation?}
Furthermore, {\it if multiple triangles simultaneously become degenerate, 
do different sequences of diagonal switch operations lead to the same solution?}

Our key observation to make the operation of diagonal switch work nicely is to switch an edge 
well before its incident triangle become degenerate. Specifically, an edge $e\in E$ shared by the triangles $f, f'$
is switched when it fails to be Delaunay, that is the sum of the angles opposite to $e$ in $f$ and $f'$
becomes bigger than $\pi$. We call it cocircular diagonal switch 
as the edge $e$ is switched at the moment that the quadrilateral $f\cup f'$ 
become cocircular. See Figure~\ref{fig:basic-operations}.
We will answer the above three questions later. 
Roughly speaking,  two PL metrics on $S$ are discrete conformal if one can be deformed to the other by a sequence of vertex 
scalings and cocircular diagonal switches. The rigorous definition is given in Definition~\ref{dc}. Based on this discrete
conformality, there always exists a PL metric which is discrete conformal to the initial PL metric and 
achieves any prescribed curvature. Furthermore, such metric can be computed using an efficient algorithm through
minimizing a convex energy. The algorithm can deal with the surface with any topology with or without boundary. 

In this paper, we describe our theory of discrete conformality with a focus on explaining
the algorithm for solving the problem of prescribing curvature, and present the
numerical examples, in particular to show the convergence of our discrete conformality. 
For the rigorous mathematical treatment of our theory, the interested readers are 
referred to~\cite{glsw1, glsw2}. 

\vspace{0.1in}
\noindent{\bf Related work.~}
There has been a lot of research into discrete conformality and we will not attempt 
a comprehensive review here. Instead, we focus on methods closely related to 
ours. Note all previous work deals with the concept of discrete conformality with 
fixed triangulations.  

Bobenko, Pinkall and Springborn~\cite{bps} introduced a geometric interpretation 
to the vertex scaling operation in both Euclidean and hyperbolic geometry using the 
volume of generalized hyperbolic tetrahedron. Glickenstein~\cite{Glickenstein1, Glickenstein2} 
extended the vertex scaling operation  to 3-dimensional piecewise flat manifolds.

One closely related work is circle patterns where a system of circles associated with
vertices. Two triangulated surfaces are considered conformally equivalent if the 
intersection angles of the circles are equal in both triangulated surfaces. 
The idea of approaching discrete
conformality through circle patterns goes back to Thurston~\cite{Stephenson03}. 
Rodin and Sullivan~\cite{RS} proved the Thurston's conjecture that Riemann mapping 
can be approximated by tangential circle packings (i.e., circle patterns with $0$ 
intersection angles) of hexagonal triangulations, and He and Schramm~\cite{he98} 
later showed the convergence is $C^{\infty}$. Colin de Verdi\'{e}re~\cite{verdiere} 
discovered a variational 
principle for circle patterns with intersection angles in $[0, \pi/2]$, 
and Chow and Luo~\cite{chow} introduced discrete Ricci flow based on circle packing and 
established a convergence theorem.. 
An issue with circle patterns is that not all metrics can be realized by circle patterns with 
intersection angles in $[0, \pi/2]$. To tackle this issue, Bowers and Stephenson~\cite{Bowers} 
introduced inversive circle patterns where circles are not necessarily intersect.
Guo~\cite{Guo} established a variational principle for inversive circle patterns 
and showed that inversive distance circle patterns are locally rigid, i.e., locally 
determined by the curvature. Luo gave a proof for global rigidity in~\cite{Luo2011}.
However, the question of existence to the problem of prescribing curvature 
remains open. It is interesting to see if diagonal switch can help solving the existence problem. 
Many practical algorithms based on circle patterns have been proposed for conformally 
flattening triangulated surfaces, including~\cite{Kharevych, Jin}.

Conformality is closely related to harmonicity. Pinkall and Polthier~\cite{pinkall} proposed
an approach for flattening a triangulated surface by computing a pair of 
discrete harmonic functions conjugate to each other. 
Gu and Yau~\cite{Gu:2003} proposed a method to conformally flatten a surface
into the plane using holomorphic one-form. Assume $h=f(z)dz$ is a 
holomorphic one-form of the surface, and then the metric 
$|f(z)|^2 dz d\bar{z}$ is conformal and flat when $f(z) \neq 0$. 
Noticing that any holomorphic one-form can be decomposed as 
$h = \omega + i(*\omega)$ where $w$ is a real harmonic one-form and 
$*w$ is its conjugate, Gu and Yau developed discrete algorithms
to approximate holomorphic one-forms from a triangulated surface by 
computing discrete harmonic one-forms and their conjugates. 

Another class of methods achieve conformality by minimizing conformal distortion. 
In these methods, piecewise linear maps are used to approximate actual conformal maps. 
Noticing that $h_{\bar{z}} = 0$ for a conformal map $h$, 
Levy et al.~\cite{Levy} proposed a method to find a piecewise linear map $f$ from a 
triangulated surface into the plane by minimizing $\|| f_{\bar{z}}|\|_{L_2}$. 
Lipman~\cite{Lipman12} proposed a method to find a piecewise linear map $f$ whose
conformal distortion $\frac{|f_z|+|f_{\bar{z}}|}{|f_z|-|f_{\bar{z}}|}$ is bounded. 
Lui~\cite{Lui} et al. noticed that the magnitude of the Beltrami coefficient 
$\mu = \left |\frac{h_{\bar{z}}}{h_z}\right|$ is constant for the extremal map $h$ 
(the map with minimal conformal distortion) and proposed an iterative procedure 
to find a piecewise linear map $f$ whose Beltrami coefficient has constant magnitude.

\section{PL metrics and triangulations}
The purpose of this section is to explain the relation between
PL metrics and triangulations and to familiarize the readers with a more
general triangulated surface than the one usually encountered
in many engineering fields with the structure of simplicial complex (i.e., 
any higher dimensional simplex is uniquely determined by its vertices). 
For simplicity, we assume the surface $S$ is closed without boundary. 
We will discuss how to deal with the 
surfaces with boundary in Section~\ref{surfaceswithboundary}. 

We start with a triangulated surface $S$ with which we are familiar, i.e., embedded in $\mathbb{R}^3$ 
where each triangle is the convex hull of its three vertices. For example, the boundary of a 
tetrahedron in $\mathbb{R}^3$ is such a triangle mesh with four Euclidean 
triangles, as shown in the leftmost picture in Figure~\ref{fig:triangulations}. 
Denote $V$ the set of vertices. 
Note that other than the vertices, any point $p\in S$ has a flat neighborhood. 
This is obvious if $p$ is in the interior of a triangle which is Euclidean. 
For $p$ in the interior of an edge, one can flatten the 
two triangles incident to the edge into the Euclidean plane, and thus $p$ also has
an (intrinsically) flat neighborhood. For a vertex $v\in V$, it has a neighborhood like a 
cone, as shown in Figure~\ref{fig:triangulations}. Thus the metric on $S$ is flat with possible
cone singularities at a discrete set of vertices. We call such a metric a {\it polyhedral metric} 
or simply  PL metric. In general, a surface with a PL metric is obtained by isometrically 
gluing pairs of edges of a finite collection of Euclidean triangles. See Figure~\ref{fig:triangulations}
for examples. The converse also holds, that is any surface with a PL metric can be 
partitioned into Euclidean triangles. In fact, this can be done by keeping connecting pairs 
of cone singular vertices with intrinsically straight edges on $S$ until no edge can be added
without intersecting the previously added edges in their interiors. Each partition is in fact
a triangulation with Euclidean triangles (in short {\it Euclidean triangulation}) on the surface $S$.
The curvature of a PL metric is $0$ everywhere except at the cone singular vertices where
the curvature is defined as $2\pi$ less the cone angle. Given a Euclidean triangulation
$T=(V, E, F)$ on the surface $S$,  one can evaluate the cone angle at a vertex $v$ 
by summing the inner angles at $v$ in the triangulation $T$, and even deform the PL metric $d$ 
by changing the edge length. For an edge $e\in E$, let $d(e)$ be the length of $e$ measured in the metric 
$d$. The edge length assignment $l:E \rightarrow \mathbb{R}^+$ with $l(e) = d(e)$ uniquely determines
the PL metric $d$.

\begin{figure}[t]
\begin{center}
\begin{tabular}{c}
\includegraphics[width=0.8\textwidth]{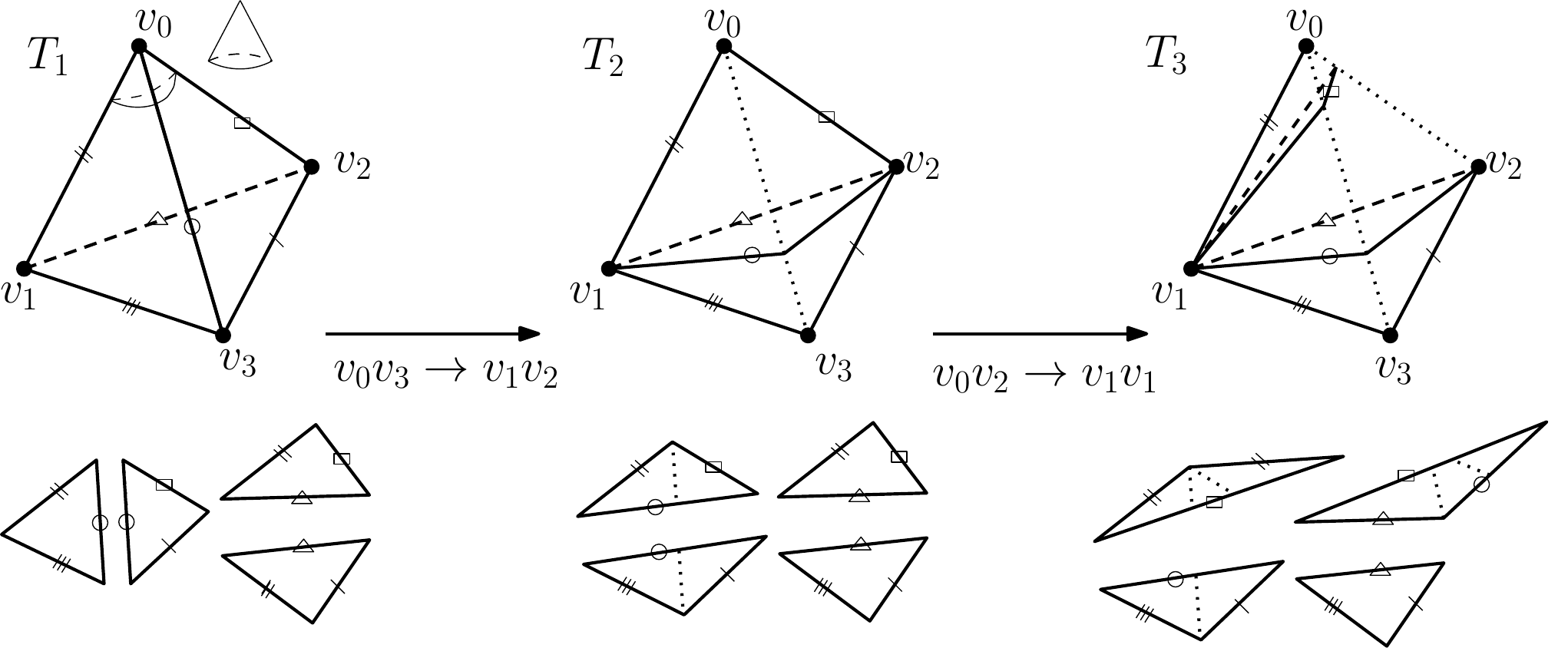}
\end{tabular}
\end{center}
\vspace{-0.1in}
\caption{Triangulations of the boundary of a tetrahedron: The second row shows the gluing pattern of 
the triangles for different triangulations, where the edges marked with the same symbol are glued together. 
The triangulation $T_2$ ($T_3$) is obtained by switching an edge in $T_1$ ($T_2$). 
\label{fig:triangulations}}
\end{figure}

Given a PL metric $d$ on $S$, there may be more than one Euclidean triangulations. 
Figure~\ref{fig:triangulations} shows three different triangulations of the boundary of a tetrahedron, 
where the triangulation $T_2$ respectively $T_3$ is obtained by diagonally switching the edge $v_0v_3$ in
$T_1$ respectively the edge $v_0v_2$ in $T_2$. It is generally true that any two (Euclidean) 
triangulations on $S$ with the same set of vertices $V$ are related by diagonal switches~\cite{Hatcher}. 
Among those Euclidean triangulations, there always exists a Delaunay triangulation where
every edge is Delaunay, that is the sum of the angles opposite to the edge is no bigger than $\pi$~\cite{Bobenko07}. 
There may exist more than one Delaunay triangulations.  If it happens that the sum of the angles 
opposite to an edge is exactly $\pi$, then by (cocircular) diagonally switching that edge, 
we obtain another Delaunay triangulation. In fact, any two Delaunay triangulations are related by 
a sequence of cocircular diagonal switches.

In the paper, we fix the topology of the closed surface $S$ and a finite non-empty set $V\subset S$
and call the pair $(S, V)$ a {\it marked surface}. A PL metric on the pair $(S, V)$ is a PL metric
on $S$ with the cone singularities in $V$, the curvature of a PL metric on $(S, V)$ is the function 
$K: V\rightarrow \mathbb{R}$ sending a vertex $v$ to $2\pi$ less than the cone angle at $v$, and a 
triangulation of the pair $(S, V)$ is a triangulation on $S$ with vertex set $V$. The curvature $K$
of a PL metric satisfies the Gauss-Bonnet formula: $\sum_{v\in V}K(v) = 2\pi\chi(S)$ where $\chi(S)$
is the Euler characteristic number of $S$. If $T = (V, E, F)$ is a triangulation on $(S, V)$, then 
$\chi(S) = |V| -|E| + |F|$. 

Let $T_{pl}(S, V)$ be the space of PL metrics on $(S, V)$ 
\footnote{Strictly speaking, we should consider the set of equivalence classes of PL metrics where two PL metric
$d, d'$ on $(S, V)$ are equivalent if there is an isometry $h: (S, V, d) \rightarrow (S, V, d')$ that is homotopic
to the identity map on $(S, V)$. However this difference is subtle and can be ignored, especially for the purpose
of understanding the algorithm. }. Given a triangulation $T$ of $(S, V)$ with set of edges $E=E(T)$, let $\mathcal{E}(T)$
be the set of edge length assignments so that the triangle inequalities are satisfied for all triangles in $T$. 
$\mathcal{E}(T)$ is a convex polytope in $\mathbb{R}^{|E(T)|}$. Since any edge length assignment $l\in \mathcal{E}(T)$ 
determines a PL metric $d$ on $(S, V)$ with $d(e) = l(e)$, there is an injective map 
\begin{equation}
\Phi_T: \mathcal{E}(T) \rightarrow T_{pl}(S, V)
\end{equation}
sending $l$ to a PL metric $d_l=\Phi_T(l) $ on $(S, V)$. The image $\mathcal{M}(T) := \Phi_T(\mathcal{E}(T))$ is the space
of all PL metrics $d$ on $(S, V)$ for which $T$ is a Euclidean triangulation in $d$. 
From the previous discussion, for any PL metric $d$ on 
$(S, V)$, there exists a Euclidean triangulation $T$ on $(S, V)$ whose edge length assignment is
given by the metric $d$, i.e, there exists an edge length assignment $l\in \mathcal{E}(T)$ with $d = \Phi_T(l)$.  
Thus we have $T_{pl}(S, V) = \cup_T \mathcal{M}(T) $ where the union is over all triangulations on $(S, V)$. Notice that
$E(T) = (-3\chi(S) + 3|V|)$ where $\chi(S)$ is the Euler characteristic number of $S$, which is independent of $T$. 
This means that  $T_{pl}(S, V)$ is a manifold of dimension $(-3\chi(S) + 3|V|)$
with coordinate charts $\{(M(T), \Phi_T^{-1}) | T~\text{is a triangulation on}~(S, V)\}$, as illustrated in 
Figure~\ref{fig:pl_metrics}. Note that in general $T_{pl}(S, V) \neq \mathcal{M}(T)$.

\begin{figure}[t]
\begin{center}
\begin{tabular}{c}
\includegraphics[width=0.8\textwidth]{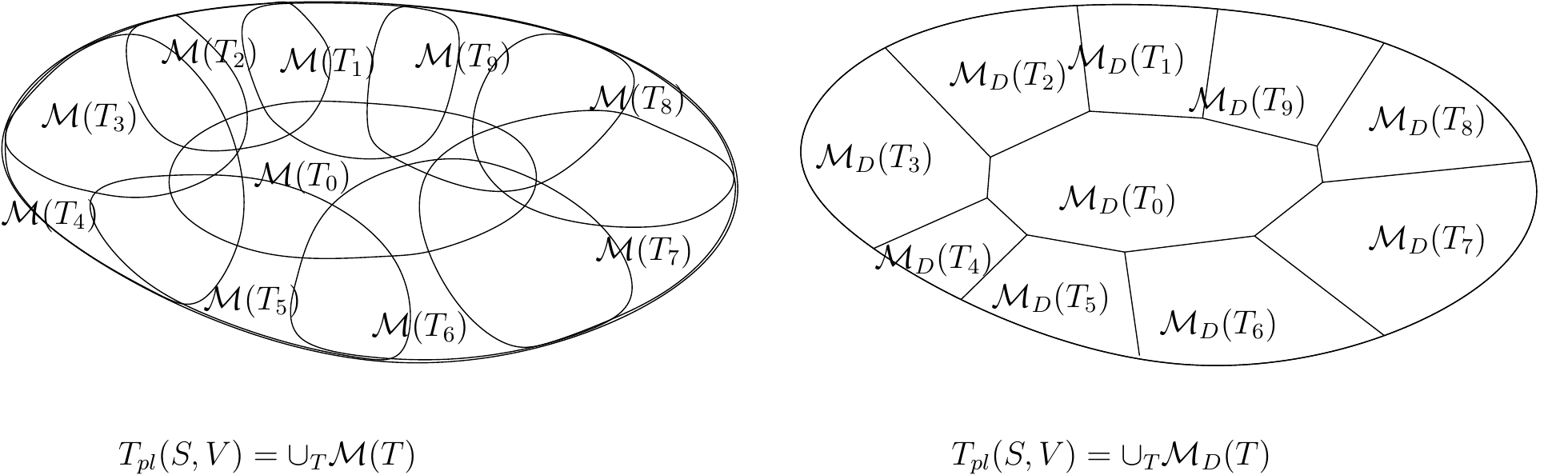}
\end{tabular}
\vspace{-0.1in}
\end{center}
\caption{Coverings of the space of PL metrics $T_{pl}(S, V)$.
\label{fig:pl_metrics}}
\end{figure}

Now we consider a subset of $\mathcal{E}(T)$:
\begin{equation}
\mathcal{E}_D(T) = \{l \in \mathcal{E}(T) | T ~\text{is a Delaunay triangulation on}~(S, V)~\text{in the PL metric}~\Phi_T(l)\}. 
\end{equation}
As we discussed before, for any PL metric $d$, there is a Delaunay triangulation $T$ whose edge length assignment $l$
is given by the metric $d$, i.e., $d = \Phi_T(l)$. 
Thus the set $\{\mathcal{E}_D(T)|  T~\text{is a triangulation on}~(S, V)\}$ also covers $T_{pl}(S, V)$. 
In fact, this set forms a cell decomposition of $T_{pl}(S, V)$~\cite{rivin, glsw1}, as illustrated in Figure~\ref{fig:pl_metrics}.
Thus one may say that Delaunay triangulation is canonical as it is uniquely determined by the PL metric except 
for those metrics on the cell boundary which have multiple Delaunay triangulations.

Finally, we remark that in a Euclidean triangulation $T$ on $(S, V)$, it is possible to have multiple intrinsically 
straight edges between two vertices (e.g., the edges marked with ``$\circ$'' and ``$\triangle$'' between 
$v_1$ and $v_2$ in the triangulation $T_2$ in Figure~\ref{fig:triangulations}), 
and even to have an intrinsically straight loop edge (e.g., the edge marked with ``$\square$'' in the 
triangulation $T_3$ in Figure~\ref{fig:triangulations}). Note, even if we start with a mesh 
with the structure of simplicial complex, we may end up with a mesh
with the more general structure as above as we allow diagonal switches. 

\section{Discrete Conformality}
Now we are ready to present our definition of discrete conformality. 

\begin{definition} \label{dc}(Discrete conformality for surfaces without boundary)
 Two PL metrics $d, {d}'$ on $(S,V)$ are discrete conformal  if there exist sequences
 of PL
metrics $d_0=d, ..., d_m ={d}'$ on $(S, V)$ and triangulations
$\T_0, ..., \T_m$ of $(S, V)$  satisfying

(a) each $\T_i$ is Delaunay in $d_i$,

(b) if $\T_i=T_{i+1}$, then $l_{i+1} = w*_{T_i} l_i$ for a conformal factor $w: V\rightarrow \mathbb{R}$
where $l_{i+1}$ and $l_i$ are the edge length assignments over the edges of $T_i$ with $l_{i+1}(e) = d_{i+1}(e)$ and 
$l_i(e) = d_i(e)$ for any edge $e$ in $T_i$. 

(c) if $\T_i \neq \T_{i+1}$, then $d_{i+1} = d_{i}$\footnote{Strictly speaking, 
$d_{i+1}=d_{i}$ in the sense of equivalence class, that is  $(S, d_i)$ is isometric 
to $(S,d_{i+1})$ by an isometry homotopic to the identity in $(S, V)$.}, and $\T_i, \T_{i+1}$ are related by cocircular diagonal switches.
\end{definition}
This definition means that $d, {d}'$ are discrete conformal if and only if there exists a path 
connecting two PL metrics in the space of $T_{pl}(S, V)$
so that within a cell $\mathcal{E}_D(T)$, the metrics deform along the path by vertex scalings, 
and on the cell boundary, the Delaunay triangulation is changed to another
via cocircular diagonal switches.  

The condition (a) in the definition is critical. Note that the operation of vertex scaling 
depends on the choice of triangulations. The vertex scaling of the same PL metric but under different
Euclidean triangulations may generate different PL metrics. So the previous definition of discrete conformality 
by vertex scaling~\cite{Rocek, luo, bps} heavily depends on the triangulations, which is
not inherent to PL metrics. On the other hand, by restricting to Delaunay triangulations which are canonical 
to PL metrics, our definition of discrete conformality is inherent to PL metrics. With this definition, 
we are able to prove the following uniformization theorem in~\cite{glsw1}. 

\begin{theorem} \label{thm:main} Suppose $(S, V)$ is a closed connected marked surface
and  $d$ is any PL metric on $(S, V)$.  Then for any $K^*:V \to
(-\infty, 2\pi)$ with $\sum_{v \in V} K^*(v) =2\pi \chi(S)$, there
exists a PL metric $d'$, unique up to scaling,  on $(S, V)$ so
that $d'$ is discrete conformal to $d$ and the discrete curvature
of $d'$ is $K^*$. 
\end{theorem}
In the above theorem, the conditions on the curvature $K^*$ are necessary for
$K^*$ to be a curvature of a PL metric on $(S, V)$. The theorem states that those conditions are 
also sufficient for $K^*$ to be achieved by a metric that is discrete conformal to the given
metric $d$. This solves the existence and the uniqueness of the discrete conformal deformation 
mentioned in the introduction.

\subsection{Surfaces with boundary}
\label{surfaceswithboundary}
To deal with a surface with boundary, our strategy is to double the surface to 
remove the boundary, that is to make another
copy of the original surface and glue them along the boundary, and apply 
the discrete conformal deformation described above to the doubled surface, and finally cut out
a copy of the deformed surface from the deformed doubled surface. 

Let $B$ be the boundary of the marked surface $(S, V)$. Given a PL metric $d$ on $S$, $B$ consists
of a set of closed polygonal loops. A Euclidean triangulation on $(S, V)$ is a partition of $S$ into 
Euclidean triangles with the vertices $V$. Note those edges of the polygonal loops of $B$ have to be 
in the triangulation. 
For a vertex $v\in V$ on the boundary, its curvature is defined as $\pi$ less than the cone angle at $v$. 
With this definition, the Gauss-Bonnet theorem still holds: $\sum_{v\in V} K(v) = 2\pi\chi(S)$. 

The doubled surface of $(S, V)$ is defined by taking the disjoint union of two 
copies of $(S, V)$ and identifying the points on the boundary by an homeomorphism 
$f: B\rightarrow B$ which preserves the vertices on the boundary. Denote $(\tilde{S}, \tilde{V})$ 
the doubled surface of $(S, V)$.  A PL metric $d$ on $(S, V)$ induces 
a PL metric $\tilde{d}$ on $(\tilde{S}, \tilde{V})$ by forcing the gluing map $f$ to 
be isometric in $d$. We call $\tilde{d}$ the doubled metric of $d$. 
Conversely, a PL metric on the doubled surface $(\tilde{S}, \tilde{V})$ is said 
to respect the doubling structure if it is the doubled metric of a PL metric on $(S, V)$. 
Let the map $h: (\tilde{S}, \tilde{V}) \rightarrow (\tilde{S}, \tilde{V}) $ be the 
mirror map sending a point to the other copy. The map $h$ is a self-isometric map if 
the PL metric on $(\tilde{S}, \tilde{V})$ respects the doubling structure. 
For convenience, the set of fixed points 
of the map $h$ is called the boundary of $(\tilde{S}, \tilde{V})$.

\begin{definition}
\label{dcwithboundary}(Discrete conformality for surfaces with boundary)
 Two PL metrics $d, {d}'$ on the surface $(S,V)$ with boundary are discrete conformal  if 
 their doubled metrics on the doubled surface of $(S, V)$ are discrete conformal according to the definition~\ref{dc}. 
\end{definition}

\begin{theorem}
\label{theorem:dc_boundary} 
Suppose $(S, V)$ is a connected marked surface with 
boundary and $d$ is any PL metric on $(S, V)$. Then for any $K^*:V \to (-\infty, 2\pi)$ 
with $\sum_{v \in V} K^*(v) =2\pi \chi(S)$ and $K^*(v)<\pi$ for a vertex $v$ on the boundary, 
there exists a PL metric $d'$, unique up to scaling, on the surface $(S, V)$ 
so that $d'$ is discrete conformal to $d$ and the discrete curvature
of $d'$ is the prescribed curvature $K^*$. 
\end{theorem}

The proof of the above theorem is deferred to the appendix. The basic idea is as follows. 
We obtain the doubled surface $(\tilde{S}, \tilde{V})$, and prescribe the curvature $\tilde{K}^*(v)$ for $(\tilde{S}, \tilde{V})$ as 
follows: for a vertex $v$ on the boundary, set $\tilde{K}^*(v) = 2 * K^*(v)$ 
and for a vertex in the interior, set $\tilde{K}^*(v) =  K^*(v)$. 
It is easy to verify that the curvature $\tilde{K}^*$ satisfies the hypotheses imposed in Theorem~\ref{thm:main}
to a target curvature on $ (\tilde{S}, \tilde{V})$. Thus there exists a PL metric $\tilde{d}'$ discrete
conformal to $\tilde{d}$ and the discrete curvature of $\tilde{d}'$ is the curvature of $\tilde{K}^*$.
It remains to show that $\tilde{d}'$ respects the doubling structure and the restriction of $\tilde{d}'$ onto
$S$ is the PL metric $d'$ with the property stated in the theorem. The key is to show that the conformal 
factor $w$ remains respecting the conformal structure, i.e., $w(h(v)) = w(v)$, and the Delaunay triangulation $T$
of $(\tilde{S}, \tilde{V})$ under metric $\tilde{d}'$ has certain symmetric property. Specifically, 
any triangle $f$ crossing an edge $ij$ on the boundary has to
have two vertices $u, u'$ so that $u'=h(u)$, and moreover, if the third vertex $v$ of the 
triangle $f$ is neither $i$ nor $j$, the neighboring triangle $f' = v'vu'$ with $v'=h(v)$ has also to cross
the edge $ij$ and two triangles $f$ and $f'$ form a cocircular quadrilateral, as shown in Figure~\ref{fig:boundary_triangles}. 
In addition, the boundary edge $ij$ remains straight after the discrete conformal deformation, 
and subdivide the crossed cocircular quadrilaterals into two identical pieces. This makes it easy 
to algorithmically cut out a copy of the deformed surface from the deformed doubled surface. 

\begin{figure}[t]
\begin{center}
\begin{tabular}{c}
\includegraphics[width=0.6\textwidth]{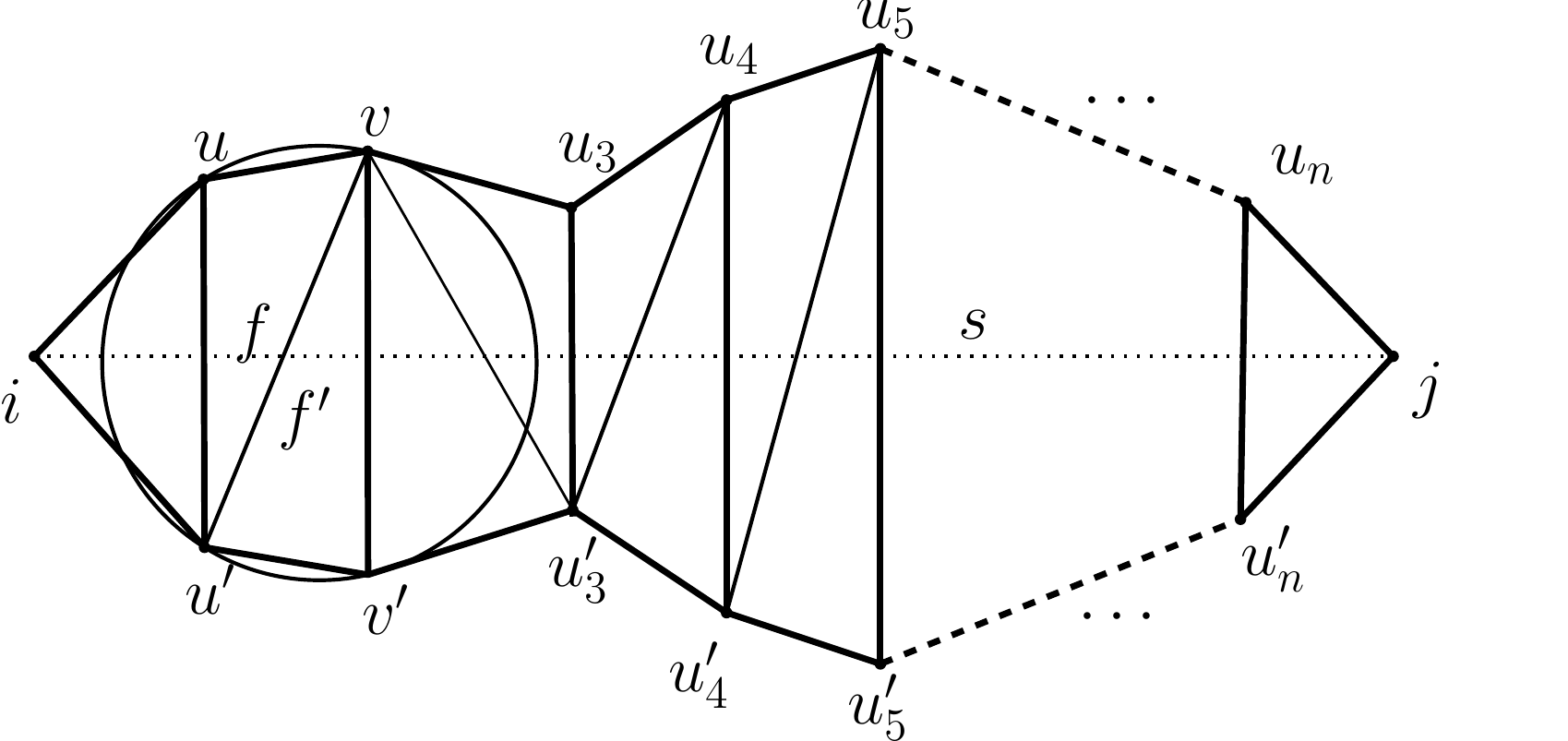}
\end{tabular}
\end{center}
\vspace{-0.1in}
\caption{The symmetric property of the triangles crossing an edge $ij$
on the boundary.
\label{fig:boundary_triangles}}
\end{figure}

\section{Convex Energy}
\label{sec:convexenergy}
In this section, we describe a convex energy for solving the problem of prescribing curvature.  
This will answer the first question raised in the introduction positively, i.e.,  there is
still a convex energy even with the seeming appearance of combinatorial variables for changing 
triangulations. Roughly speaking, as our discrete conformality only involves Delaunay
triangulations, which are canonical and determined by PL metrics, the combinatorial variables
of triangulations are not independent. 

\begin{figure}[t]
\begin{center}
\begin{tabular}{c}
\includegraphics[width=0.75\textwidth]{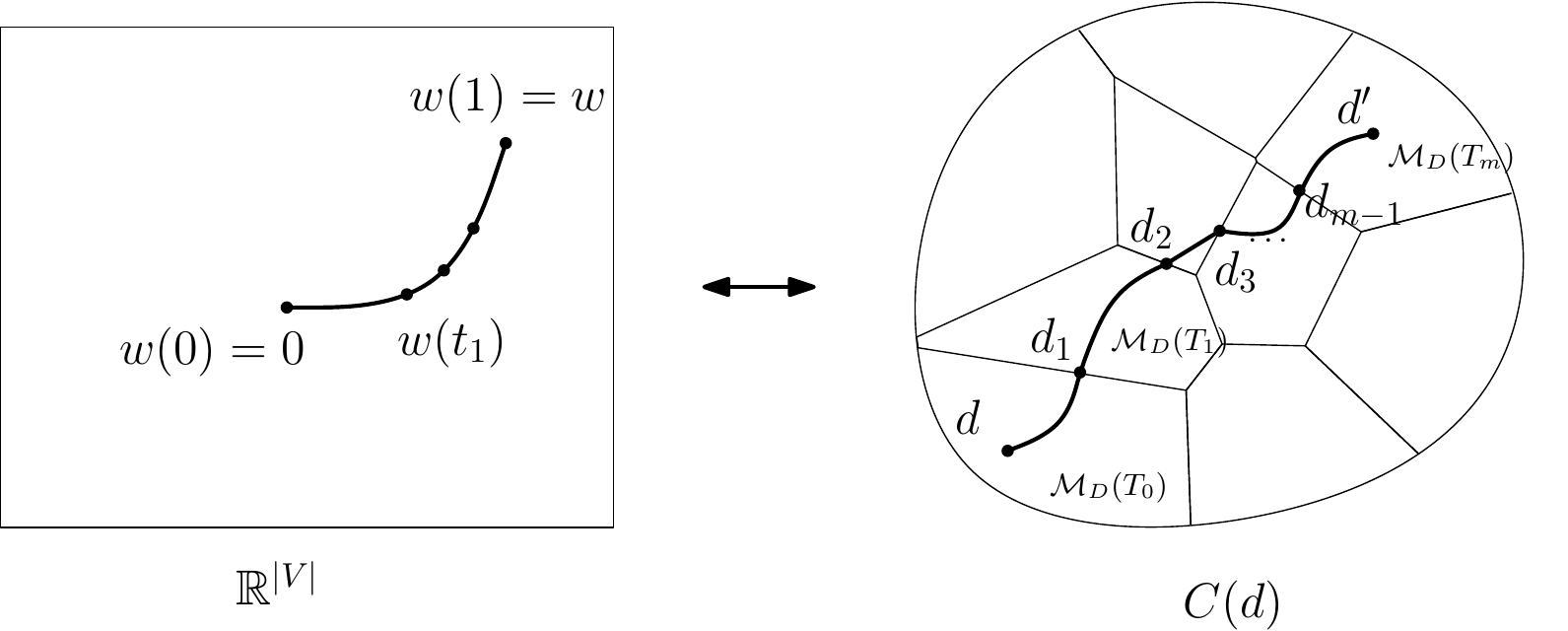}
\end{tabular}
\end{center}
\vspace{-0.1in}
\caption{1-1 correspondence between the conformal factors and the PL metrics on $(S, V)$ discrete 
conformal to a PL metric $d$. 
\label{fig:conformal_class}}
\end{figure}

Given a PL metric $d$ on $(S, V)$, we let $C(d)$ denote the space of PL metrics that
are discrete conformal to $d$. The following lemma about $C(d)$ is important. 
\begin{lemma}
\label{lem:conformalmetricspace}
There is a $C^1$ diffeomorphism from $C(d)$ to $\mathbb{R}^{|V|}$ where a point $w \in \mathbb{R}^{|V|}$ is 
understood as a discrete conformal factor on $V$. 
\end{lemma}
This lemma means there is a one-to-one correspondence between the PL metrics
discrete conformal to $d$ and all discrete conformal factors on $V$.
The energy is defined over $\mathbb{R}^{|V|}$, the space of all discrete conformal factors.  
The rigorous mathematical proof of this lemma uses the Teichm\"uller theory by 
establishing a one-to-one correspondence between PL metrics on $(S, V)$ and the 
hyperbolic metrics on $S\setminus V$ with cusps and 
decorations at $V$~\cite{glsw1}. In this paper, we will not explain this 
connection to hyperbolic metrics. Instead, we will give an intuitive explanation of
the lemma aiming at the understanding of the algorithmic aspects, which although does not 
mathematically prove it. 

The space $C(d)\subset T_{pl}(S, V)$ has a cell 
decomposition induced by that of $T_{pl}(S, V)$, where a cell is the intersection 
$C(d)\cap \mathcal{M}_D(T)$ for some triangulation $T$ on $(S, V)$. 
Note that the number of cells in $C(d)$ is finite~\cite{glsw1}. See Figure~\ref{fig:conformal_class}.
Let $T_0$ be a Delaunay triangulation in the initial PL metric $d$, and $l_{T_0}$ be the
edge length assignment with $l_{T_0}(e)=d(e)$ for any edge $e$ of $T_0$. Given a conformal factor 
$w\in \mathbb{R}^{|V|}$, let $w$ also denote a path in $\mathbb{R}^{|V|}$ from $0$ and $w$, 
that is $w: [0, 1]\rightarrow  \mathbb{R}^{|V|}$ with $w(0) =0$ and $w(1)= w$.
We have $l_{T_0} = w(0)*_{T_0}l_{T_0}$. As we move along the path $w$, we continuously 
deform the PL metric $d$ discrete conformally through vertex scaling $l_{T_0}$ by $w(t)$.
This will trace out a path $\Phi_{T_0}(w(t)*_{T_0} l_{T_0})$ in the cell $C(d)\cap \mathcal{M}_D(T_0)$. 
At some point, this path may hit the boundary of 
the cell. Assume that happened at $t=t_1$ and for example the quadrilateral $f\cup f'$ with the diagonal $e$
becomes cocircular in the metric $d_1 = \Phi_{T_0}(w(t_1)*_{T_0} l_{T_0})$.
We diagonal switch the edge $e$ to the edge $e'$ and obtain another Delaunay 
triangulation $T_1$ in $d_1$, as shown in Figure~\ref{fig:basic-operations}.
Note that $T_0$ is also Delaunay in $d_1$. Due to the well-known Ptolemy identity 
for a  cocircular quadrilateral, we have 
\begin{eqnarray}
d_1(e') &=& \frac{d_1(e_1)d_1(e'_1) + d_1(e_2)d_1(e'_2)}{d_1(e)} \nonumber \\
        &=& \frac{l_{T_0}(e_1)l_{T_0}(e'_1) + l_{T_0}(e_2)l_{T_0}(e'_2)}{l_{T_0}(e)}e^{w(t_1)(u') + w(t_1)(v')}
\end{eqnarray}
where $u'$ and $v'$ are the endpoints of $e'$.
If let $l_{T_1}$ be the edge length assignment over the 
edges of $T_1$ so that $l_{T_1}(e) = l_{T_0}(e)$ for $e\neq e'$ and 
\begin{equation}
l_{T_1}(e') = \frac{l_{T_0}(e_1)l_{T_0}(e'_1) + l_{T_0}(e_2)l_{T_0}(e'_2)}{l_{T_0}(e)}, 
\end{equation}
then we have $d_1 =\Phi_{T_1}(w(t_1)*_{T_1} l_{T_1})$. Note that $l_{T_1}(e')$ for the new edge $e'$ 
depends only on $l_{T_0}$, in particular is independent of the conformal factor $w(t_1)$.  
This answers the second question raised in the introduction on how to assign the initial edge lengths
for the new edges after diagonal switches. 
Repeat the above procedure as we continuously move along the path $w$. At the end, 
we reach a metric $d'= \Phi_{T_m}(w *_{T_m} l_{T_m})$ in the cell  $C(d)\cap \mathcal{M}_D(T_m)$ 
for some triangulation $T_m$. Mathematically, we can show that the final metric $d'$ is independent 
of the choice of path, namely if we choose another path connecting $0$ and $w$ and repeat the above
procedure, we reach the same metric $d'$. Thus $d'$ depends only on the initial PL metric $d$ and the 
conformal factor $w$. We write $d'= w*d$.  Conversely, for any PL metric $d'\in C(d)$, one can find a conformal
factor $w\in \mathbb{R}^{|V|}$ so that $d'= w*d$. To see this, from the definition of discrete conformality, 
there is a path in $C(d)$ connecting $d$ and $d'$. From the above procedure, it is easy to trace out 
a path $w: [0, 1] \rightarrow  \mathbb{R}^{|V|}$ starting at $0$ so that $w(t)*d$ is the path 
in  in $C(d)$ connecting $d$ and $d'$. This shows that there is 
a one-to-one correspondence between $\mathbb{R}^{|V|}$ and $C(d)$. This in fact answers the third question 
raised in the introduction positively, i.e., different orders of switching diagonals leading to 
the same final PL metric.

We follow Luo~\cite{luo} and define the energy as a path integral of a differential 1-form on $\mathbb{R}^{|V|}$. 
Given a PL metric $d$ on $(S, V)$, let $K: \mathbb{R}^{|V|} \rightarrow \mathbb{R}^{|V|}$ be the curvature map
so that $K(w)$ is the curvature of the PL metric $w*d$ on $(S, V)$ for any conformal factor $w\in \mathbb{R}^{|V|}$. 
Label the vertices $V$ using $1, 2, \cdots, n=|V|$. 
Let $K_i$ and $w_i$ denote the curvature $K$ and the conformal factor $w$ evaluated at the vertex $i$, respectively.
Given a Euclidean triangulation $T=(V, E, F)$ of $(S, V)$, associate each edge $ij \in E$ into two oriented 
half edges, one from $i$ to $j$ and the other from $j$ to $i$. 
Let $E_{ij}(T)$ be the set of oriented edges in $T$ starting 
at the vertex $i$ and pointing to the vertex $j$. Note that $E_{ii}(T)$ may not be empty. 
Let $E_i(T)$ be the set of oriented edges in $T$ starting from the vertex $i$, i.e.,  
$E_i(T) = \cup _{j \sim i} E_{ij}(T)$. 
For an edge $e$ shared by the triangles $f$ and $f'$, let $\alpha_e$ and $\alpha'_e$ be the angles opposite 
to $e$ in $f$ and $f'$ respectively. We have the following lemma on the curvature $K$.  
\begin{lemma}
\begin{itemize}
\item[(i)] $K_i$ is a $C^1$ function on $\mathbb{R}^{|V|}$ for any vertex $i$. 
\item[(ii)] Let $T$ be a Delaunay triangulation in the metric $w*d$ and then 
\begin{eqnarray}
\frac{\partial K_i}{\partial w_j} = \left\{
\begin{array}{rl}
& -\sum_{e\in E_{ij}(T)}(\cot\alpha_e + \cot \alpha'_e)  \text{~~~if $i\neq j$}\\
& \sum_{e\in E_i(T)} (\cot\alpha_e + \cot \alpha'_e) - \sum_{e\in E_{ii}(T)} (\cot\alpha_e + \cot \alpha'_e)  \text{~~~if  $i = j$ }
\end{array} \right.
\label{eqn:pKpw}
\end{eqnarray}
\item[(iii)] The matrix $(\frac{\partial K_i}{\partial w_j})_{i, j}$ is semi-positive definite and its null space
 only consists of constant vectors. 
\end{itemize}
\label{lem:pKpw}
\end{lemma}
\begin{proof}
For a $w \in \mathbb{R}^{|V|}$ so that $w*d$ is in the interior of a cell of $C(d)$, the above lemma was proved 
by Luo~\cite{luo}. In fact, in our setting, due to that the triangulation $T$ is Delaunay, we have  
$\alpha_e + \alpha'_e\leq \pi$ and thus $\cot\alpha_e + \cot \alpha_e'\geq 0$ for any edge $e$, which means the 
matrix $(\frac{\partial K_i}{\partial w_j})_{i, j}$ is diagonally dominant. 
So it remains to show that $K_i$ is $C^1$ on the cell boundaries.

Assume $T'$ is another Delaunay triangulation in the metric $w*d$. Since $T$ and $T'$ are related by a sequence of 
cocircular diagonal switches, we may assume $T'$ is obtained from $T$ by one cocircular diagonal switch. 
Assume the diagonal $e$ of the quadrilateral $f\cup f'$ is switched to the other diagonal $e'$, as shown in 
the right picture of Figure~\ref{fig:basic-operations}. 
For any vertex $i$, $K_i$ obviously remains the same before and after the diagonal switch. From the equation~\eqref{eqn:pKpw}, 
the evaluation of $\frac{\partial K_i}{\partial w_j}$ only involves the quantity $\cot\alpha_e + \cot \alpha'_e$ 
associated to any edge $e$. Observe that only for the sides and the diagonals of the quadrilateral $f\cup f'$, 
this quantity may differ before and after the diagonal switch. For the diagonal $e$, this quantity is $0$ in $T$ 
due to that  $\alpha_e + \alpha'_e = \pi$, and remains $0$ in $T'$ as $e$ is not an edge in $T'$, and 
similarly for another diagonal $e'$. For any side, say $e_1$ (see Figure~\ref{fig:basic-operations}), 
as the angle opposite to $e_1$ in the triangle $f$ equals the angle opposite to $e_1$ in the triangle $g'$, 
this quantity associated to $e_1$ remains the same before and 
after the diagonal switch. This shows that $K_i$ is $C^1$ for any vertex $i$. 
\end{proof}

Define a differential 1-form on the space of conformal factors as 
$\Omega(w) = \sum_{i=1}^{n} K_i(w)dw_i$ for any $w \in \mathbb{R}^{n}$
From Equation~\eqref{eqn:pKpw}, $\frac{\partial K_i}{\partial w_j} = \frac{\partial K_j}{\partial w_i}$ for 
any $i, j$, implying that $\Omega$ is closed and thus exact as the domain  $\mathbb{R}^{n}$ is simply connected. 
This means the path integral of $\Omega$ only depends on the endpoints of the path. Given a prescribed curvature 
$K^* \in \mathbb{R}^{n}$, define the energy $E$ over the space of discrete conformal factors as 
\begin{equation}
E(w) = \int_0^w\sum_{i=1}^{n} K_i(w)dw_i - \sum_{i=1}^{n} K^*_iw_i.
\end{equation}
Note that the gradient of the energy $\nabla E = (K_1 - K^*_1, \cdots, K_{n}- K^*_{n})^t$ and 
the Hessian of the energy $H(E) = (\frac{\partial K_i}{\partial w_j})_{i, j}$. From Lemma~\ref{lem:pKpw}, the Hessian 
$H(E)$ is semi-positive definite and thus the energy $E$ is convex and strictly convex restricted to the subspace 
$W = \{w\in \mathbb{R}^{n} | w_1+\cdots+w_{n} = 0\}$. If the prescribed curvature $K^*$ satisfies the conditions
stated in Theorem~\ref{thm:main}, there exists a discrete conformal factor $w^*\in W$ so that $K^* = K(w^*)$. 
This means $\nabla E(w^*)=0$, implying that $w^*$ is the unique minimum of the energy $E$ on the subspace $W$.
Thus, one can employ the Newton's method to find $w^*$ and hence the PL metric $w^* *d$ which realizes the 
prescribed curvature. 

\section{Discrete Conformal Map}
\label{sec:disconf-map}
In this section, we construct a map $\phi: (S, V, d) \rightarrow (S, V, d')$ on the same marked 
surface $(S, V)$ but with two PL metrics $d$ and $d'$ discrete conformal to each other, which
we call the discrete conformal map from $d$ to $d'$. 
In~\cite{glsw2}, given a PL metric $d$ on $(S, V)$,  we equip $(S, V)$ with another hyperbolic metric 
with cusps (but no decorations) at $V$, denoted $h(d)$. We show that $d'$ and $d$ are discrete conformal 
to each other if and only if $h(d)$ and $h(d')$ are isometric to each other by an isometry homotopy to the identity. 
The discrete conformal map $\phi$ from $d$ to $d'$ is defined as that isometry from $h(d)$ to $h(d')$. 
In this paper, instead of establishing the connection to the hyperbolic metric, we give a more constructive 
description of the discrete conformal map for the purpose of better understanding the algorithm of 
explicitly constructing the map. To make it concrete, assume the triangulations
$T=(V, E, F)$ and $T'=(V, E', F')$ are Delaunay under $d$ and $d'$, respectively. Think of the 
surface $(S, V, d)$ as the disjoint union of the Euclidean triangles in $F$  with pairs of 
edges identified by isometries, and similarly for the surface $(S, V, d')$. 
Note that the map $\phi$ restricted to $V$ is the identity map on $V$ and the task is to extend 
the map to the interiors of the edges in $E$ and the interiors of the triangles in $F$. 

\begin{figure}[t]
\begin{center}
\begin{tabular}{c}
\includegraphics[width=0.5\textwidth]{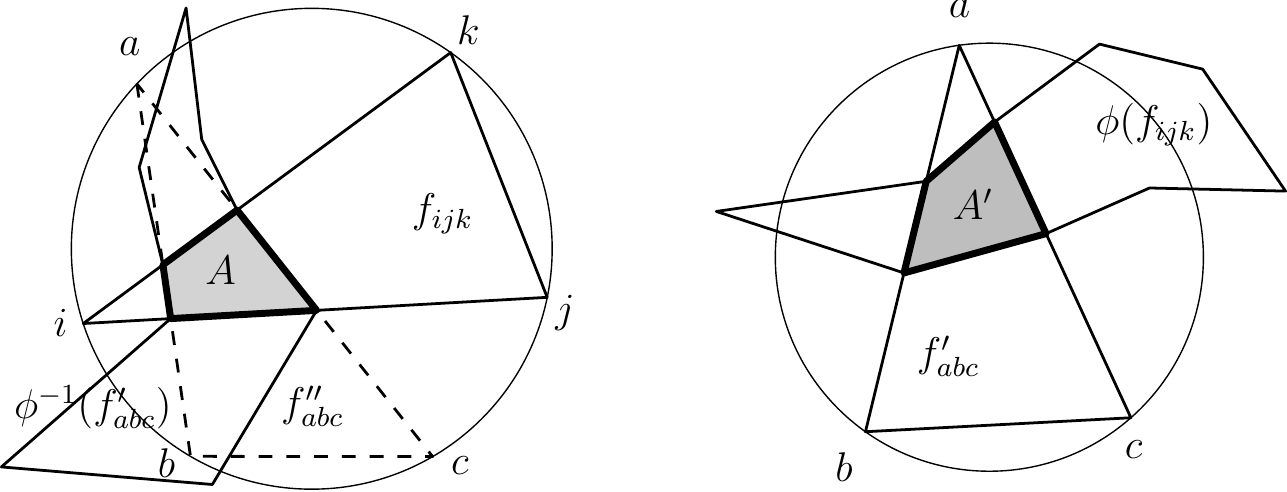}
\end{tabular}
\vspace{-0.1in}
\end{center}
\caption{The mapping triangle of a polygonal facet of $T\cup T'$. 
\label{fig:mapping_piece}}
\end{figure}

Let $w\in \mathbb{R}^{n} $ be the conformal factor so that $d' = w*d$. 
First, we consider the special case where there is a triangulation $T$ which 
is Delaunay in both $d$ and $d'$, i.e.,  $d, d'$ are in the same cell $\mathcal{M}_D(T)$.
In this case, the discrete conformal map is the so-called piecewise circumcircle
preserving projective map introduced by Bobenko et al.~\cite{bps}. 
Let the triangles $f_{ijk}$ and $f'_{ijk}$ be the same triangle in $T$ with 
the vertices $i, j, k$ and the edge lengths measured in $d$ and $d'$ respectively. The 
the map
$\phi|_{f_{ijk}}:
f_{ijk} \rightarrow f'_{ijk}$ is defined in terms of the barycentric coordinates as
\begin{equation}
\phi|_{f_{ijk}}(u_i, u_j, u_k) = (u_ie^{-2w_i}/z, u_je^{-2w_j}/z, u_ke^{-2w_k}/z)
\label{eqn:cppm}
\end{equation}
where $z = u_ie^{-2w_i}+u_je^{-2w_j} +u_ke^{-2w_k}$ is the normalizing factor. 
It is shown in~\cite{bps} the map $\phi|_{f_{ijk}}$ is a projective map from $f_{ijk}$
onto $f'_{ijk}$ which also maps the circumcircle of $f_{ijk}$ to the circumcircle
of $f'_{ijk}$. For two triangles $f_{ijk}$ and $f_{jil}$ sharing the 
edge $e_{ij}$, the maps $\phi|_{f_{ijk}}$ and  $\phi|_{f_{jil}}$ coincide on the 
common edge $e_{ij}$. Thus, we can glue the maps on individual triangles together
to form a globally continuous map, which by definition is the discrete conformal map
$\phi: (S, V, d)\rightarrow (S, V, d')$. Note that the straight line remains 
straight within a triangle under the map $\phi$ as any projective map preserves straight lines. 

Next, we consider the general case where $d$ and $d'$ may not be in the same cell in $C(d)$.
Consider a path $\gamma: [0, 1] \rightarrow C(d)$ with $\gamma(0) = d$ to $\gamma(1) = d'$. 
Let $d_1, d_2, \cdots, d_{m-1}$ be the intersections of $\gamma$  with the boundaries of the cells
in $C(d)$ listed in the increasing order of their path parameter. See Figure~\ref{fig:conformal_class}
for an illustration. For convenience, let
$d_0 = d$ and $d_{m} = d'$. For any $i=0, 1, \cdots, m-1$,  $d_{i}$ and $d_{i + 1}$ are in the 
same cell $\mathcal{M}_D(T_i)$ for some triangulation $T_i$. 
Let $\phi_i: (S, V, d_i)\rightarrow (S, V, d_{i+1})$ be the discrete conformal map from $d_i$ to
$d_{i + 1}$ defined in the above special case. Then the discrete conformal map from $d$ to $d'$ 
by definition is the compositions of the above maps $\phi = \phi_0\circ\phi_1\circ\cdots\circ\phi_m$.

We now state some properties of the discrete conformal map. For their proofs, the interested readers
are referred to the paper~\cite{glsw2}. The most important property is that the map $\phi$ is 
independent of the choice of the path $\gamma$. Namely, if we choose another path $\gamma'$, 
we may end up with a different set of maps $\phi'_0, \phi'_1, \cdots, \phi'_k$ but their composition 
gives the same map $\phi$. Therefore, the map $\phi$ is indeed a well-defined map from $d$ to $d'$. 
The second property is that a straight line on $(S, V, d)$ 
remains straight within a triangle in $T'$ under the map $\phi$ and similarly for a straight line 
on $(S, V, d')$ under the inverse of $\phi$. Another important property is
that $\phi$ remains a piecewise circumcircle preserving projective map but on the smaller pieces. 
Specifically, for two triangles $f_{ijk} \in T$ and $f'_{abc}\in T'$, let $A = f_{ijk} \cap \phi^{-1}(f'_{abc})$ 
and $A' = f'_{abc} \cap \phi(f_{ijk})$. If $A \neq \emptyset$, then $\phi(A) = A'$ and 
$\phi|_A: A\rightarrow A'$ is the restriction onto $A$ of the circumcircle preserving projective
map from a triangle $f''_{abc}$ to the triangle $f'_{abc}$. The triangle $f''_{abc}$ is constructed as 
follows.  The preimage of the edges of $f'_{abc}$ inside $f_{ijk}$ are straight segments. 
See Figure~\ref{fig:mapping_piece} for an illustration. 
Extend them linearly to intersect the circumcircle of the triangle $f_{ijk}$. One can show that there are always 
exactly three intersection points. If we labeled the intersections according to the labels of the endpoints of 
the edges in $f'_{abc}$, this constructs the triangle $f''_{abc}$, which we call the {\it mapping triangle} of $A$. 
Let $d'(st)$ and $d''(st)$ the lengths of the edge $st$ in $f'_{abc}$ and $f''_{abc}$ respectively for 
any $\{s, t\} \subset \{a, b, c\}$.  Calculate $w'_a = \left(\frac{d'(ab)d'(ac)}{d'(bc)}\right)^{1/2}$ and similarly 
for $w'_b, w'_c$. Then by replacing $w$ by $w'$ in \eqref{eqn:cppm}, we construct the circumcircle preserving
projective map from $f''_{abc}$ to $f'_{abc}$. 


For the surface $(S, V)$ with boundary, one can verify that the straight line which cuts the region 
$\cup_{f\in F_s(T_i)} f$ into two identical subregion (see Figure~\ref{fig:boundary_triangles}) is 
the image of the segment $s$ under the discrete conformal map $\tilde{\phi}$. 
Therefore, if let $\tilde{\phi}$ denote the discrete conformal map on the doubled surface, then the restriction 
of $\tilde{\phi}$ onto a copy of $S$ is a map from $(S, V, d)$ to $(S, V, d')$, which we define as the discrete 
conformal map $\phi$ from $(S, V, d)$ to $(S, V, d')$. 

\begin{figure}[!t]
\begin{center}
\begin{tabular}{c}
\includegraphics[width=0.75\textwidth]{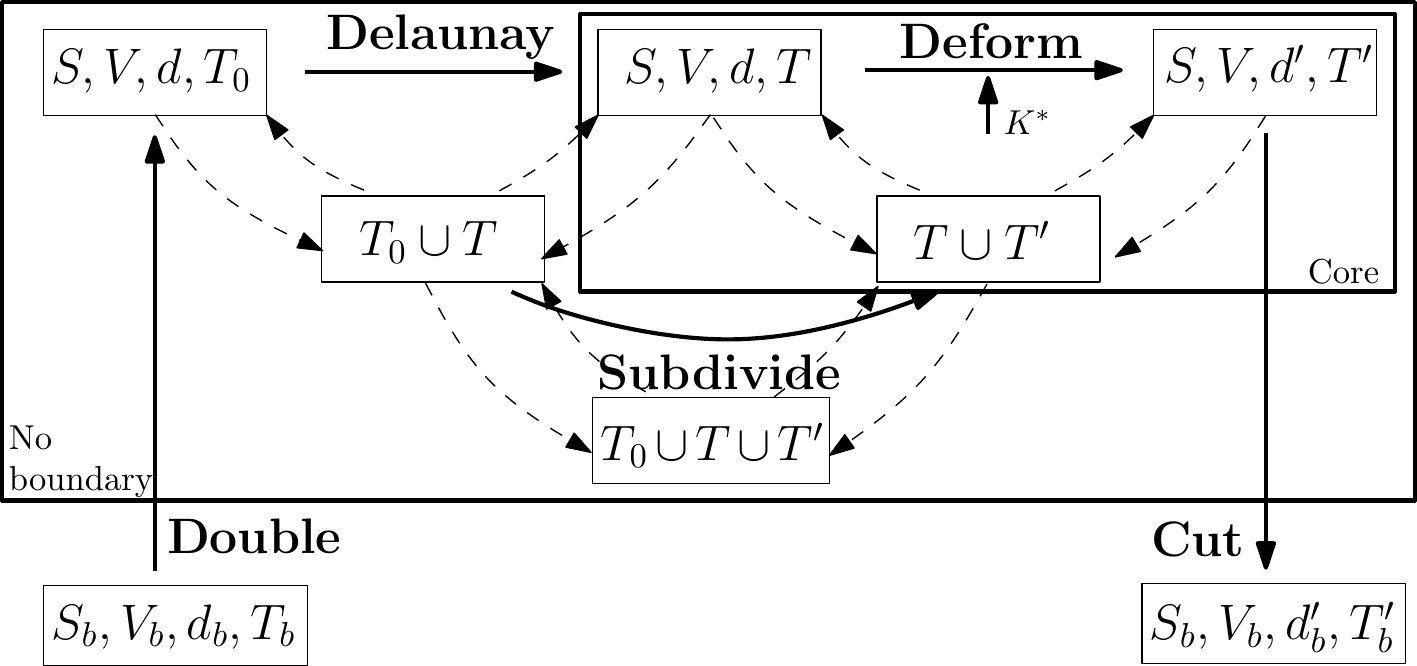}
\end{tabular}
\end{center}
\vspace{0.1in}
\caption{The main objects and the procedures of the algorithm for solving 
the problem of prescribing curvature. 
\label{fig:algorithm_overview}}
\end{figure}

\section{Algorithm}
\label{sec:algorithm}
We have presented the main ideas of the algorithm for solving the problem 
of prescribing curvature. In this section, we give more details of the algorithm at 
the implementation level.

The main objects and the procedures used in the algorithm are shown in Figure~\ref{fig:algorithm_overview}. 
In the problem of prescribing curvature, assume we are given a closed surface $S$ with 
a Euclidean triangulation $T=(V, E, F)$ which is Delaunay, and a desired curvature
$K^*: V \rightarrow \mathbb{R}$. Note that the initial PL metric $d$ on $(S, V)$
is determined by the edge lengths of the Euclidean triangles in $T$. The goal of the algorithm is: 
(1) to find a triangulation $T'=(V, E', F')$ on $S$ and an edge length assignment $l'$ 
over the edges in $E'$ so that the PL metric $d'$ on $(S, V)$ determined by $l'$
is discrete conformal to $d$ and the curvature of $d'$ equals $K^*$, 
and (2) to construct the discrete conformal map $\phi$ from $d$ to $d'$. This is the 
core of the algorithm, which is performed by the procedure ``Deform'' using the Newton's method
described in Section~\ref{sec:convexenergy}. In many applications, 
the given Euclidean triangulation $T_0$ on $S$ may not be Delaunay. The procedure 
``Delaunay'' is employed to convert $T_0$ to a Delaunay triangulation $T$ under the same PL metric
by diagonal switches, as described in~\cite{Fisher:2006}. 
Moreover, if the initial surface has boundaries, the procedure ``Double'' is to double the surface 
to remove the boundary, and the procedure ``Cut'' is to cut out a copy of the deformed surface 
from the deformed doubled surface, as described in Section~\ref{surfaceswithboundary}. 
Both the procedures ``Double'' and ``Cut'' are straightforward to implement.

The procedure ``Deform'' deforms the metric and also changes the triangulations as shown 
in Algorithm~\ref{alg:deform}, whose implementation is more involved than the other procedures. 
Note that the discrete conformal map $\phi$ is a piecewise circumcircle preserving projective map 
on the pieces of the common refinement of the triangulations $T$ and $T'$, denoted $T\cup T'$. 
Topologically, the refinement $T\cup T'$
is a also polyhedral surface whose vertices consists of the
vertices $V$ and the intersections of the edges in $T$
with the edges in $T'$. To see the geometry of $T\cup T'$,
consider the discrete conformal map $\phi: (S, V, d)\rightarrow (S, V, d')$.
According to the theory of discrete conformal mapping described in~\cite{glsw2},
an edge in $T'$ is pulled back to $(S, V, d)$ and geometrically
becomes a polygonal line which is straight inside a triangle of $T$. 
Similarly, an edge in $T$ is pushed forward
to $(S, V, d')$ and is straight within a triangle of $T'$. Therefore,
each edge of $T\cup T'$ is geometrically straight on both $(S, V, d)$ and $(S, V, d')$.
See Figure~\ref{fig:refinement} for an illustration. For instance, the triangle $v_1v_2v_{10}$ in
$T$ is subdivided into three polygonal facets in $T\cup T'$ and similarly
for the triangle $v_9v_0v_{10}$ in $T'$.
All of the involved polyhedral surfaces are represented by halfedge data structures.  
A mechanism is built for these polyhedral surfaces to communicate with each other as follows:
each edge in $T$ or $T'$ has the access to its first sub-edge in $T\cup T'$,
and each edge in $T\cup T'$ has the access to the edge in $T$ and/or $T'$
to which it belongs. In the example shown in Figure~\ref{fig:refinement}, for instance,
each halfedge of the edge $v_7v_9$ in $T$ has a pointer pointing to its first
sub-halfedge in $T\cup T'$ and similarly for the halfedges of $v_6v_8$ in $T'$.
At the same time, each sub-halfedge of $v_7v_9$ ($v_6v_8$) in $T\cup T'$ is equipped
with a pointer pointing back to the corresponding halfedge of $v_7v_9$ in $T$
($v_6v_8$ in $T'$).

\begin{figure}[!t]
\centering
\includegraphics[width = 0.6\textwidth]{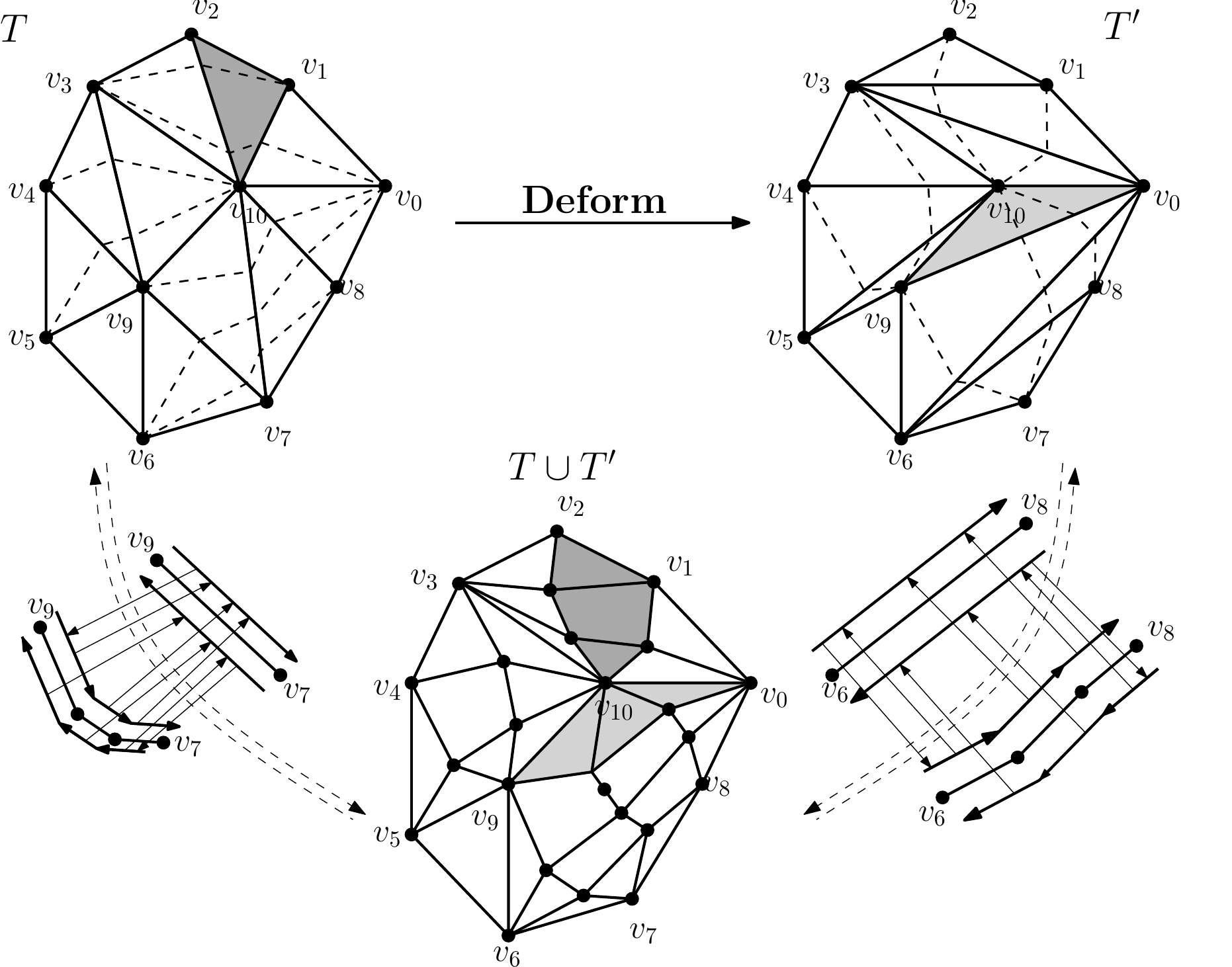}
\vspace{0.1in}
\caption{Refinement: The dotted lines in $T$  are the images
of the edges in $T'$ under the discrete conformal map from $d$ to $d'$ ,
which are straight within a triangle in $T$, and similarly for the dotted lines in $T'$.
}
\label{fig:refinement}
\end{figure}

The final PL metric is determined by the edge length assignment $w *_{T'} l_{T'}$ over the edges in $T'$.
For a vertex $v$ of $T\cup T'$ which is the intersection of an edge $e$ in $T$ and 
an edge $e'$ in $T'$ in the interior, we store both its positions on $e$ and $e'$. In this way, 
we can visualize $\phi^{-1}(e)$ for any edge $e$ in $T'$ on the input surface 
$(S, V, d)$ and $\phi(e)$ for any edge $e$ in $T$ on the deformed surface $(S, V, d')$. 
For a polygon $A$ in $T\cup T'$, we store the edge lengths of its mapping triangle
for the purpose of constructing the discrete conformal map $\phi$.  
In each iteration in the Newton's method, the conformal factor $w$ is updated to $w-\Delta w$, 
which may change the triangulation $T'$, the refinement $T\cup T'$ and $l_{T'}$ as well. 
The sub-procedure ``MoveTo'' shown in Algorithm~\ref{alg:moveto} presents more details on how to 
update the conformal factor and the combinatorial structures of $T'$ and $T\cup T'$.

Let $l_e = l_{T'}(e)$ for an 
edge $e$ in $T'$ and $x(u) =e^{-2w(u)}$. Consider an edge $e$
as shown in Figure~\ref{fig:basic-operations}.  That $e$ is Delaunay is by cosine law
equivalent to 
\begin{eqnarray}
\frac{l_{e_1}l_{e'_1}+l_{e_2}l_{e'_2}}{l_{e_1}l_{e'_2}}x(v)+ \frac{l_{e_1}l_{e'_1}+l_{e_2}l_{e'_2}}{l_{e_2}l_{e'_1}}
x(u)- \frac{l_{e}^2}{l_{e_1}l_{e_2}}x(u')- \frac{l_{e}^2}{l_{e'_1}l_{e'_2}}x(v') \geq 0, 
\end{eqnarray}
which is a linear constraint in the variables $x$. Thus if we change the variables from $w$ to $x$, the 
cell $C(d) \cap \mathcal{M}_D(T')$ becomes a convex polytope. We choose a path from $w$ to $w+\Delta w$
so that it is a line segment in the variables $x$. This makes it easy to detect which edge to switch first
as it amounts to compute the intersections of the line segment with the hyperplanes defined by the 
linear constraints.  

\begin{algorithm}[!h]
\floatname{algorithm}{Algorithm}
\caption{Deform($T=(V, E, F)$, $l_T$, $K^*$ and $\epsilon$)}
\label{alg:deform}
\begin{algorithmic}[1]
\STATE Initialize $T\cup T' = T' = T$
\STATE Set $w=0$;
\STATE Evaluate $K(w)$ and set $\nabla E = K(w) - K^*$
\WHILE {$\|\nabla E\|>\epsilon$} 
\STATE Evaluate $H(E)$ and set $\Delta w = H(E)^{-1} \nabla E$ 
\STATE MoveTo($w$, $w-\Delta w$, $T'$, $T\cup T'$, $l_{T'}$)
\STATE $w \leftarrow w-\Delta w$
\ENDWHILE
\STATE Output $T'$, $T\cup T'$, $l_{T'}$.
\end{algorithmic}
\end{algorithm}

\begin{algorithm*}[!h]
\floatname{algorithm}{Algorithm}
\caption{MoveTo($w_1$, $w_2$, $T'$, $T\cup T'$, $l_{T'}$)}
\label{alg:moveto}
\begin{algorithmic}[1]
\STATE Assume $w(t)$ with $t \in [0, 1]$ be a path satisfying $e^{-2w(t)} = (1-t)e^{-2w_1}+ te^{-2w_2}$. 
\STATE Let the edge $e$ in $T'$ be the first edge that fails the Delaunay condition along the path $w(t)$. 
\IF  {$e$ exists}
\STATE Assume $e$ fails to be Delaunay at $w(t_1)$
\STATE Switch the edge $e$, and update $l_{T'}$
\STATE Update $T\cup T'$: (i) for the newly generated polygons, compute the edge lengths of their mapping triangles, 
and (ii) for the vertices of $T\cup T'$ which are not the vertices of $T'$, 
compute their new positions on the edges of $T'$ under the edge length assignment $w(t_1)*_{T'} l_{T'}$. 
\STATE MoveTo($w(t_1)$, $w_2$, $T'$, $T\cup T'$, $l_{T'}$)
\ELSE
\STATE For the vertices of $T\cup T'$ which are not the vertices of $T'$, compute their new positions 
on the edges in $T'$ with the edge length assignment $w_2*_{T'} l_{T'}$.
\ENDIF
\STATE Output $T'$, $T\cup T'$, $l_{T'}$, and $w$.
\end{algorithmic}
\end{algorithm*}

\begin{figure}[!t]
\begin{center}
\begin{tabular}{c}
\includegraphics[width=0.9\textwidth]{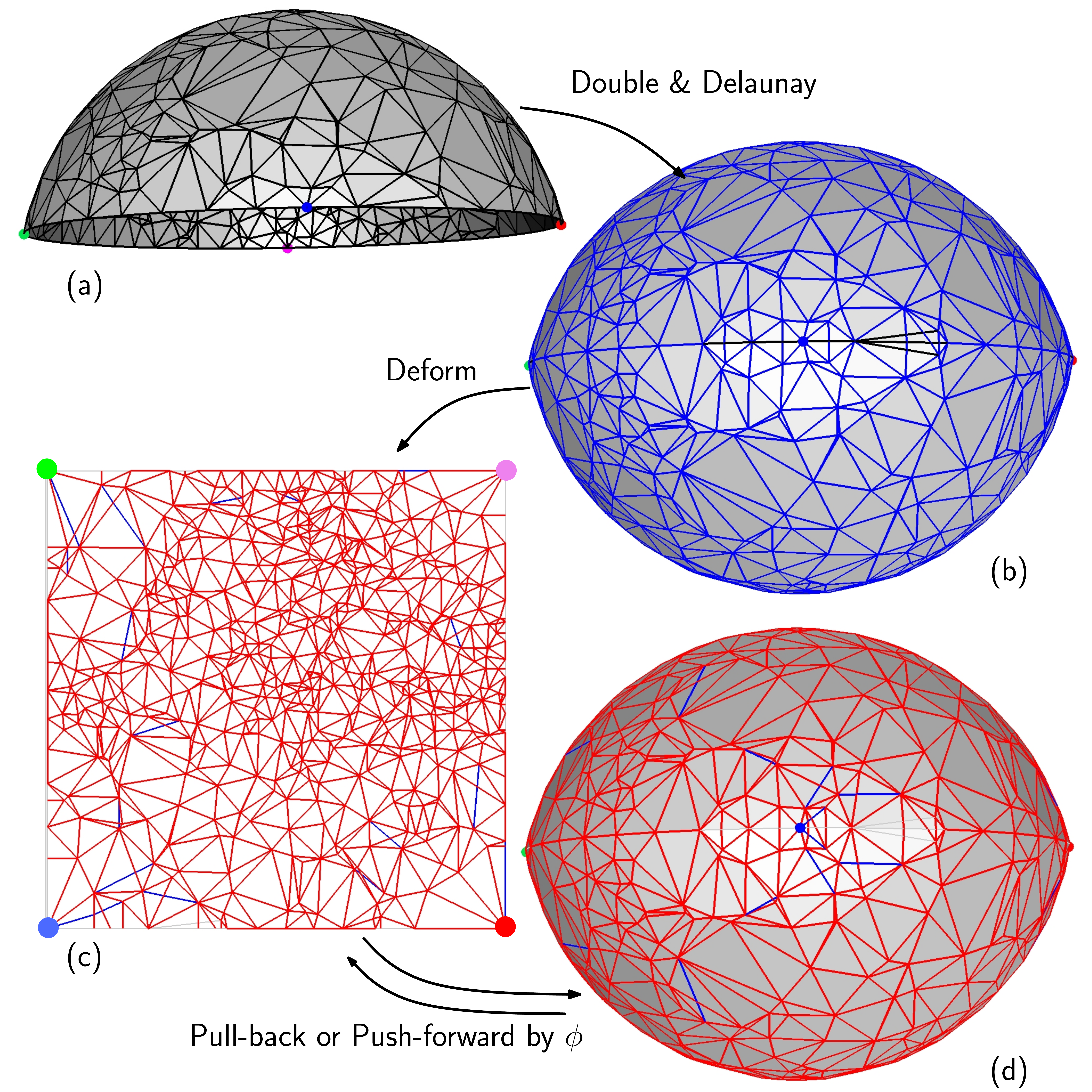}
\end{tabular}
\end{center}
\vspace{0.1in}
\caption{(a): The polyhedral surface $(S_b, V_b, d_b, T_b)$ of a spherical cap.  
(b): The doubled polyhedral surface. The Delaunay 
triangulation $T$ consists of the triangles with blue edges. 
The black edges are non Delaunay edges in the triangulation $T_0$. 
(c): Half of the doubled polyhedral surface after discrete conformal deformation. 
The Delaunay triangulation $T'$ consists of the triangles with red edges. 
The blue edges are the push-forward of the switched edges in $T$
under the discrete conformal map $\phi$. 
Note that the blue edges may not be straight. 
(d): The triangles with red edges are the pull-back of the edges in $T'$ 
under the map $\phi$. Note that the red edges may not be straight. 
\label{fig:algorithm_illustration}}
\end{figure}

Finally, the purpose of constructing the refinements $T_0\cup T$ and $T_0\cup T \cup T'$ 
is for visualization. When the input Euclidean triangulation $T_0$ on $S$ is embedded 
in $\mathbb{R}^3$, we can pull back the triangulations $T$ and $T'$ onto $T_0$ for the purpose 
of visualization.  The common refinement of the triangulations $T_0$ and $T$, denoted $T_0\cup T$,
is also computed in the procedure ``Delaunay''. The procedure ``Subdivide'' is to
compute the common refinement of the triangulations of $T_0, T$ and $T'$. In this way, we can
pull back the edges in $T$ back to $T_0$ under the identity map over $(S, V, d)$ and the edges in
$T'$ back to $T_0$ under the discrete conformal map $\phi$ from $d$ to $d'$. 

Figure~\ref{fig:algorithm_illustration} shows the results of the different procedures 
when the algorithm runs over the polyhedral surface of a spherical cap.  
In this example, we can embed the doubled polyhedral surface into $\mathbb{R}^3$
and visualize both the triangulations $T_0$ and $T$ (Figure~\ref{fig:algorithm_overview}(b)). 
Moreover, we can visualize the 
pull-back of the triangulation $T'$ under the map $\phi$ (Figure~\ref{fig:algorithm_overview}(d)). 
In addition, we set the target curvature $0$ everywhere except at four marked points on the boundary 
where it is set to be $\pi/2$. In this way, we can embed the deformed polyhedral surface 
into a rectangle as shown in Figure~\ref{fig:algorithm_overview}(c). We use the procedure described 
in~\cite{gu:2008} to layout a flat surface into the plane. 
Note in all the examples shown in the paper, we fix the $\epsilon$ in Algorithm~\ref{alg:deform} to be $10^{-5}$. 

\begin{figure}[!ht]
\begin{center}
\begin{tabular}{c}
\includegraphics[width=0.9\textwidth]{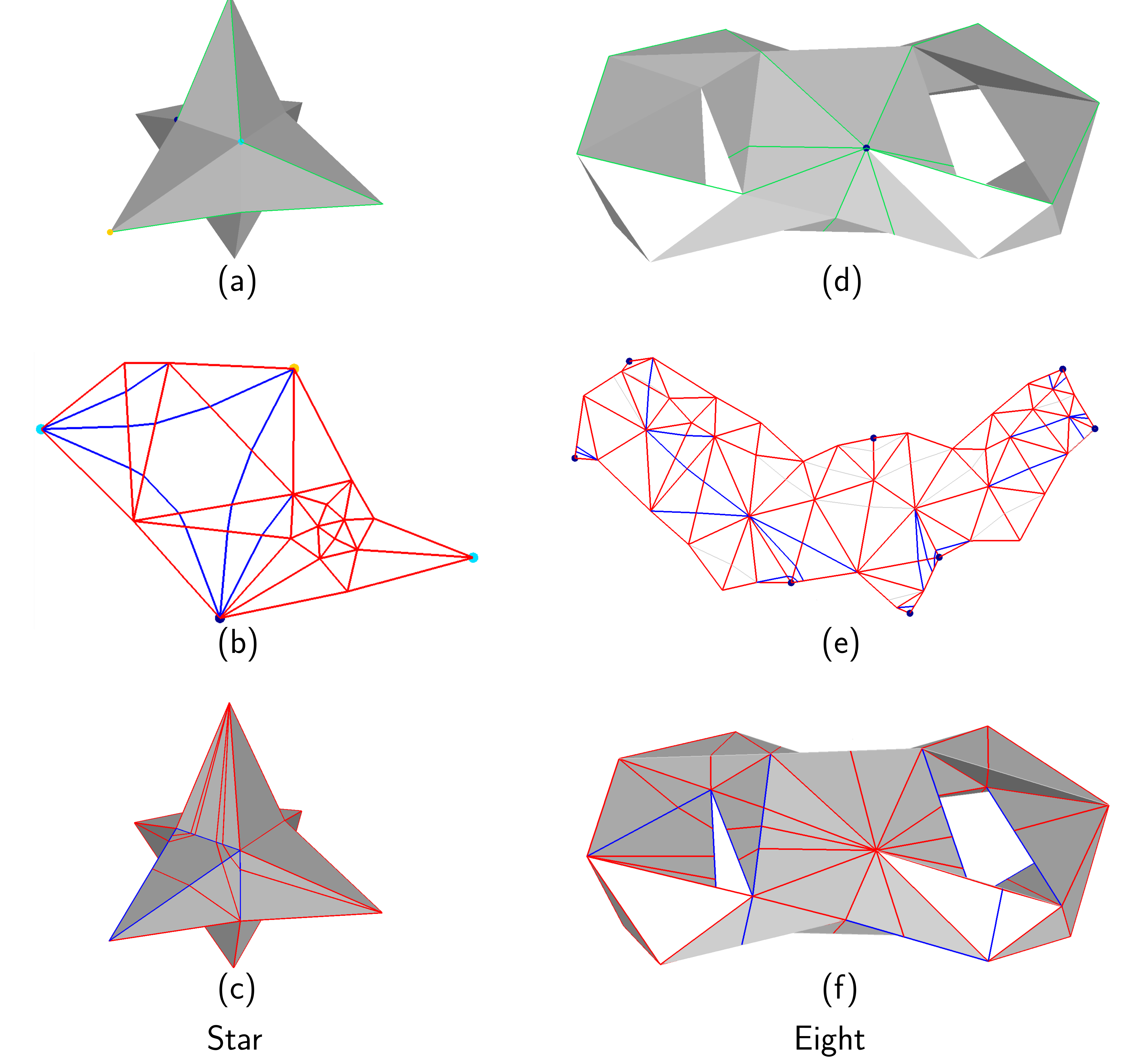}
\end{tabular}
\end{center}
\vspace{0.1in}
\caption{(a, d): The original polyhedral surfaces. The marked vertices
have non zero prescribed curvature. The green edges show a tree (a cut graph)
passing the marked vertices.
(b, e): The planar embedding of the deformed polyhedral surfaces after cutting
them open along the tree (the cut graph). 
The red edges are the edges in the triangulation $T'$ and the blue edges are the images of
the switched edges in the triangulation $T$ under the discrete conformal map $\phi$. 
(c, f): The red edges are the preimages of the edges in the triangulation $T'$ under the map $\phi$, 
the blue edges are the edges in $T$ got switched during the conformal deformation.
The gray edges in (b, c, e, f) are the non Delaunay edges in the 
triangulation $T_0$.
\label{fig:simple_examples}}
\end{figure}

\section{Experimental Results}
In this section, we will show numerical examples, demonstrate numerically 
the convergence of our discrete conformality, and compare to the state of the art. 
For convenience, we follow the notation in Figure~\ref{fig:algorithm_overview}, and 
denote $T_0$ the (doubled) input triangulation, $T$ and $T'$ the Delaunay triangulations
under the initial metric $d$ and the deformed metric $d'$ respectively. 

\subsection{Simple Examples}
\label{sec:simple-examples}
In this subsection, we show a few examples with small number of triangles 
for a clear illustration of the geometric deformation of the metric and the 
combinatorial changes of the triangulation. In the first couple of examples, 
to visualize the result metric, we prescribe the curvature to be zero except at 
a few vertices in order to satisfy the Gauss-Bonnet theorem. A vertex with non 
zero curvature is called {\it singular}. The first example is a polyhedral surface of 
topological sphere which we call Star shown in the left column of 
Figure~\ref{fig:simple_examples}. The total curvature 
of Star is $4\pi$. We choose three singular vertices as marked in 
Figure~\ref{fig:simple_examples}(a) where the curvature is set to be $\frac{4\pi}{3}$. 
To embedding the deformed
Star into the plane, we cut Star along a tree of the edges in $T'$ passing through three singular 
vertices. The tree is shown in green in Figure~\ref{fig:simple_examples}(a). 
The planar embedding of the deformed Star is shown in Figure~\ref{fig:simple_examples}(b),
where the red edges are the edges of the triangulation $T'$ and the blue edges are the images 
of the switched edges in the triangulation $T$ under the discrete conformal map $\phi$. 
The red edges in Figure~\ref{fig:simple_examples}(c) are the preimage of the edges in the 
triangulation $T'$ under the map $\phi$.  The second example is a polyhedral surface of 
genus two which we call Eight shown in the right column of Figure~\ref{fig:simple_examples}. 
The total curvature of Eight is $-4\pi$. We choose one singular vertex  as marked in 
Figure~\ref{fig:simple_examples}(b)  whose curvature is set to be  $-4\pi$.
To embedding the deformed Eight into the plane, we cut Eight
along a cut graph consisting of the edges in $T'$ passing through the singular vertex.
The cut graph is shown in green in Figure~\ref{fig:simple_examples}(d). 
The planar embedding of the deformed Eight is shown in Figure~\ref{fig:simple_examples}(e). 
The red edges and the blue edges in 
Figure~\ref{fig:simple_examples}(e, b) have the same meaning as those in Star. The gray edges
are the non Delaunay edges in the triangulation $T_0$. 

The main purpose of the next couple of examples is to show the triangulation of $T'$ when we 
prescribe a curvature close to the boundary of the domain of all possible curvatures. In both examples, 
the curvature is prescribed to be $2\pi -0.1$ at every vertex except at one vertex (labeled by $a$ 
for later reference) whose curvature is 
set to satisfy the Gauss-Bonnet theorem and usually a negative value. In the example of Star, the 
prescribed curvature at the vertex $a$ is $-22\pi + 1.4$, and in the example of Eight, it is $-58\pi +2.7$. 
In Figure~\ref{fig:extreme-curvatrue}, the red edges are the preimage of the edges in the triangulation $T'$
pulled-back by the discrete conformal map $\phi$ into the input surface. All of the triangle in $T'$ has 
$a$ as its vertex. In fact, in these two examples, at least two of three vertices of any triangle 
in $T'$ are $a$. 

\begin{figure}[!t]
\begin{center}
\begin{tabular}{cc}
\includegraphics[width=0.3\textwidth]{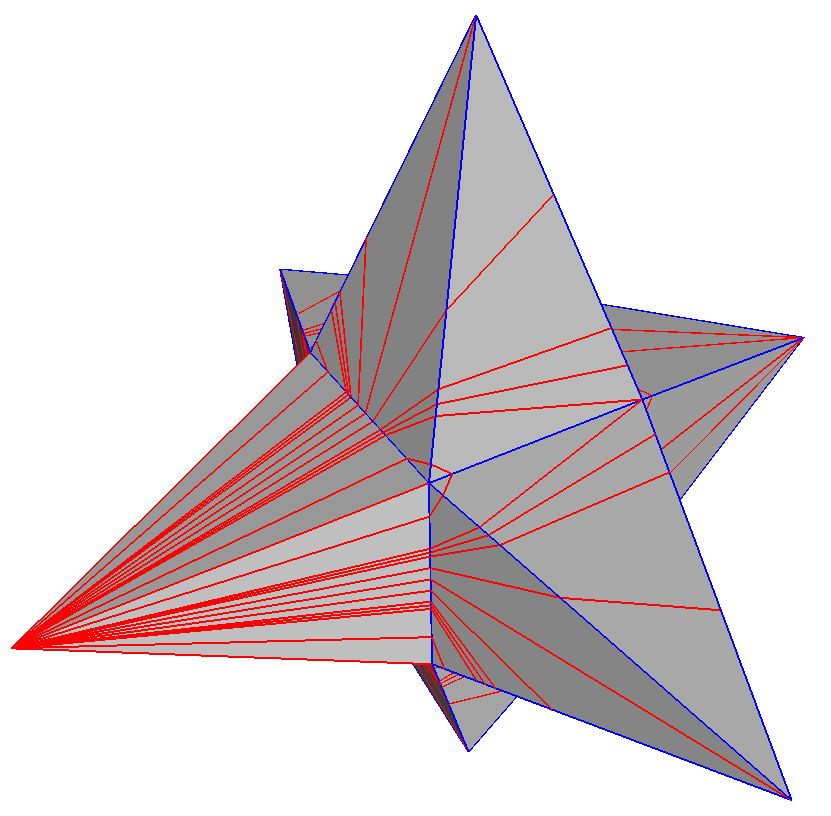} & 
\includegraphics[width=0.5\textwidth]{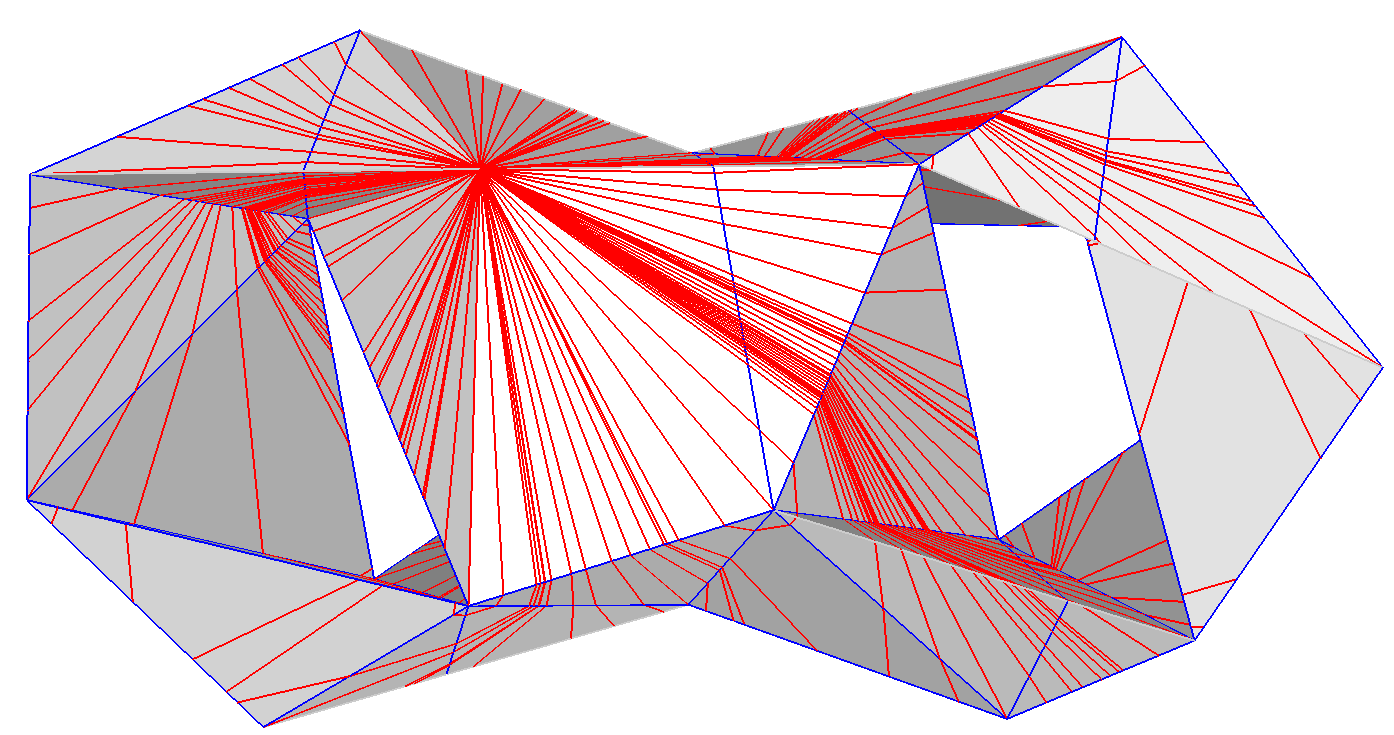}\\
Star & Eight
\end{tabular}
\end{center}
\vspace{-0.1in}
\caption{
The color scheme of the edges are the same as that in Figure~\ref{fig:simple_examples}(c, f). 
\label{fig:extreme-curvatrue}}
\end{figure}

\subsection{Convergence}
In this subsection, we will present numerical evidences showing the convergence of our discrete conformality.
In addition, we will demonstrate the efficiency and the robustness of our algorithm, in particular against 
the quality of the input triangulations, and compare its performance to the state of the art. 
We check how much the conformality is preserved when the triangulated surfaces
are flattened into the plane, and use two types of criteria to measure the conformality. 

\subsubsection{Criteria} 
For the examples where the (approximated) ground truth of conformal flattening is known, we can compare
the results with the ground truth. 
Let $u_{gt}$ be the flattening map of the (approximated) ground truth, and 
$u$ be the flattening map constructed by our algorithm or other methods described below. 
We use the following two norms to measure the approximation error:
\begin{equation}
e_2 = \left(\frac{\sum_{i\in V}\|u(i) - u_{gt}(i)\|^2 A_i}{\sum_{i\in V} A_i}\right)^{1/2},~\text{and}~~e_{\infty} = \max_{i\in V}\{\|u(i) - u_{gt}(i)\|\}.
\end{equation}
where $A_i$ is the area weight, which is estimated as a third of the total area of 
the triangles in $T$ incident to the vertex $i$. 

In general, the ground truth of conformal flattening is not known. Given an orientation 
preserving map $h$ between two Riemann surfaces, the Beltrami coefficient is
$\mu = \frac{h_{\bar{z}}}{h_{z}}$, where $z$ is a complex number representing the local coordinates. 
The map $h$ sends an infinitesimal circle to an 
infinitesimal ellipse with the ratio of major semiaxis to minor semiaxis equal
$D(h) = \frac{1+|\mu|}{1-|\mu|}$. Note $|\mu|< 1$ as the map $h$ preserves orientation. 
$D(h)$ is called the conformal distortion of the map $h$ and $D(h) = 1$ if and only if $h$ 
is conformal. So we check the conformality of the map $h$
by measuring how far $D(h)$ is away from $1$. Specifically, we estimate $\|D(h)-1\|_{L_2}$ and 
$\|D(h)-1\|_{L_\infty}$.

Let $(S, T=(V, E, F))$ be the input triangulated surface.
In the methods we described below for comparison, 
the constructed flattening map $h$ is piecewise linear, namely on a triangle $f \in F$, 
$h|_f$ is the linear extension of the map on the vertices of the triangle. 
let $L_f(z) = \alpha z + \beta \bar{z}$ represent the linear map $h|_f$. The conformal distortion 
of this linear map $L_f$ can be computed as  $D(L_f) = \frac{|\alpha| + |\beta|}{|\alpha| - |\beta|}$. 
For a piecewise linear flattening map $h$, we have
\begin{eqnarray}
\|D(h)-1\|_{L_2} &=& d_2 = \left(\frac{\sum_{f \in F}(D(L_f) - 1)^2 \text{area}(f)}{\sum_{f\in F} \text{area}(f)}\right)^{1/2},~\text{and}\\
\|D(h)-1\|_{L_\infty} &=& d_{\infty} = \max_{f\in F}\{D(L_f) - 1\}.
\end{eqnarray}
In our method, from the discussion in Section~\ref{sec:disconf-map}, the constructed flattening map $h$ 
is piecewise circumcircle preserving projective. Specifically, for a polygonal face $A$ in
the common refinement $T\cup T'$, let $f$ and $f'$ be the triangle in $T$ and $T'$ containing $A$. 
The map $h$ restricted to $A$, denoted $h|_A$, is the restriction to $A$ of the circumcircle preserving 
projective map from the mapping triangle $f''$ to the triangle $f'$.
Let $L_A(z) = \alpha z + \beta \bar{z}$ be the linear map from $f''$ to $f'$. In~\cite{glsw2}, we have shown
that $D(h|_A)\leq D(L_A)$. Therefore, for our flattening map $h$, we have the 
following upper bounds on $\|D(h)-1\|_{L_2}$ and $\|D(h)-1\|_{L_\infty}$, which are easy to estimate. 
\begin{eqnarray}
\|D(h)-1\|_{L_2} &\leq& d_2 = \left(\frac{\sum_{A \in F(T\cup T')}(D(L_A) - 1)^2 \text{area}(A)}{\sum_{f\in F(T\cup T')} \text{area}(A)}\right)^{1/2},~\text{and}\\
\|D(h)-1\|_{L_\infty} &\leq& d_{\infty} = \max_{f\in F}\{D(L_A) - 1\}, 
\end{eqnarray}
where $F(T\cup T')$ denotes the set of the polygonal faces in $T\cup T'$ and $\text{area}(A)$ denotes
the area of $A$ as a subset of the triangle $f$. 

\subsubsection{Conformal Flattening Methods}
We briefly describe three methods including ours of conformally flattening triangulated surfaces into the plane.

\vspace{0.1in}
\noindent{\bf Method of Discrete conformal Deformation (DC).} 
This flattening method is based on our discrete conformal deformation. To flatten a triangulated surface
into the plane, we basically prescribe the curvature to be $0$ and solve the problem of prescribing 
curvature using the algorithm described in Section~\ref{sec:algorithm}. Due to the obstruction of 
topology, the target curvature can not be $0$ everywhere. We call those whose curvature are not zero 
the singular vertices. For a surface of topological disk, we choose three singular vertices on the boundary
and set the curvature there to be $2\pi/3$. In this way, we flatten a triangulated surface of topological
disk onto an equilateral triangle. This flattening map is guaranteed to be one to one. For a surface of 
topological sphere, as we did in Section~\ref{sec:simple-examples} for Star, we choose three singular 
vertices whose curvatures are set to be $4\pi/3$, and cut the surface along a tree of the edges in $T'$
passing through the singular vertices for flattening the triangulated surface. 
For a surface of genus $g\geq 1$, we choose $2(g-1)$ singular vertices whose curvatures are set to be $-2\pi$. 
and cut the surface along a cut graph consisting of the edges in $T'$ and passing the singular vertices 
for flattening the triangulated surface.  

\vspace{0.1in}
\noindent{\bf Method of Holomorphic Form (HF).} 
Gu and Yau~\cite{Gu:2003} proposed a method to conformally flatten a surface of 
genus $g\geq 1$ into the plane using holomorphic one-form. Assume $h=f(z)dz$ is a 
holomorphic one-form of the surface and it is well-known that the metric 
$|f(z)|^2 dz d\bar{z}$ is conformal and flat when $f(z) \neq 0$. Noticing that 
any holomorphic one-form can be decomposed as $h = \omega + i(*\omega)$
where $w$ is a real harmonic one-form and $*w$ is its conjugate, Gu and Yau 
developed discrete algorithms for approximating from a triangulated surface
a basis $\{\omega_1, \cdots, \omega_{2g}\}$ of the space of real harmonic one-forms
and their conjugates ${*\omega_1, \cdots, *\omega_{2g}}$. Then 
$\{h_1=\omega_1 + i(*\omega_1), \cdots, h_{2g}=\omega_{2g} + i(*\omega_{2g})\}$ 
contains a basis of the space of holomorphic one-forms and any linear combination
$h=\sum_i a_i h_i$ is a holomorphic one form. Integrate the real part 
and the imaginary part of $h$ along the edges of the triangulated surface to
obtain the $x$-coordinates and the $y$-coordinates  respectively for the vertices. 
Note the $x, y$-coordinate functions computed by integration are multi-valued 
at a subset of vertices. The edges with both endpoints multi-valued form a cut-graph 
of the surface. Cut the surface along this cut-graph and obtain a fundamental domain 
of the surface. The $x, y$-coordinate functions conformally map this fundamental domain 
into a planar region. For the surfaces with boundary, one can double the surface to remove
the boundary by gluing two copies along the boundary, and apply the above algorithm and 
take half of the computed planar embedding. For a surface with topological disk, in order
to obtain an embedding onto unit disk, the following procedure is used: (1) remove a triangle
from the given triangulated surface to make an annulus; (2) apply the above algorithm to 
obtain an planar embedding of rectangular shape; (3) take exponential to map the rectangle domain
into an annulus with unit outer radius, and put back the removed triangle to obtain the final embedding
onto unit disk.  
In our experiments, we use the implementation made available to us by the authors.

\vspace{0.1in}
\noindent{\bf Method of Bounded Distortion (BD).} 
Lipman~\cite{Lipman12} considered the problem of finding a piecewise linear map $h$ mapping
a triangulated surface into the plane so that the conformal distortion $D(h)$ is less than 
some prescribed number $C$. This amounts to solve a non-convex optimization problem, 
which was reduced to a conic optimization problem by restricting the domain of optimization
to a convex subset. In~\cite{Lipman12}, Lipman also proposed a binary search strategy to find
a map with the ``optimal'' conformal distortion. Note the reduced conic optimization may miss
a map with the conformal distortion less than the prescribed number, and thus gives a wrong 
feedback to the binary search, which therefore may not reach a map with true
optimal conformal distortion. In our experiments, the following iterative procedure is used to 
find a map with small conformal distortion, which is more efficient comparing to 
the binary search strategy. In each iteration, assume a piecewise linear map $h^k$ is given,
and one construct a convex set of maps whose conformal distortion is less than 
$\|D(h^k)\|_\infty$ and then use the conic optimization to find the next map $h^{k+1}$ 
in this convex set.  The iteration is started with the Tutte embedding where the position 
of an interior vertex is the average of the positions of its neighboring vertices, and 
iterate at most 10 times. For the vertices on the boundary, we fix the positions of 
three of them and impose the linear constraints on the others so that the resulting 
range is a triangle. The optimization package MOSEK~\cite{mosek} is used for conic optimization. 
In our experiments, we use the implementation made available to us by the author.

\subsubsection{Examples}
Now we run the above conformal flattening methods over several examples to show their performance, in particular
to compare their convergence property.  

\vspace{0.1in}
\noindent{\bf Spherical Cap.} A spherical cap is a portion of a sphere cut off by a plane, which can be
conformally flattened onto unit disk by composing the stereographic projection with a scaling. So in this 
example, we have the ground truth for conformal flattening, denoted as $u_{gt}$. We run the aforementioned 
methods over four continuously refined triangulated surfaces with approximately $1000$, $4000$, $16000$ and $64000$ vertices, 
which are obtained by applying Cocone~\cite{Amenta00asimple}, a surface reconstruction algorithm, to the randomly drawn samples 
on the spherical cap and some additional samples on its boundary. Figure~\ref{fig:hemisphere_input} shows the triangulated surfaces
with $1000$ and $4000$ vertices.  

\begin{figure}[!t]
\begin{center}
\begin{tabular}{cc}
\includegraphics[width=0.42\textwidth]{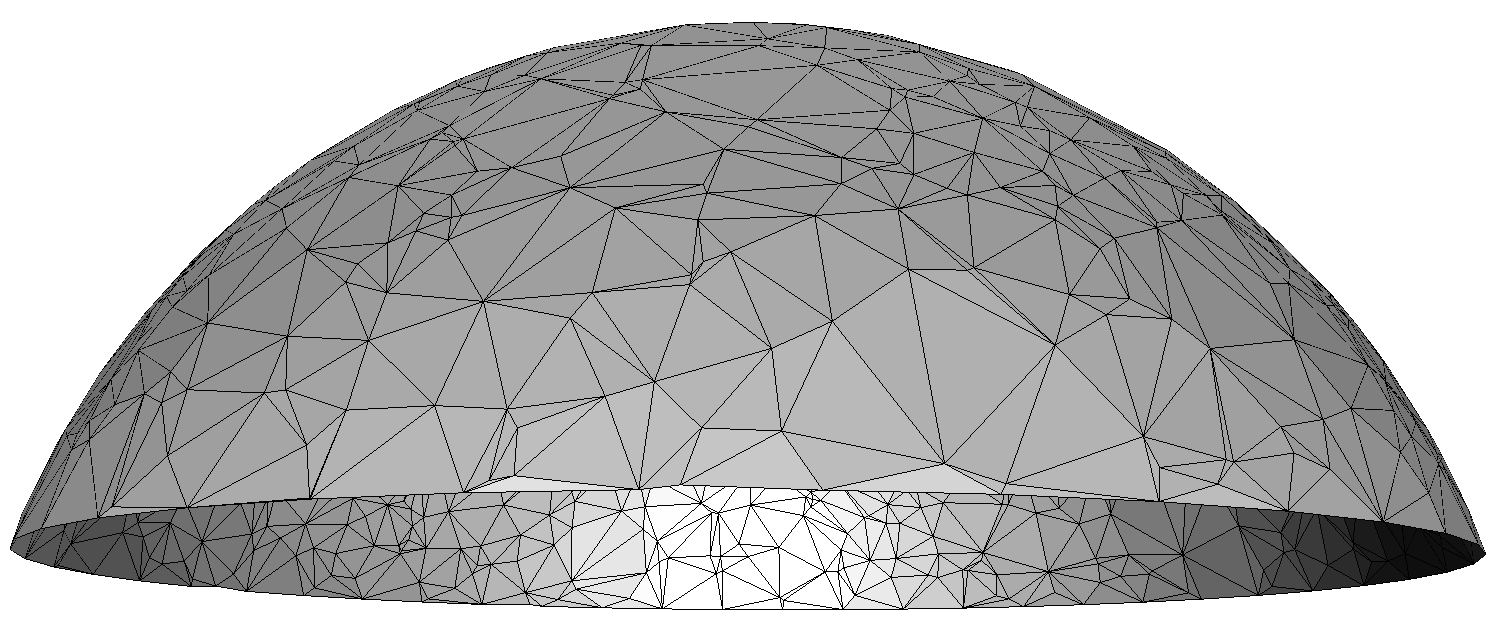} & 
\includegraphics[width=0.42\textwidth]{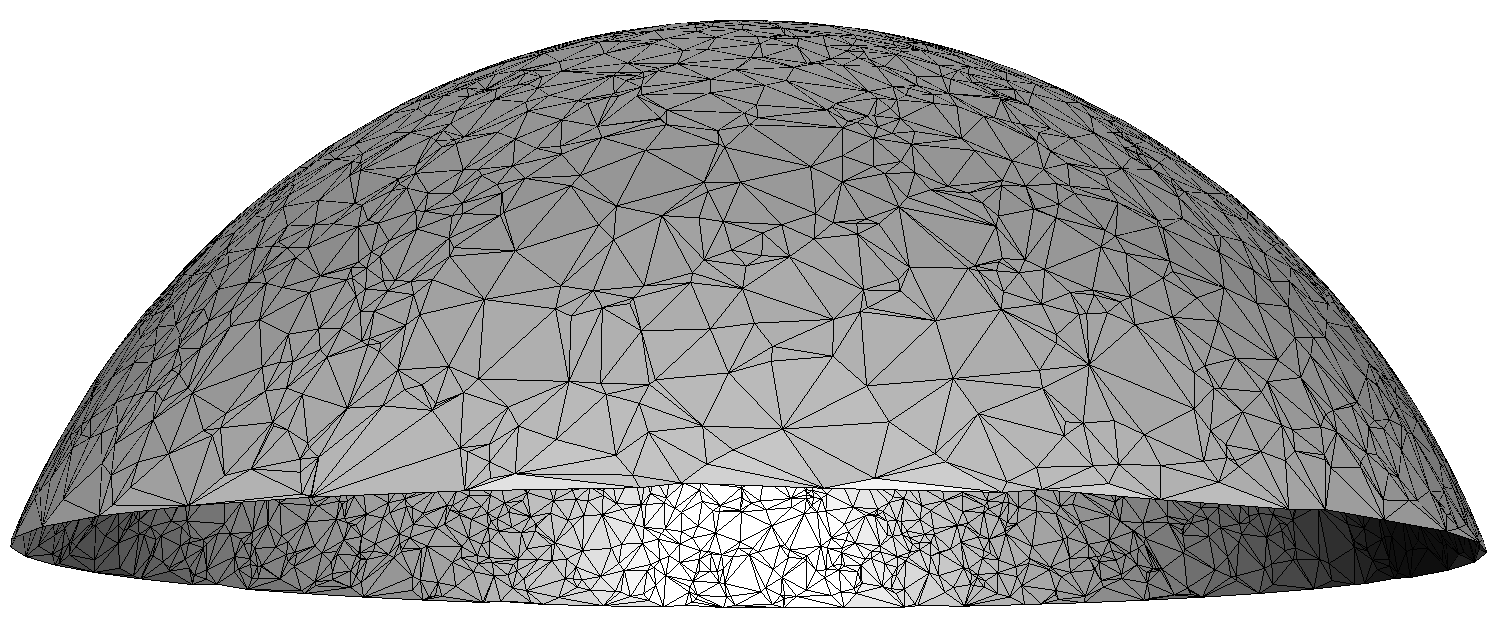}
\end{tabular}
\end{center}
\vspace{-0.1in}
\caption{The triangulated surfaces of Spherical Cap with $1000$ (Left) and $4000$ (Right) vertices. 
\label{fig:hemisphere_input}}
\end{figure}

\begin{figure}[!h]
\begin{center}
\begin{tabular}{ccc}
\includegraphics[width=0.33\textwidth]{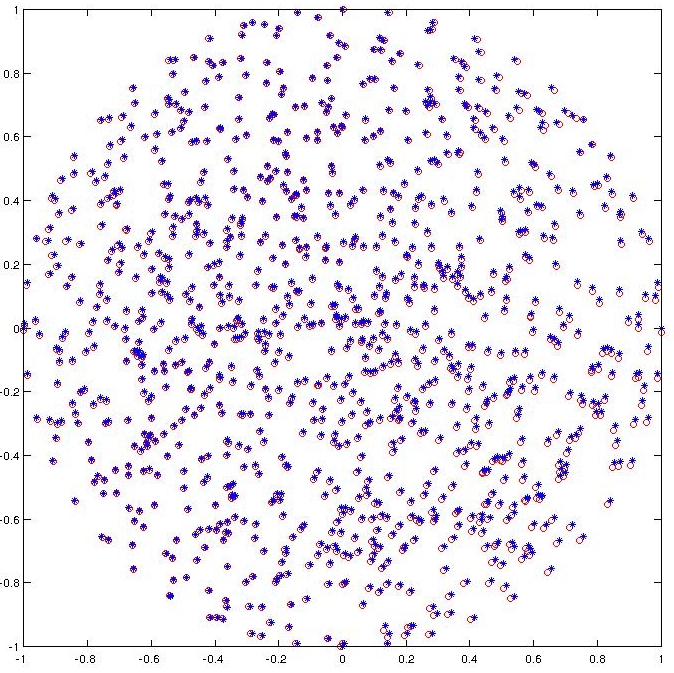} & 
\includegraphics[width=0.33\textwidth]{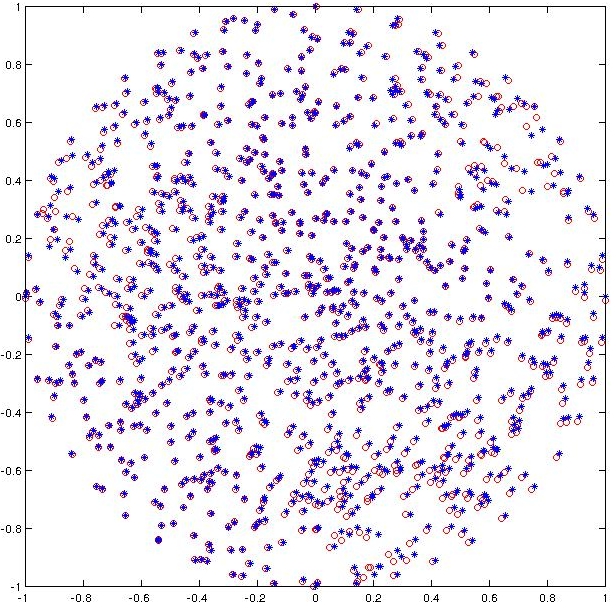}& 
\includegraphics[width=0.33\textwidth]{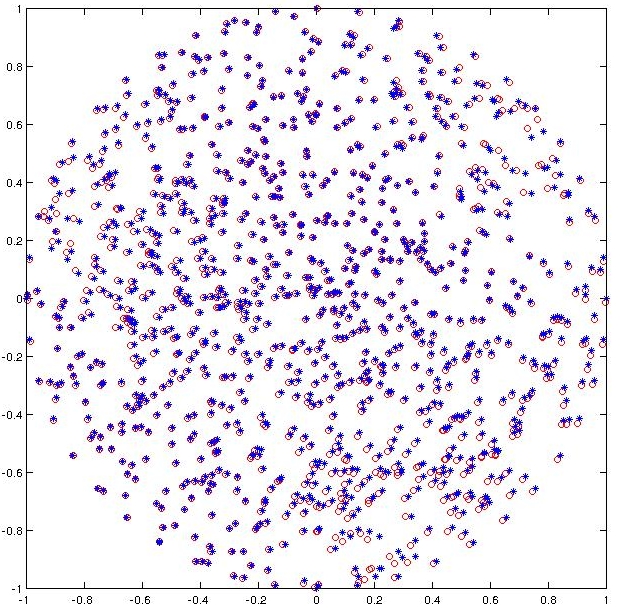} \\
\includegraphics[width=0.33\textwidth]{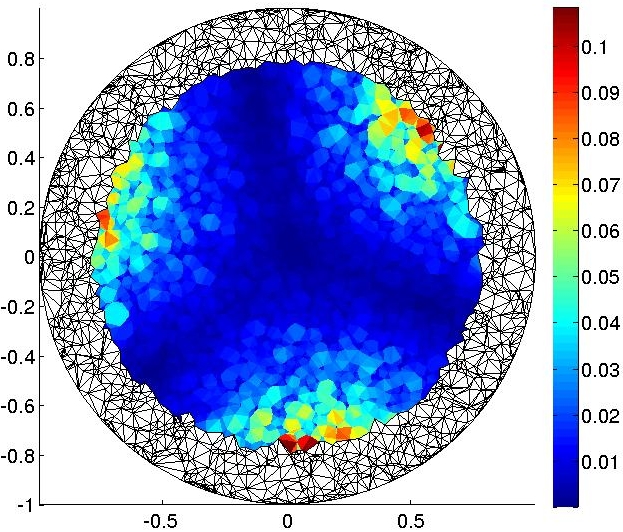} & 
\includegraphics[width=0.33\textwidth]{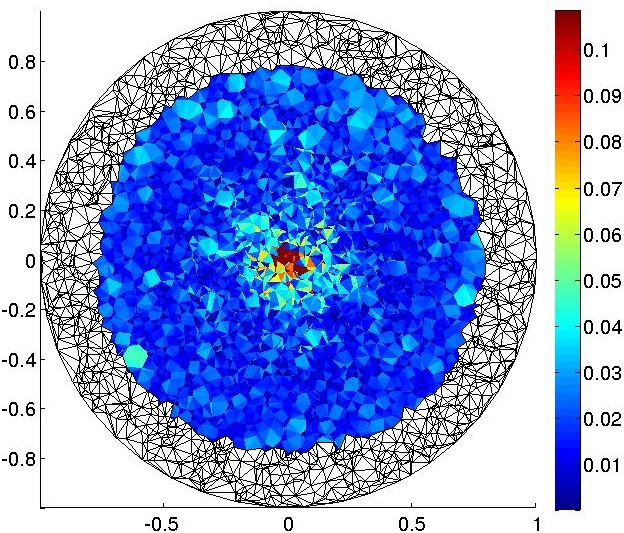} &
\includegraphics[width=0.33\textwidth]{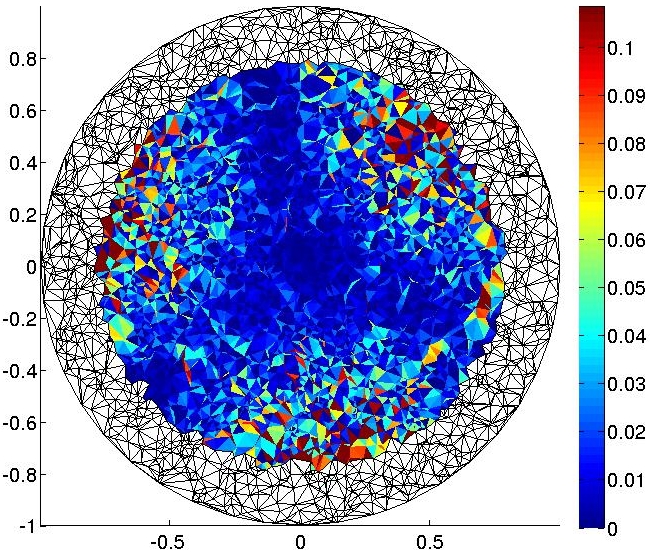} \\
{\bf DC} & {\bf HF} & {\bf BD}
\end{tabular}
\end{center}
\vspace{-0.1in}
\caption{Results for Spherical Cap: The first row shows the planar embedding of the vertices
for the triangulated surface with $1000$ vertices. The blue "*" is the ground truth and the 
red "o" is the results computed by different methods; 
The second row plots the conformal distortion $D(h)-1$ of the planar embedding computed by 
different methods from the triangulated surface
with $4000$ vertices. Note that for the purpose of comparison, the range of color map
is fixed as $[0, 0.11]$, although the maximal conformal distortion of the planar embedding 
by HF and BD is larger than $1.11$, as shown in Table~\ref{tbl:hemisphere}.
\label{fig:hemisphere}}
\end{figure}

\begin{table}[!h]
\begin{center}
\begin{tabular}{| c | c | c | c | c |}
\hline
Method  & 1000 &  4000 & 16000 & 64000   \\
\hline
{\bf DC} & (0.0102, 0.0137) & (0.0034, 0.0048) & (0.0014, 0.0020) & (0.0010, 0.0014) \\
\hline
{\bf HF} & (0.0023, 0.0078 ) & (0.0010, 0.0036) & (0.0011, 0.0030) & (0.0009, 0.0016)\\
\hline
{\bf BD}& (0.0131, 0.0300 ) & (0.0054, 0.0150) & (0.0040, 0.0092) & (0.0023, 0.0051)\\
\hline
\multicolumn{5}{|c|}{ $(e_2, e_\infty)$} \\
\hline
{\bf DC} & (0.0505, 0.1592) & (0.0303, 0.1084) & (0.0160, 0.0570) & (0.0082, 0.0377) \\
\hline
{\bf HF} & (0.0638, 12.9485 ) & (0.0270, 1.4460) & (0.0169, 1.3563) & (0.0085, 0.8276)\\
\hline
{\bf BD}& (0.0829, 0.6582 ) & (0.0547, 0.9318) & (0.0313, 1.0055) & (0.0192, 1.1988)\\
\hline
\multicolumn{5}{|c|}{ $(d_2, d_\infty)$} \\
\hline
{\bf DC} & 0.152 & 0.636 & 4.54 & 26.9 \\
\hline
{\bf HF} &  1.03 & 3.14 & 12.7 & 53.7 \\
\hline
{\bf BD}& 28.4 & 182  & 744 & 3395\\
\hline
\multicolumn{5}{|c|}{ timing (sec)} \\
\hline
\end{tabular}
\end{center}
\vspace{-0.1in}
\caption{Spherical Cap: approximation errors and running time.
\label{tbl:hemisphere}
}
\end{table}

For the methods of DC and BD, we choose three vertices $\{a, b, c\}$ on the boundary so that they are mapped to 
the vertices of an equilateral triangle. To compare with the ground truth, we use the Schwarz-Christoffel mapping
which can explicitly evaluate the  conformal transformation mapping unit disk onto a triangle. 
In fact, we use the Schwarz-Christoffel Toolbox~\cite{driscoll} to compute the inverse map sending $a, b, c$ to 
$u_{gt}(a), u_{gt}(b), u_{gt}(c)$, respectively. In the method of HF, we have already embedded the spherical
cap onto unit disk. We align the computed embedding to the ground truth by a M\"{o}bius transformation.

The first row of Figure~\ref{fig:hemisphere} shows the embedding of the vertices computed by different methods. 
In Table~\ref{tbl:hemisphere}, we show the approximation errors: $e_2$, $e_\infty$, $d_2$, $d_\infty$, and 
the timing in seconds used by different conformal flattening methods to compute the embedding. 
Note since the convergence is often stated for a compact region away from the boundary, we estimate the errors
$e_2$ and $e_\infty$ over the vertices which are mapped into the disk $D$ of radius $0.8$, and the errors
$d_2$ and $d_\infty$ over the polygons or triangles whose vertices are mapped into the disk $D$. 
In the second row of Figure~\ref{fig:hemisphere}, we only plot the conformal distortion of those polygons 
or triangles used to evaluate $d_2$ and $d_\infty$. 

From Table~\ref{tbl:hemisphere}, three methods all converge about linearly in terms of the $e_2$, 
$e_\infty$ and $d_2$ errors. In terms of the absolute value of these approximation errors, 
BD performs worse than DC and HF.  Only DC has a convergent $d_\infty$ error, which is 
approximately linear. In terms of running time, DC and HF have a similar performance, 
while $BD$ is much slower.


\begin{figure}[!t]
\begin{center}
\begin{tabular}{cc}
\includegraphics[width=0.35\textwidth]{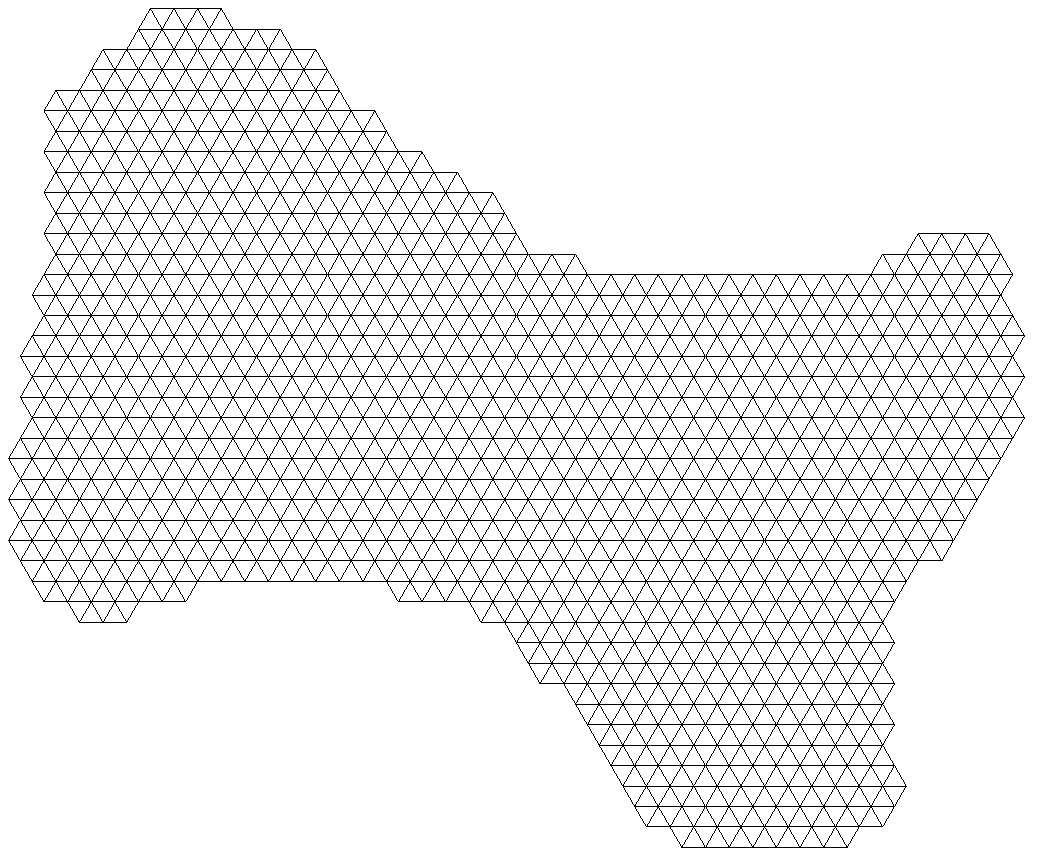} & 
\includegraphics[width=0.35\textwidth]{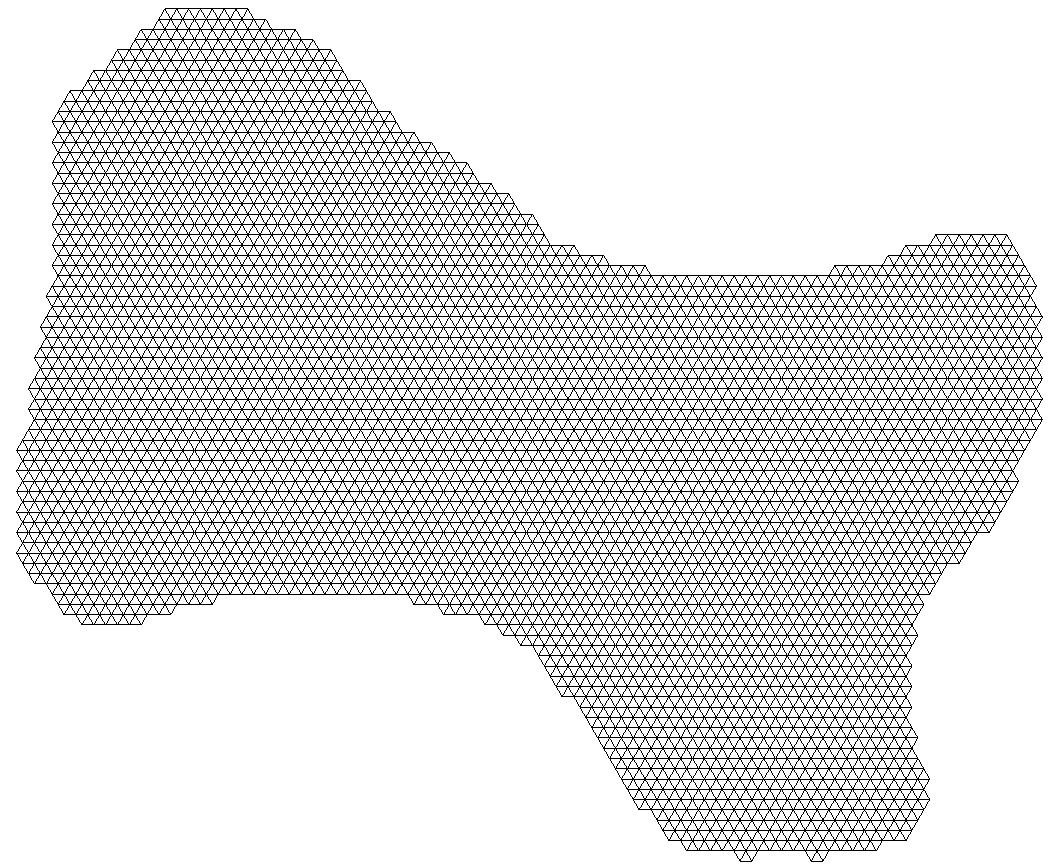} \\
\includegraphics[width=0.3\textwidth]{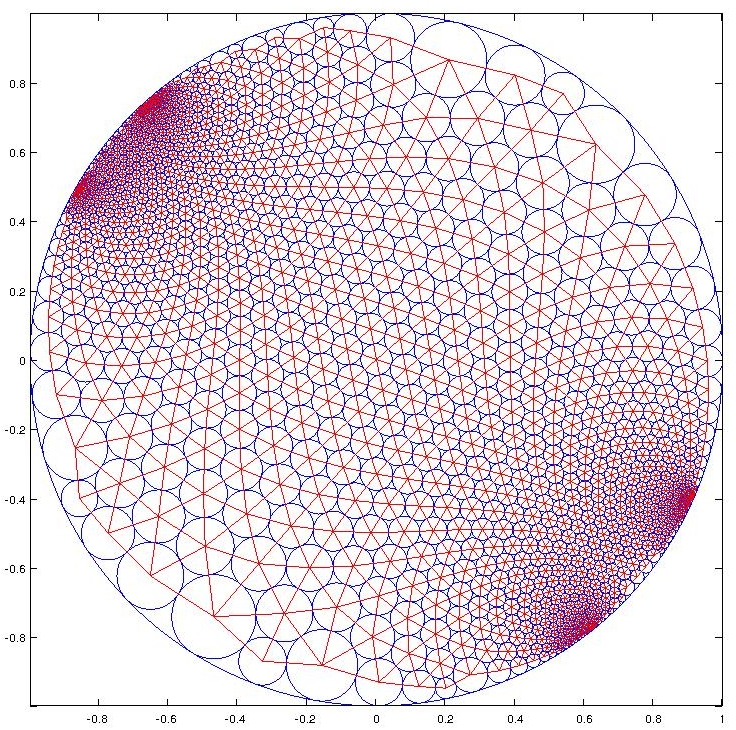} &
\includegraphics[width=0.3\textwidth]{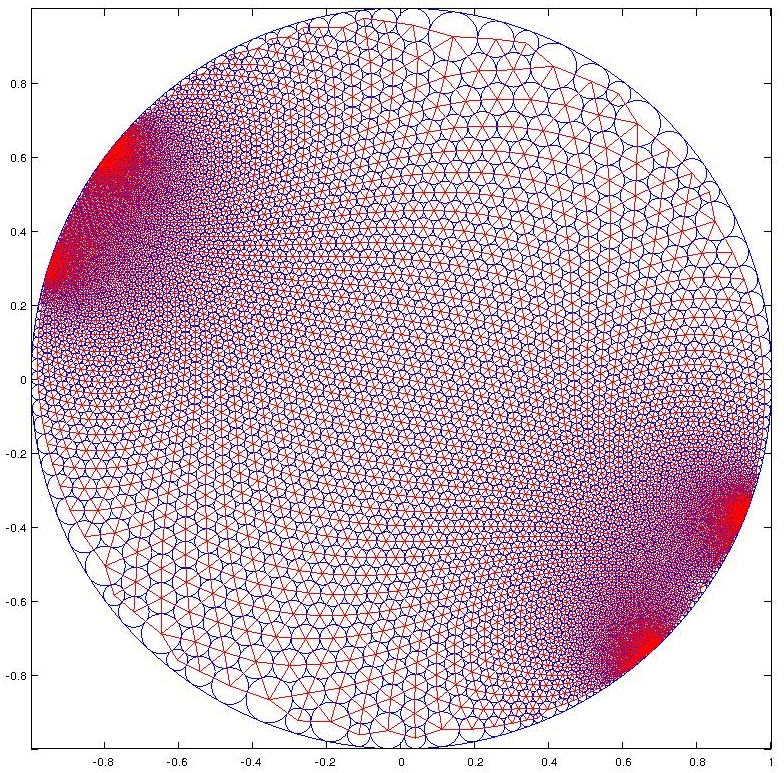}
\end{tabular}
\end{center}
\vspace{-0.1in}
\caption{Hexagonal Meshes and their circle packing. The first row: the input hexagonal triangulations with 1000 vertices (Left)
and 4000 vertices (Right). The second row: the circle packing in unit disk to the hexagonal triangulation above. 
\label{fig:hexagon_input}}
\end{figure}

\vspace{0.1in}
\noindent{\bf Hexagonal Mesh.}
The Riemann mapping from a planar region to unit disk
can be approximated using the Thurston's circle packing. 
Consider a hexagonal triangulation inside a planar region. 
One can explicitly construct a circle packing of unit disk, that is a collection of closed disks inside  unit disk 
having the following properties: (1) the interiors of the disks are disjoint; (2) the nerve of this collection 
of disks is isomorphic as graph to the $1$-skeleton of the hexagonal triangulation; (3) the boundary of the disk
corresponding to each boundary vertex of the hexagonal triangulation tangentially touches unit circle.  
This construction induces a map, denoted $u_{gt}$, from the hexagonal triangulation into unit disk by 
mapping the vertices of the hexagonal triangulation to the centers of the corresponding disks and extending 
linearly to the triangles. Figure~\ref{fig:hexagon_input} shows two hexagonal triangulations inside a fixed planar 
region and their corresponding circle packings of unit disk.
Rodin and Sullivan~\cite{RS} showed the above induced map converges to the Riemann mapping
from the planar region to unit disk as the size of the triangles in the hexagonal triangulation goes to $0$.

\begin{figure}[!t]
\begin{center}
\begin{tabular}{ccc}
\includegraphics[width=0.33\textwidth]{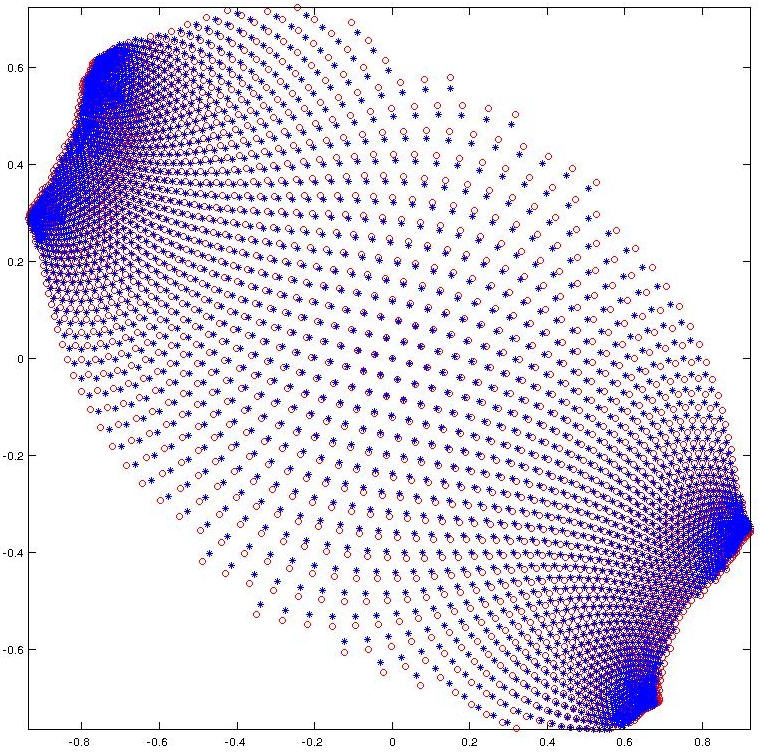} & 
\includegraphics[width=0.33\textwidth]{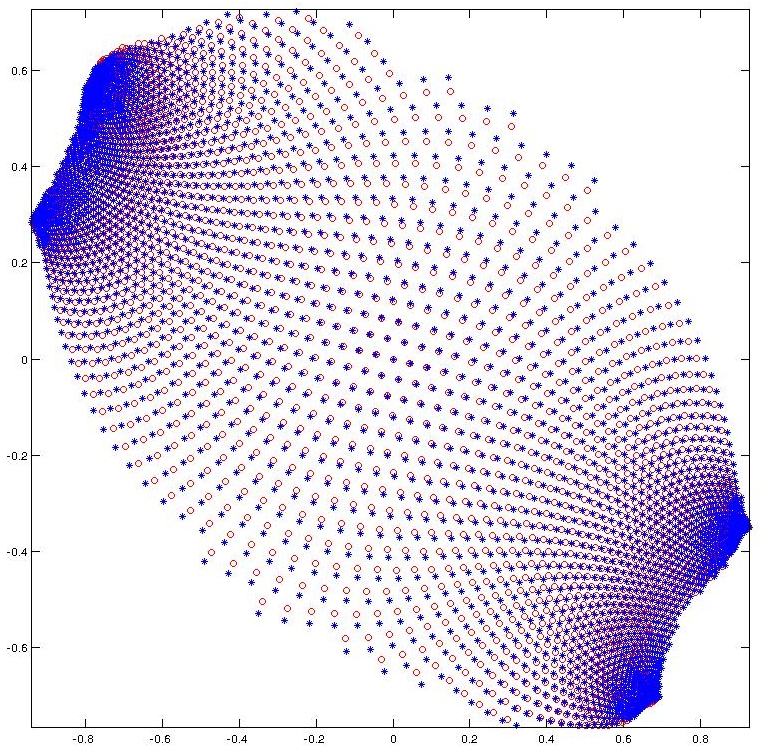}& 
\includegraphics[width=0.33\textwidth]{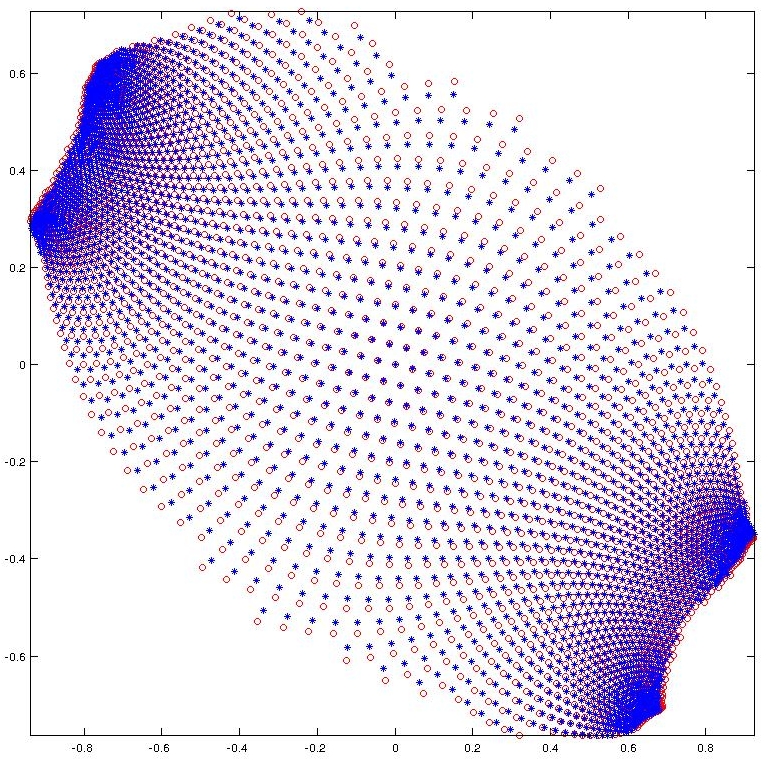} \\
\includegraphics[width=0.33\textwidth]{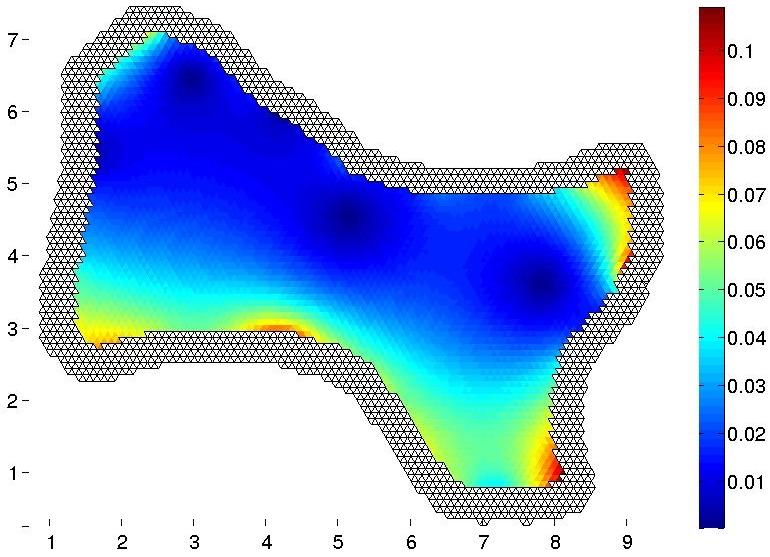} & 
\includegraphics[width=0.33\textwidth]{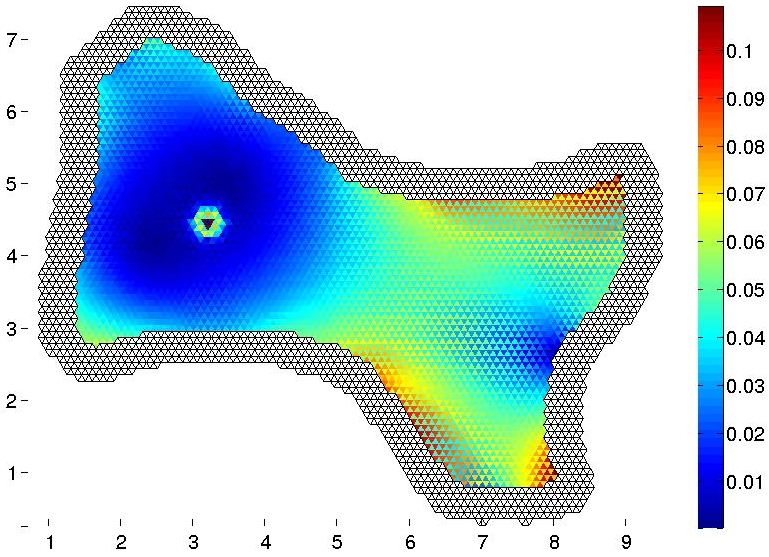} &
\includegraphics[width=0.33\textwidth]{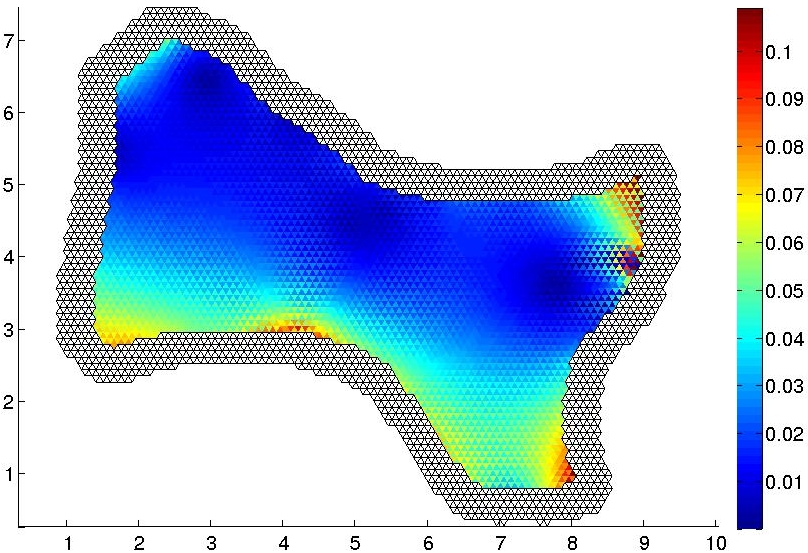} \\
{\bf DC} & {\bf HF} & {\bf BD}
\end{tabular}
\end{center}
\vspace{-0.1in}
\caption{
Results for Hexagonal Mesh: The first row shows the planar embedding of the vertices
for the triangulated surface with $1000$ vertices. The blue "*" is the ground truth and the 
red "o" is the results computed by different methods; 
The second row plots the conformal distortion $D(h)-1$ of the planar embedding computed by 
different methods from the triangulated surface with $4000$ vertices. 
Note that for the purpose of comparison, the range of color map
is fixed as $[0, 0.11]$, although the maximal conformal distortion of the planar embedding 
by HF and BD is larger than $1.11$.
\label{fig:hexagon}}
\end{figure}

\begin{table}[!h]
\begin{center}
\begin{tabular}{| c | c | c | c | c |}
\hline
Method  & 1000 &  4000 & 16000 & 64000   \\
\hline
{\bf DC} & (0.0257, 0.0486) & (0.0133, 0.0266) & (0.0067, 0.0141) & (0.0034, 0.0074) \\
\hline
{\bf HF} & (0.0275, 0.0538 ) & (0.0142, 0.0305) & (0.0070, 0.0154) & (0.0035, 0.0079)\\
\hline
{\bf BD}& (0.0273, 0.0524) & (0.0137, 0.0280) & (0.0069, 0.0153) & (0.0035, 0.0081)\\
\hline
\multicolumn{5}{|c|}{ $(e_2, e_\infty)$} \\
\hline
{\bf DC} & (0.0650, 0.2000) & (0.0326, 0.1088) & (0.0163, 0.0490) & (0.0081, 0.0259) \\
\hline
{\bf HF} & (0.0883, 0.2238 ) & (0.0418, 0.1405) & (0.0203, 0.1386) & (0.0100, 0.1379)\\
\hline
{\bf BD}& (0.1333, 0.2028 ) & (0.0781, 0.1278) & (0.0445, 0.0690) & (0.0277, 0.0344)\\
\hline
\multicolumn{5}{|c|}{ $(d_2, d_\infty)$} \\
\hline
{\bf DC} & 0.096 & 0.584 & 3.11 & 24.5 \\
\hline
{\bf HF} &  0.987 & 3.14 & 12.0 & 47.3 \\
\hline
{\bf BD}& 19.0 & 77.8  & 322 & 1821\\
\hline
\multicolumn{5}{|c|}{ timing (sec)} \\
\hline
\end{tabular}
\end{center}
\vspace{-0.1in}
\caption{Hexagonal Mesh: approximation errors and running time.
\label{tbl:hexagon}
}
\end{table}

We run the aforementioned methods over four continuously refined hexagonal triangulations of a planar region
of the side lengths $0.2$, $0.1$, $0.05$ and $0.025$. The number of vertices in those triangulations are approximately
$1000$, $4000$, $16000$ and $64000$. The first row in Figure~\ref{fig:hexagon_input} shows the input hexagonal triangulations. 
Note a circle packing in unit disk of a triangulation is not unique. We normalize the circle 
packing by choosing a vertex, denoted $o$ for later reference, from the hexagonal triangulation and mapping it to the origin.
The remaining freedom is a rotation, which however doe not affect the consistency of the error estimations.
To make the normalization consistent across different hexagonal triangulations, the four chosen vertices $o$ 
(one from each triangulation) have the same coordinate. See the second row in Figure~\ref{fig:hexagon_input} for the
resulting circle packings of the hexagonal triangulations. 

For the methods of DC and BD, we again choose three vertices $\{a, b, c\}$ 
on the boundary so that they are mapped to the vertices of an equilateral triangle, and then
we map the equilateral triangle onto unit disk by the inverse of the Schwarz-Christoffel mapping
and finally, we employ an automorphism of unit disk to obtain the map $u$ with $u(o) = 0$ and 
$\text{arg}(u(a)) = \text{arg}(u_{gt}(a))$. For the method of HF, we also apply an automorphism 
of unit disk to the computed embedding to obtain the same alignment.

Similarly, since the convergence is often stated for a compact region away from the boundary,
we estimate the errors $e_2$ and $e_\infty$ over the vertices which are more than $0.4$ away 
from the boundary of the planar region, and the errors $d_2$ and $d_\infty$ over the polygons 
or triangles with their vertices satisfying the same requirement. The first row of Figure~\ref{fig:hexagon} 
shows the embedding of the vertices computed by different methods, and the second row of Figure~\ref{fig:hexagon} 
plots the conformal distortion of the approximated Riemann mapping by different methods.
In Table~\ref{tbl:hexagon}, we show the approximation errors: $e_2$, $e_\infty$, $d_2$, $d_\infty$, and 
the timing in seconds used by different conformal flattening methods for computing the embedding. 
From Table~\ref{tbl:hexagon}, a similar pattern as Spherical Cap is observed: 
all of the methods show linear convergence in the errors $e_2$, $e_\infty$, $d_2$, and 
DC remains converging linearly in the $d_\infty$ error. In this example, BD becomes convergent linearly
in the $d_\infty$ error. This may be due to the fact that the triangles are all well-shaped in the 
hexagonal triangulations.

\vspace{0.1in}
\noindent{\bf Planar Region.}
The main purpose of this example is to show how the quality of the input triangulation affects 
the conformality. We generate a triangulation with $1500$ vertices of a planar region as shown on 
the left in Figure~\ref{fig:plane_triangulation_input}, and then subdivide the triangulation three times 
by adding the midpoints of the edges and splitting each triangle into four smaller ones, and finally
obtain three more continuously refined triangulations of the planar region with approximately 
$5000$, $20000$ and $80000$ vertices. The right picture in Figure~\ref{fig:plane_triangulation_input}
shows the one with about $5000$ vertices. There are a few triangles, in particular near the boundary, having 
the largest angle close to $\pi$.
\begin{figure}[!t]
\begin{center}
\begin{tabular}{cc}
\includegraphics[width=0.4\textwidth]{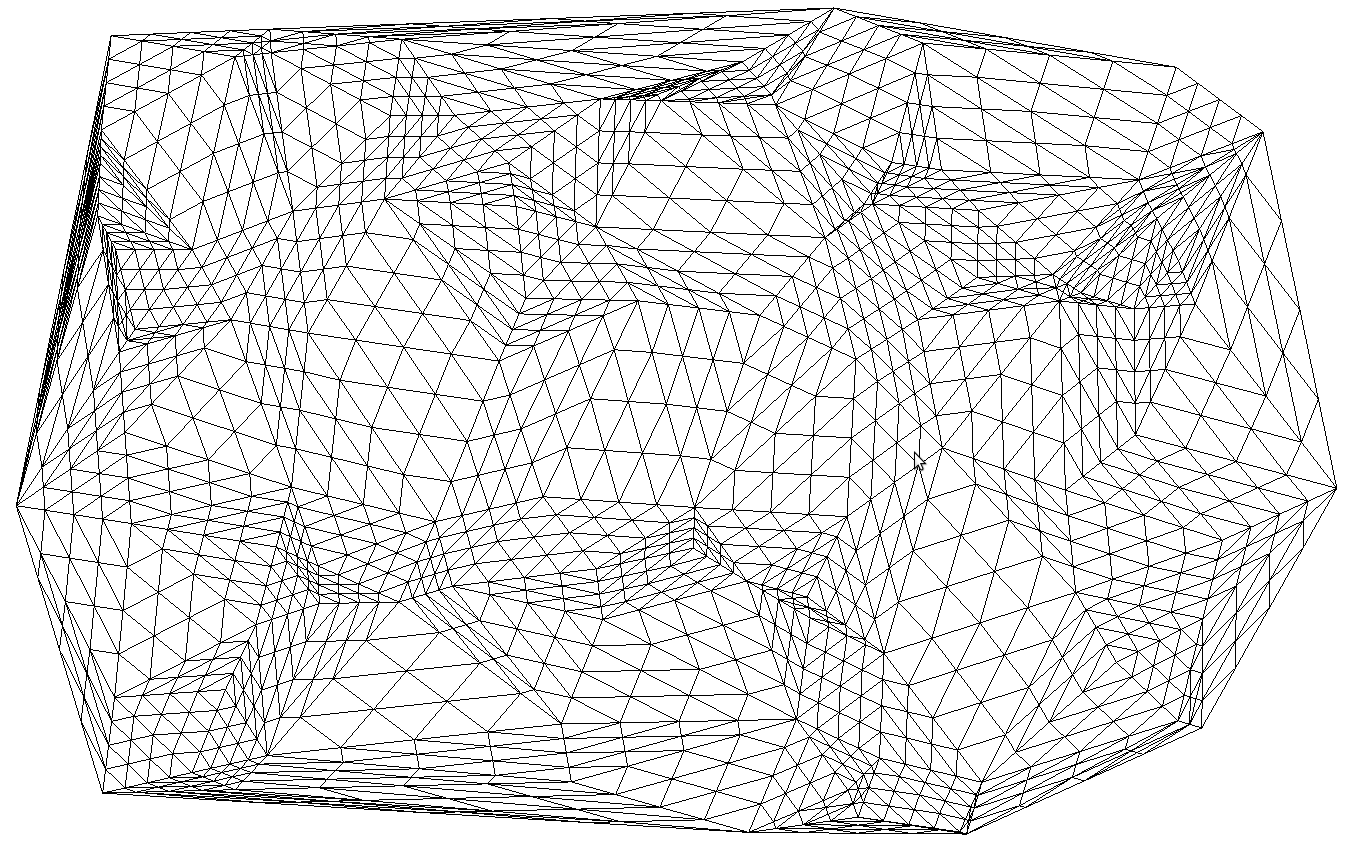} & 
\includegraphics[width=0.4\textwidth]{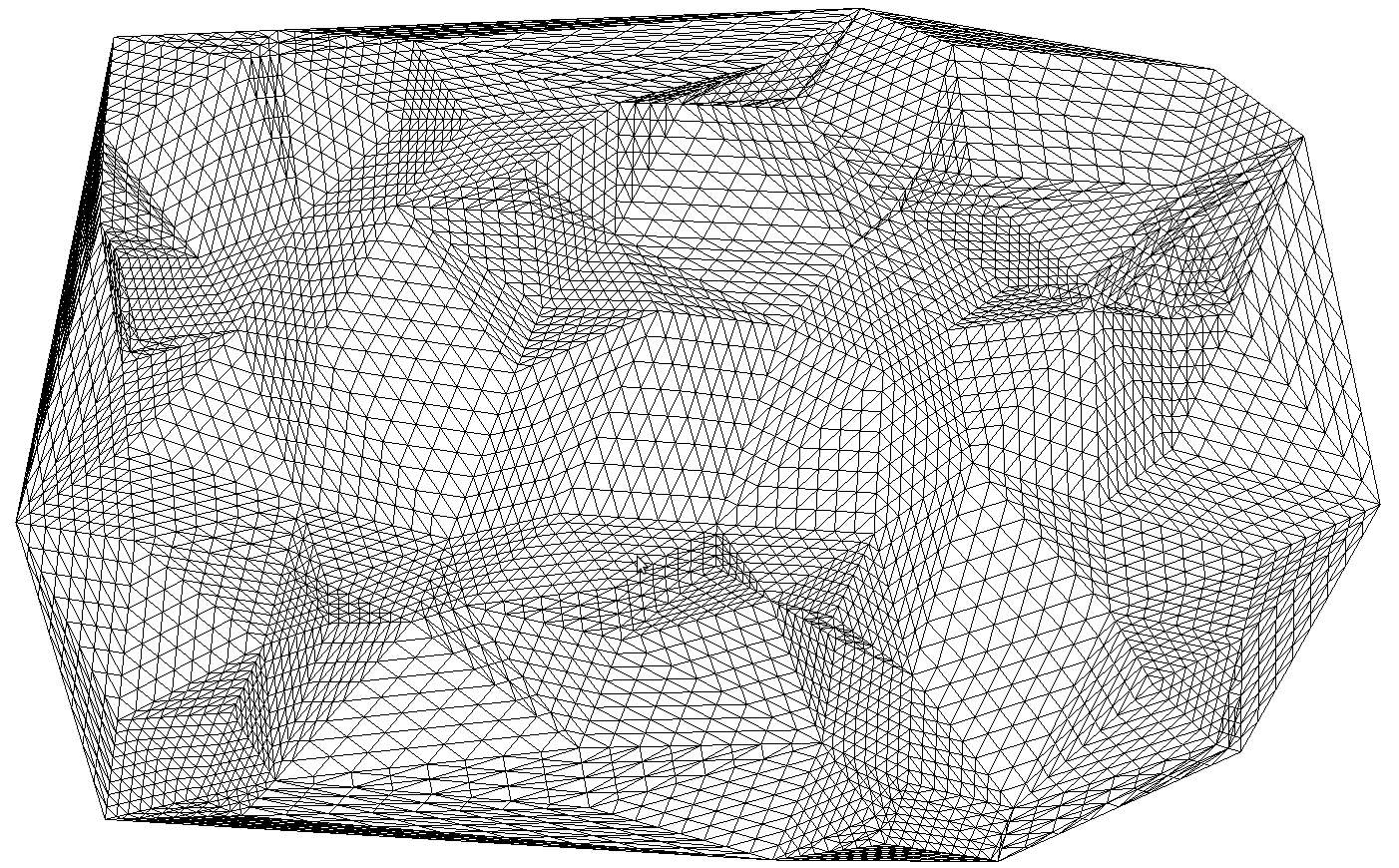} \\
\end{tabular}
\end{center}
\vspace{-0.1in}
\caption{The input triangulations of Planar Region with  $1500$ vertices (Left) 
and $5000$ vertices (Right). 
\label{fig:plane_triangulation_input}}
\end{figure}

\begin{figure}[t]
\begin{center}
\begin{tabular}{ccc}
\includegraphics[width=0.33\textwidth]{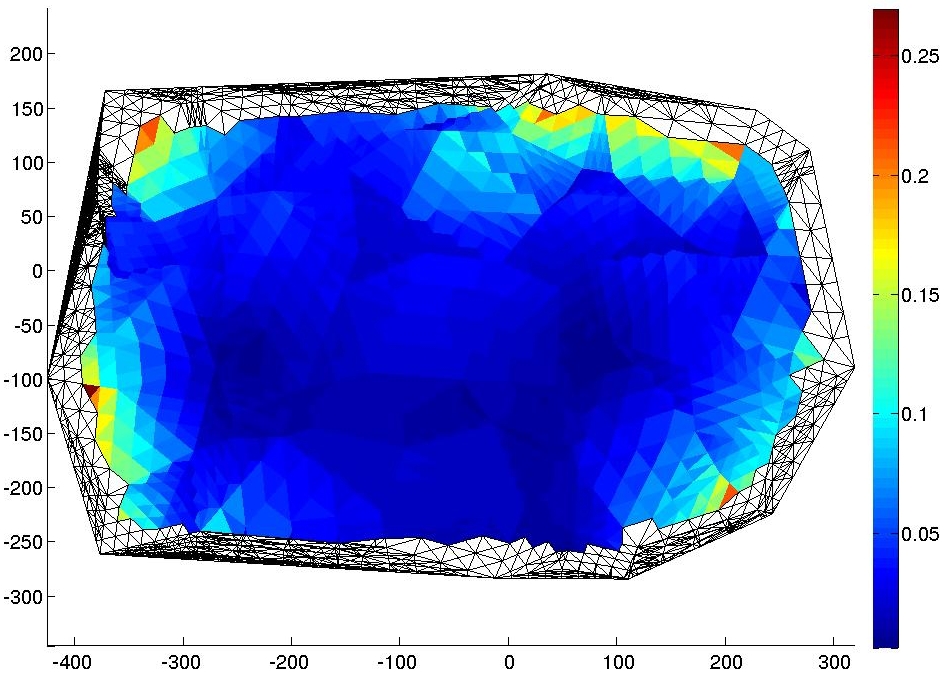} & 
\includegraphics[width=0.33\textwidth]{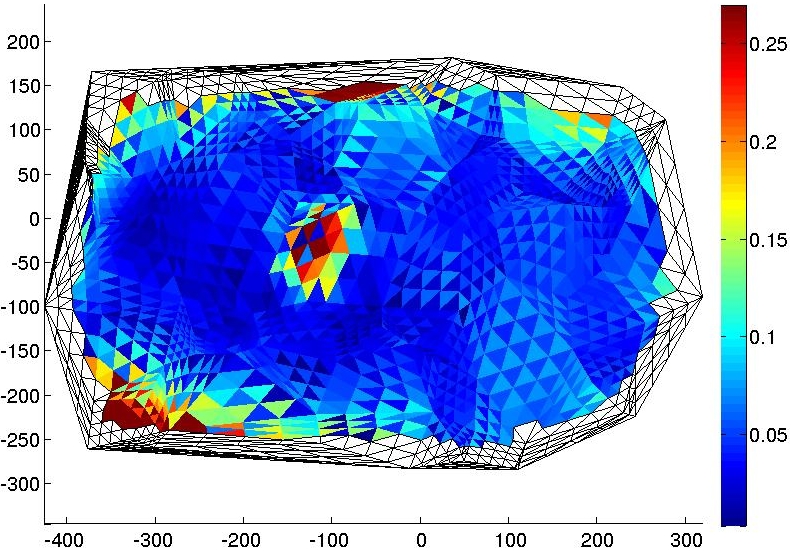} &
\includegraphics[width=0.33\textwidth]{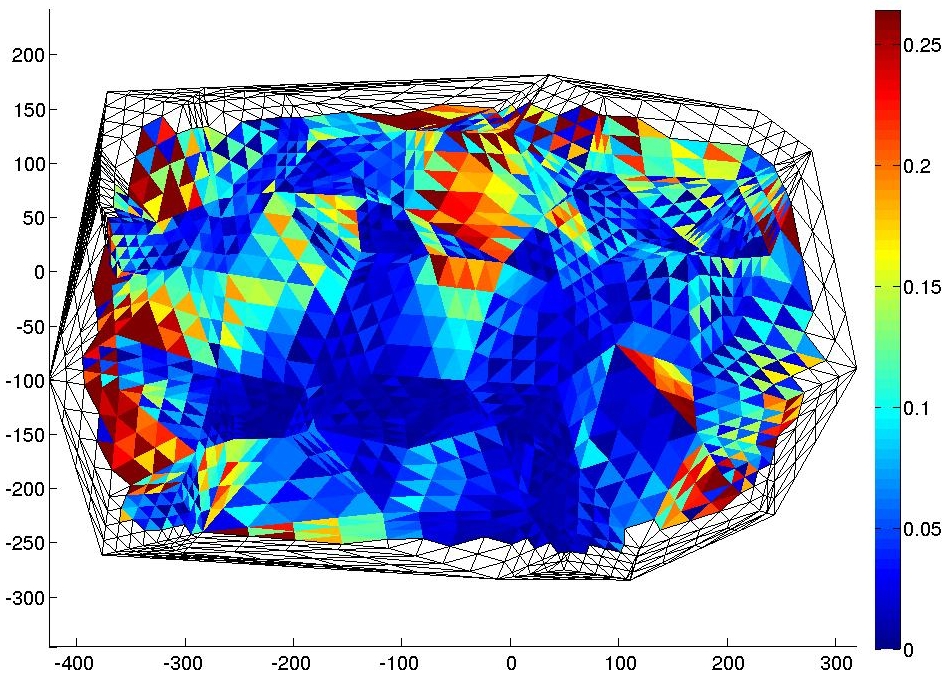} \\
{\bf DC} & {\bf HF} & {\bf BD}
\end{tabular}
\end{center}
\vspace{-0.1in}
\caption{
Results for Planar Region: the plots the conformal distortion $D(h)-1$ of the planar embedding computed by 
different methods from the triangulation with $1500$ vertices. 
For the purpose of comparison, the range of color map is fixed as $[0, 0.27]$, 
although the maximal conformal distortion of the planar embeddings computed 
by HF and BD is larger than $1.27$. 
\label{fig:plane_triangulation}}
\end{figure}

\begin{table}[!h]
\begin{center}
\begin{tabular}{| c | c | c | c | c |}
\hline
Method  & 1500 &  5000 & 20000 & 80000   \\
\hline
{\bf DC} & (0.0553, 0.2692) & (0.0286, 0.1292) & (0.0144, 0.0738) & (0.0072, 0.0401) \\
\hline
{\bf HF} & (0.0881, 0.9840 ) & (0.0473, 0.5649) & (0.0210, 0.5558) & (0.0093, 0.5570)\\
\hline
{\bf BD}& (0.1333, 1.307) & (0.0781, 1.285) & (0.0445, 1.482) & (0.0277, 1.324)\\
\hline
\multicolumn{5}{|c|}{ $(d_2, d_\infty)$} \\
\hline
{\bf DC} & 0.096 & 0.584 & 3.11 & 24.5 \\
\hline
{\bf HF} &  0.987 & 3.14 & 12.0 & 47.3 \\
\hline
{\bf BD}& 74.2 & 227  & 940 & 4204\\
\hline
\multicolumn{5}{|c|}{ timing (sec)} \\
\hline
\end{tabular}
\end{center}
\vspace{-0.1in}
\caption{Planar Region: Approximation errors and running time.
\label{tbl:plane_triangulation}
}
\end{table}

We run the aforementioned methods over these triangulations.  
The methods of DC and BD map them onto an equilateral triangle while
the method HF map them onto unit disk. In this example, we do not have 
the ground truth and thus only estimate $d_2$ and $d_\infty$ errors as 
shown in Table~\ref{tbl:plane_triangulation}. Note the errors are estimated over the 
polygons or triangles with their vertices more than $1/20$ of the diameter
of the planar region away from the boundary.
Figure~\ref{fig:plane_triangulation} shows the conformal distortion 
by different methods from the triangulation with $1500$ vertices.
As we can see, the method of DC converges linearly
in both the $d_2$ and $d_\infty$ errors, and the methods of HF and BD only converge 
in the $d_2$ error. In the embedding computed by HF, there are some triangles close to 
the boundary whose orientations get reversed. 

\begin{figure}[t]
\begin{center}
\begin{tabular}{cc}
\includegraphics[width=0.3\textwidth]{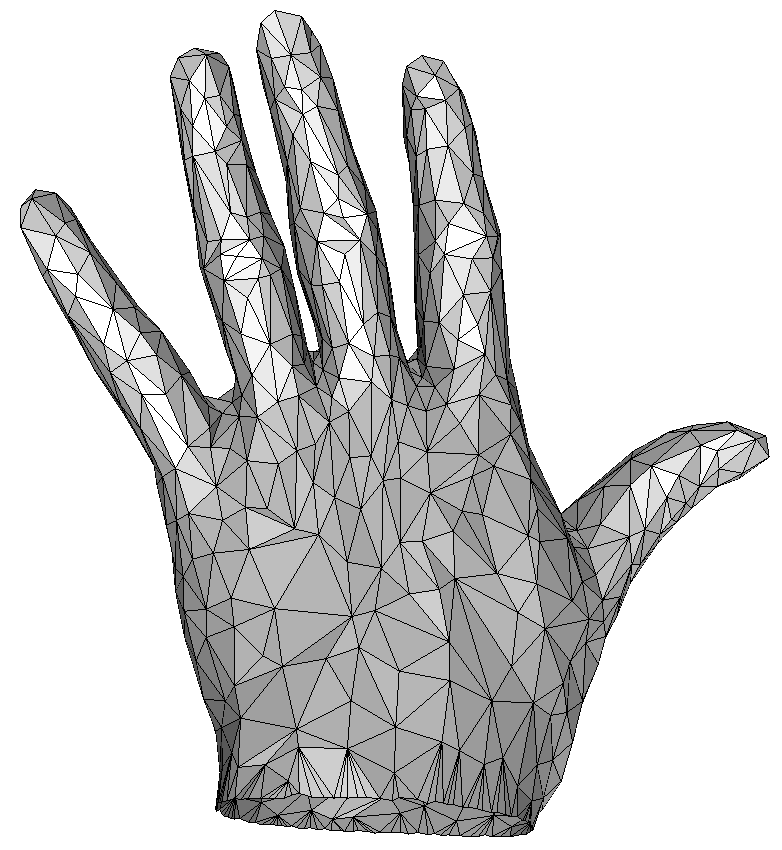} & 
\includegraphics[width=0.3\textwidth]{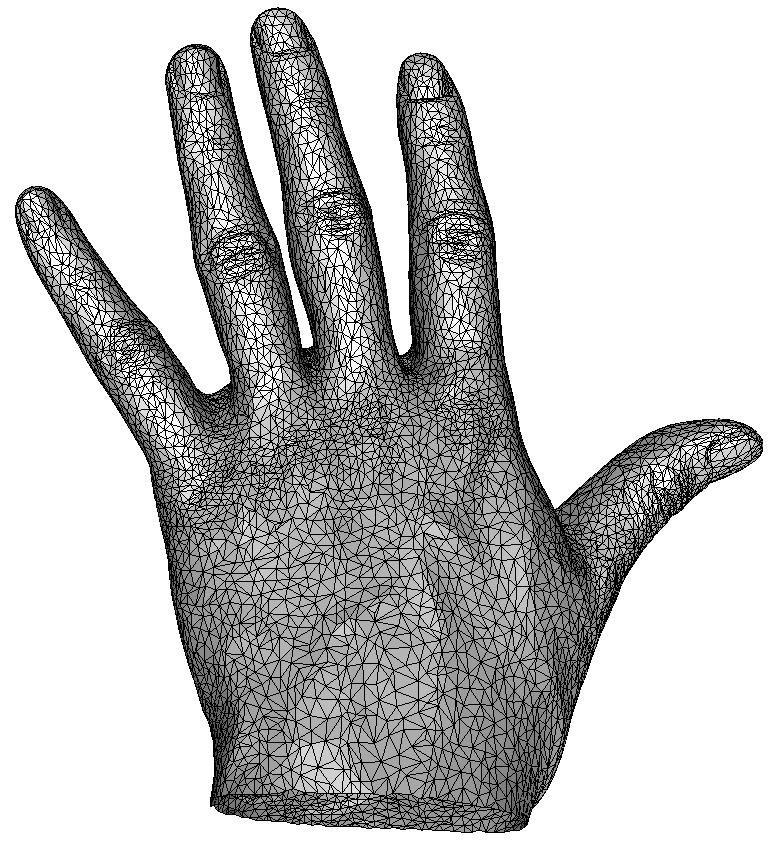} \\
\end{tabular}
\end{center}
\vspace{-0.1in}
\caption{The input triangulation of Left Hand with $800$ vertices (Left) and $10000$ vertices (Right). 
\label{fig:left_hand_input}}
\end{figure}

\begin{figure}[t]
\begin{center}
\begin{tabular}{ccc}
\includegraphics[width=0.36\textwidth]{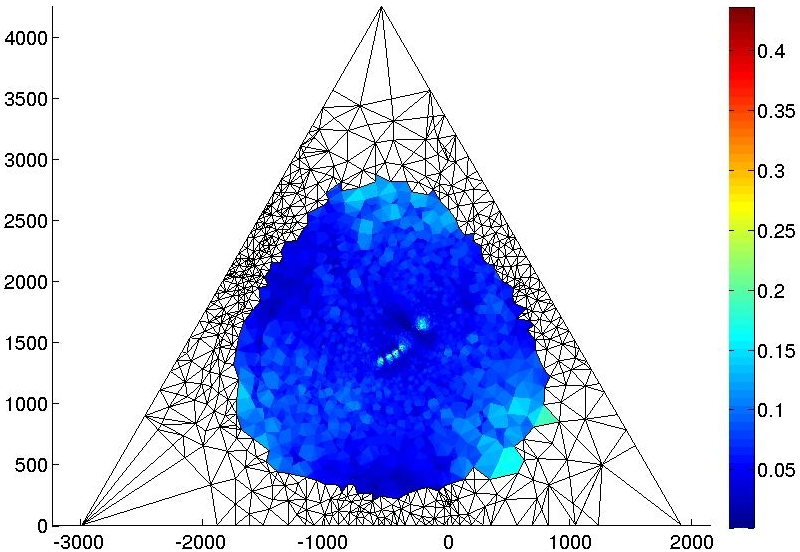} & 
\includegraphics[width=0.28\textwidth]{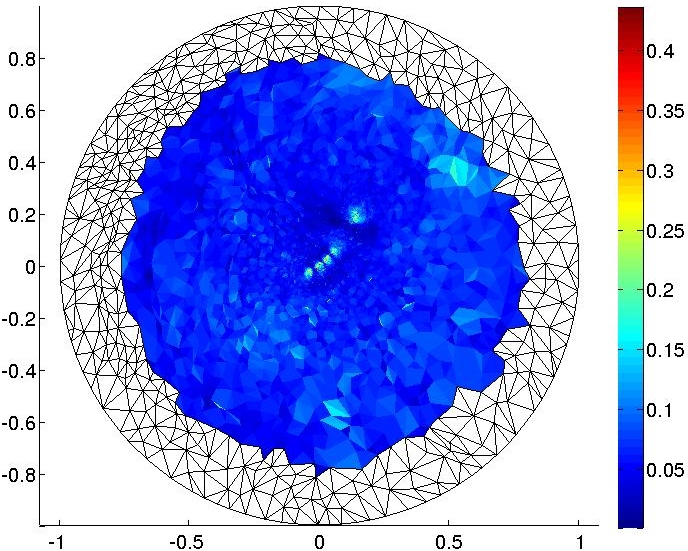} &
\includegraphics[width=0.36\textwidth]{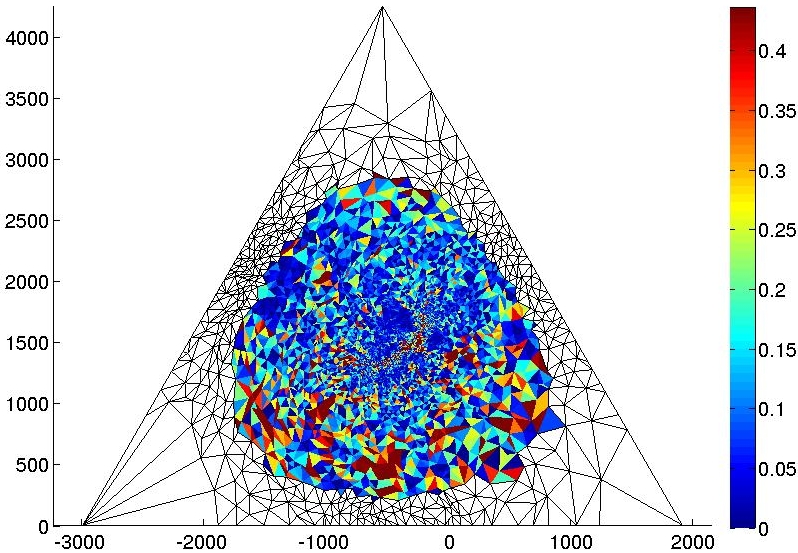} \\
{\bf DC} & {\bf HF} & {\bf BD}
\end{tabular}
\end{center}
\vspace{-0.1in}
\caption{Results for Left Hand: the plots the conformal distortion $D(h)-1$ of the 
planar embedding computed by different methods from the triangulation with $10000$ 
vertices. Note that for the purpose of comparison, the range of color map is fixed as $[0, 0.44]$, 
although the maximal conformal distortion of the planar embeddings computed 
by HF and BD is larger than $1.44$.
\label{fig:left_hand}}
\end{figure}

\begin{table}[!h]
\begin{center}
\begin{tabular}{| c | c | c | c | c |}
\hline
Method  & 800 &  2500 & 10000 & 40000   \\
\hline
{\bf DC} & (0.7436, 2.1873) & (0.3404, 1.1705) & (0.1457, 0.4359) & (0.0686, 0.2651) \\ 
\hline
{\bf HF} & ($\infty, \infty$) & (0.4247, 4.0458) & (0.1545, 1.3129) & (0.0716, 2.0045)\\
\hline
{\bf BD}& (1.1520, 6.5863) & (0.7330, 4.9571) & (0.3930, 3.9686) & (0.2873, 5.2264)\\
\hline
\multicolumn{5}{|c|}{ $(d_2, d_\infty)$} \\
\hline
{\bf DC} & 0.236 & 1.27 & 5.84 & 42.3 \\
\hline
{\bf HF} &  0.966 & 2.24& 8.52 & 37.4 \\
\hline
{\bf BD}& 38.0  &  116 & 475 & 1941\\
\hline
\multicolumn{5}{|c|}{ timing (sec)} \\
\hline
\end{tabular}
\end{center}
\vspace{-0.1in}
\caption{Left Hand: approximation errors and running time.
\label{tbl:left_hand}
}
\end{table}

\vspace{0.3in}
\noindent{\bf Left Hand.}
The model Left Hand is obtained using 3D scanning. The original model has 200k vertices, which 
is simplified using Meshlab~\cite{meshlab} to the triangulated surfaces with $800$, 
$2500$, $10000$ and $40000$ vertices. See Figure~\ref{fig:left_hand_input} for two of them. 
Again, the methods of DC and BD map these triangulated surfaces onto an equilateral triangle 
and the method HF map them onto unit disk. We estimate the $d_2$ and $d_\infty$ errors as 
shown in Table~\ref{tbl:left_hand}. Note the errors are estimated over the 
polygons or triangles with their vertices more than $1/60$ of the diameter
of the planar region away from the boundary. Figure~\ref{fig:left_hand} plots the conformal
distortions of the planar embedding computed by different methods. For a better visualization, 
in this example, we show the plots over the planar embedding. 

Again the method of DC converges linearly in both the $d_2$ and $d_\infty$ errors, and the methods 
of HF and BD only converge in the $d_2$ error. In the planar embedding of Left Hand with $800$ 
vertices computed by the method HF, there are some triangles even away from the boundary whose orientations 
get reversed. This is the reason that the corresponding $d_2$ and $d_\infty$ errors are $\infty$
in this case.


\vspace{0.1in}
\noindent{\bf Eight.}
Finally, we check the convergence for different methods over a model called Eight, which is a 
surface with genus $2$. We use Loop subdivision to subdivide a triangulated Eight with about $750$ 
vertices to obtain four more refined triangulated Eight with about $3000$, $12000$, $50000$ and $200000$ vertices. 
Figure~\ref{fig:eight_input} shows two of them. 

As the implementation of the method BD for surfaces of non disk topology is not available, we only 
show the performance of the methods of DC and HF. 
We estimate the $d_2$ and $d_\infty$ errors as shown in Table~\ref{tbl:eight}. Note the errors are
estimated over the polygons or triangles with their vertices more than one twentieth of the diameter
of Eight away from the singular vertices. 
Again, the method of DC converges linearly in both $d_2$ and $d_\infty$ errors. 
For the method of HF, the $d_2$ error decreases but the convergence rate is not clear, 
and the $d_\infty$ error does not even decrease. 
In this case, the planar embedding is in fact just an immersion and not necessary globally one-to-one. 
To visualize the conformal distortion, we plot it on the input triangulated surface as shown in Figure~\ref{fig:eight}.

\subsection{More examples and Statistics}
In this subsection, we present a few more examples and collect a few statistics
showing the performance of our algorithm. 

We run our DC method on four more examples: Maxplanck (a disk), 
Brain (a sphere), Protein (a torus), and Genus3 (a 3-hole torus). 
The results are shown in Figure~\ref{fig:more_examples}.  
The $d_2$ and  $d_\infty$ errors are estimated over the polygonal faces
whose vertices are more than one sixtieth of the diameter of the 
model away from either the boundary or the singular vertices.
Note there is no singular points for Protein since its Euler characteristic
number is $0$. For the model Genus3, we only choose one singular vertex.


In Table~\ref{tbl:statistics}, we collect the following statistics when the algorithm
runs over various examples: (1) the number of triangles in the input triangulated 
surface, labeled by {\it \#Fin}, representing the input complexity; (2) the number of 
polygonal faces in the common refinement of $T\cup T'$, labeled by {\it \#Fout}, representing
the output complexity. For surfaces with boundary, we count half of the faces
in $T\cup T'$; (3) the number of diagonal switches needed to transform the input triangulation
$T_0$ to the Delaunay triangulation $T$, labeled by {\it \#DelSW}; (4) the running time in second
of the above diagonal switches, labeled by {\it tDelSW}; (5) the number of co-circular diagonal switches performed
during the discrete conformal deformation, labeled by {\it \#CocSW}; (6) the running time in second
of the above co-circular diagonal switches, labeled by {\it tCocSW}; (7) the number of the Newton iterations, labeled by
{\it \#Newton}; (8) the running time in second of the Newton iterations excluding tCocSW, labeled by {\it tNewton}.
From Table~\ref{tbl:statistics}, we observe that the number of faces in $T\cup T'$ is often only a few hundred more
than that in $T$, and the Newtons method converges very fast and takes only 5-10 iterations, 
and the operation of diagonal switch costs very little. Note that the procedure of cocircular diagonal 
switch takes more time as it requires to find the first edge failing the Delaunay condition along the 
deforming path. 

\begin{figure}[t]
\begin{center}
\begin{tabular}{cc}
\includegraphics[width=0.4\textwidth]{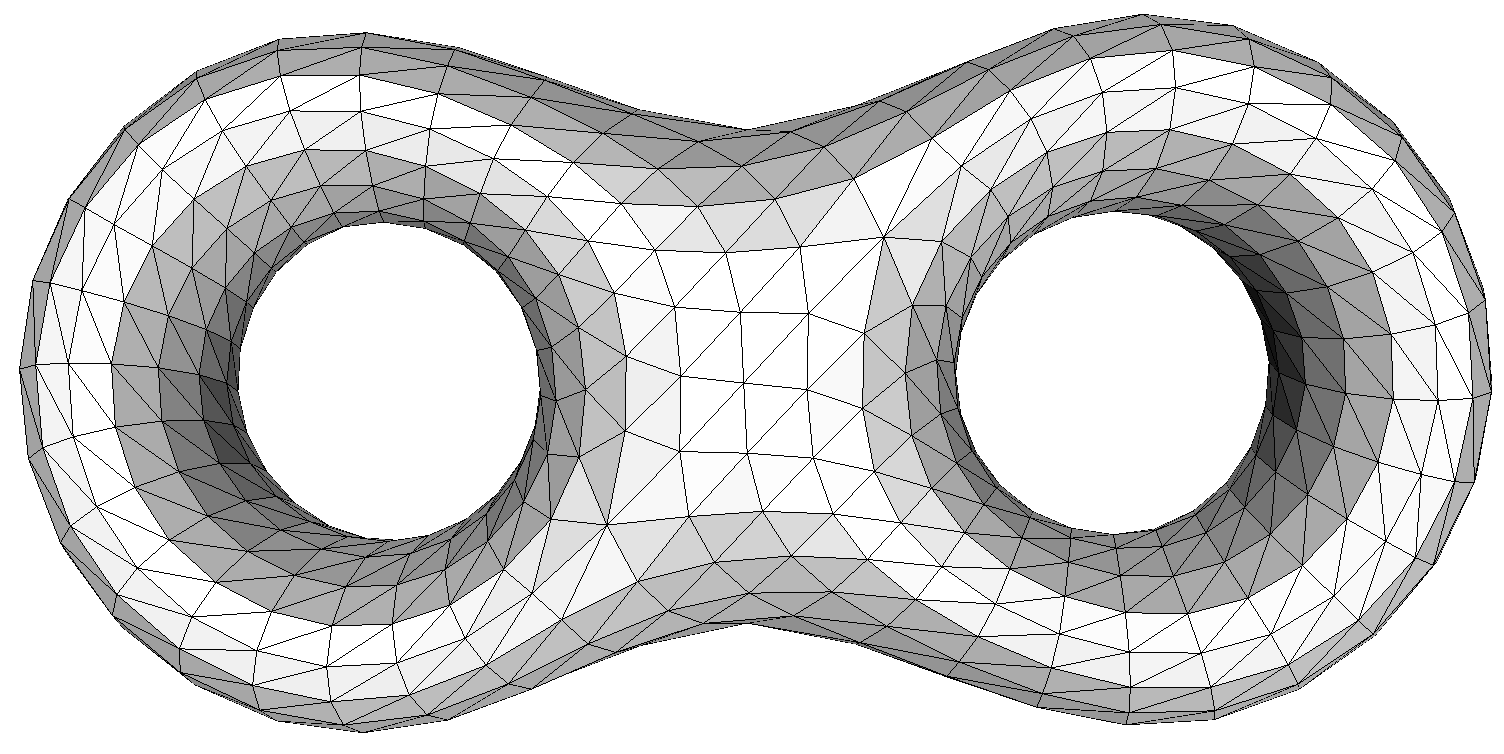} & 
\includegraphics[width=0.4\textwidth]{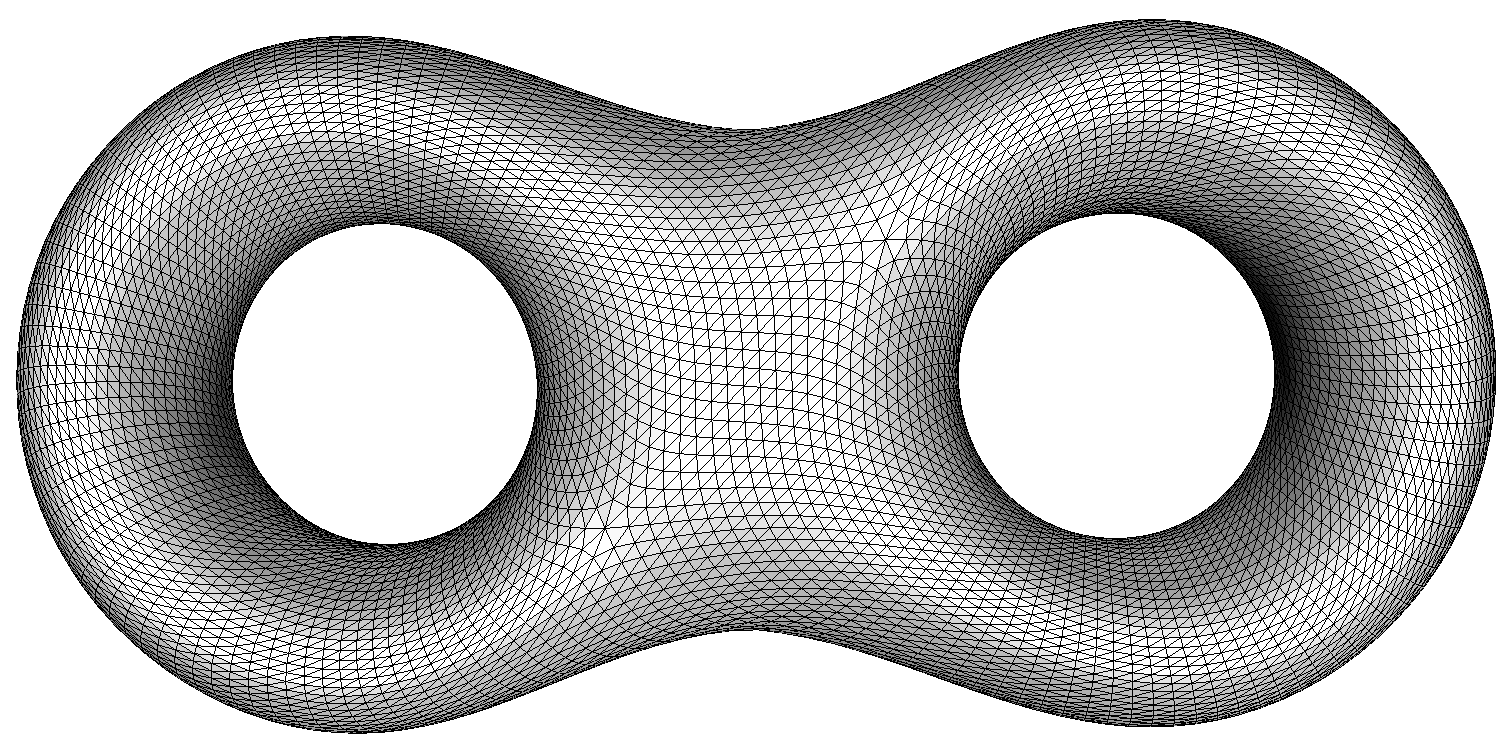} \\
\end{tabular}
\end{center}
\vspace{-0.1in}
\caption{
The input triangulations of Eight with $759$ vertices (Left) and $12000$ vertices (Right). 
\label{fig:eight_input}}
\end{figure}

\begin{figure}[t]
\begin{center}
\begin{tabular}{ccc}
\includegraphics[width=0.4\textwidth]{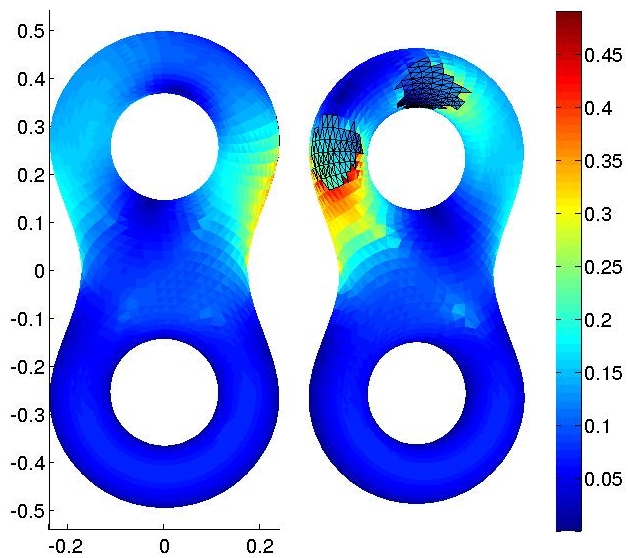} & 
\includegraphics[width=0.4\textwidth]{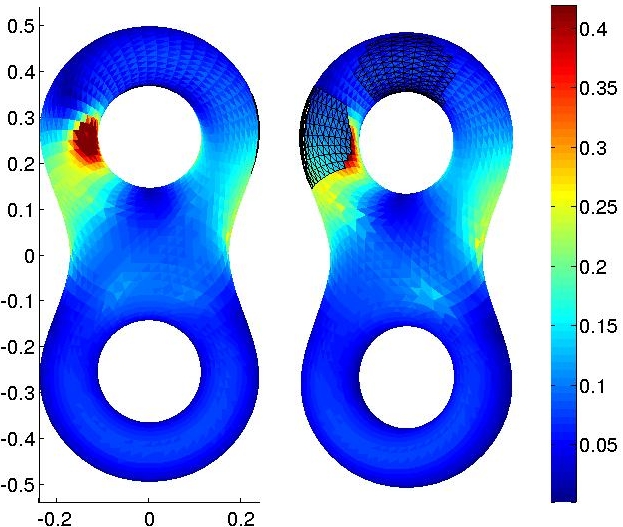} \\
{\bf DC} & {\bf HF} 
\end{tabular}
\end{center}
\vspace{-0.1in}
\caption{Results for Eight: the conformal distortion $D(h)-1$ of the 
planar embedding computed by DC and HF from the triangulation with $3000$ vertices
plotted on the input surface and shown in two different views. 
Note that for the purpose of comparison, the range of color map is fixed as $[0, 0.42]$, 
although the maximal conformal distortion of the planar embeddings computed 
by HF is larger than $1.42$.
\label{fig:eight}}
\end{figure}

\begin{table}[!h]
\begin{center}
\begin{tabular}{| c | c | c | c | c |c|}
\hline
Method  & 750 &  3000 & 12000 & 50000 & 200000 \\
\hline
{\bf DC} & (0.1422, 1.5259) & (0.1112, 0.4187) & ( 0.0290, 0.2091) & (0.0184, 0.1071) & (0.0085, 0.0361)\\
\hline
{\bf HF} & (0.1133, 0.3985) & (0.1236, 3.0128) & (0.0828, 6.2417) & (0.0325, 4.7646) & (0.0343, 14.158)\\
\hline
\multicolumn{6}{|c|}{ $(d_2, d_\infty)$} \\
\hline
{\bf DC} & 0.040 & 0.196 & 1.12 &5.88  & 59.1 \\
\hline
{\bf HF} & 1.58 & 4.91 & 19.2 & 80.6 & 339\\
\hline
\multicolumn{6}{|c|}{ timing (sec)} \\
\hline
\end{tabular}
\end{center}
\vspace{-0.1in}
\caption{Eight: Approximation errors and timing.
\label{tbl:eight}
}
\end{table}

\begin{table}[htbp]
\begin{center}
\begin{tabular}{| c | c | c | c | c | c|c|c|c|}
\hline
Model       & ~~\#Fin (k)~~ &  ~~\#Fout (k)~~ & ~~\#DelSW~~ & ~~tDelSW~~ &~~\#CocSW~~ & ~~tCocSW~~ & ~~\#Newton~~ & ~~tNewton~~   \\
\hline
{Planar Region} & 10.22	& 10.25   	  & 3223    & 0.072  & 14      & 0.272  & 5        & 0.34\\
\hline
{Planar Region} & 40.86	& 40.90   	  & 14113   & 0.38   & 17      & 1.30   & 5        & 2.01\\
\hline
{Planar Region} & 163.47	& 163.50   & 59288   & 1.50   & 18      & 5.41   & 5        & 13.53\\
\hline
{Left Hand} & 79.98	& 80.54   	  & 14192   & 0.52   & 272     & 30.6   & 7        & 4.97\\
\hline
{Eight} & 196.61	& 196.65   	  & 68428   & 0.34   & 44     & 13.84   & 5        & 43.1\\
\hline
{Maxplanck} & 47.08 		& 47.27   	  & 5969    & 0.248  & 93      & 4.53   & 6        & 2.22\\
\hline
{Brain} 	 	& 15.00	& 15.30   	  & 1048   & 0.016   & 300     & 4.6   & 6        & 1.05\\
\hline
{Protein} 	 	& 46.24	&  46.37  	  & 3064   & 0.072   & 132     & 8.07   & 9        & 5.34\\
\hline
{Genus3} 	 	& 26.62	&  26.74  	  & 12454   & 0.052   & 111     & 4.33   & 10        & 4.56\\
\hline
\end{tabular}
\end{center}
\caption{Statistics.
\label{tbl:statistics}
}
\end{table}

\begin{figure}[!t]
\begin{center}
\begin{tabular}{cc}
\includegraphics[width=0.2\textwidth]{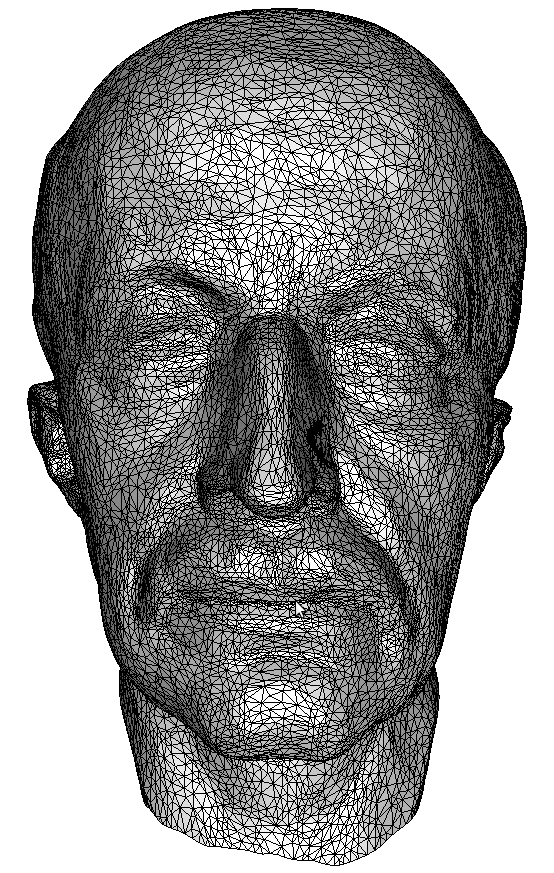} & 
\includegraphics[width=0.38\textwidth]{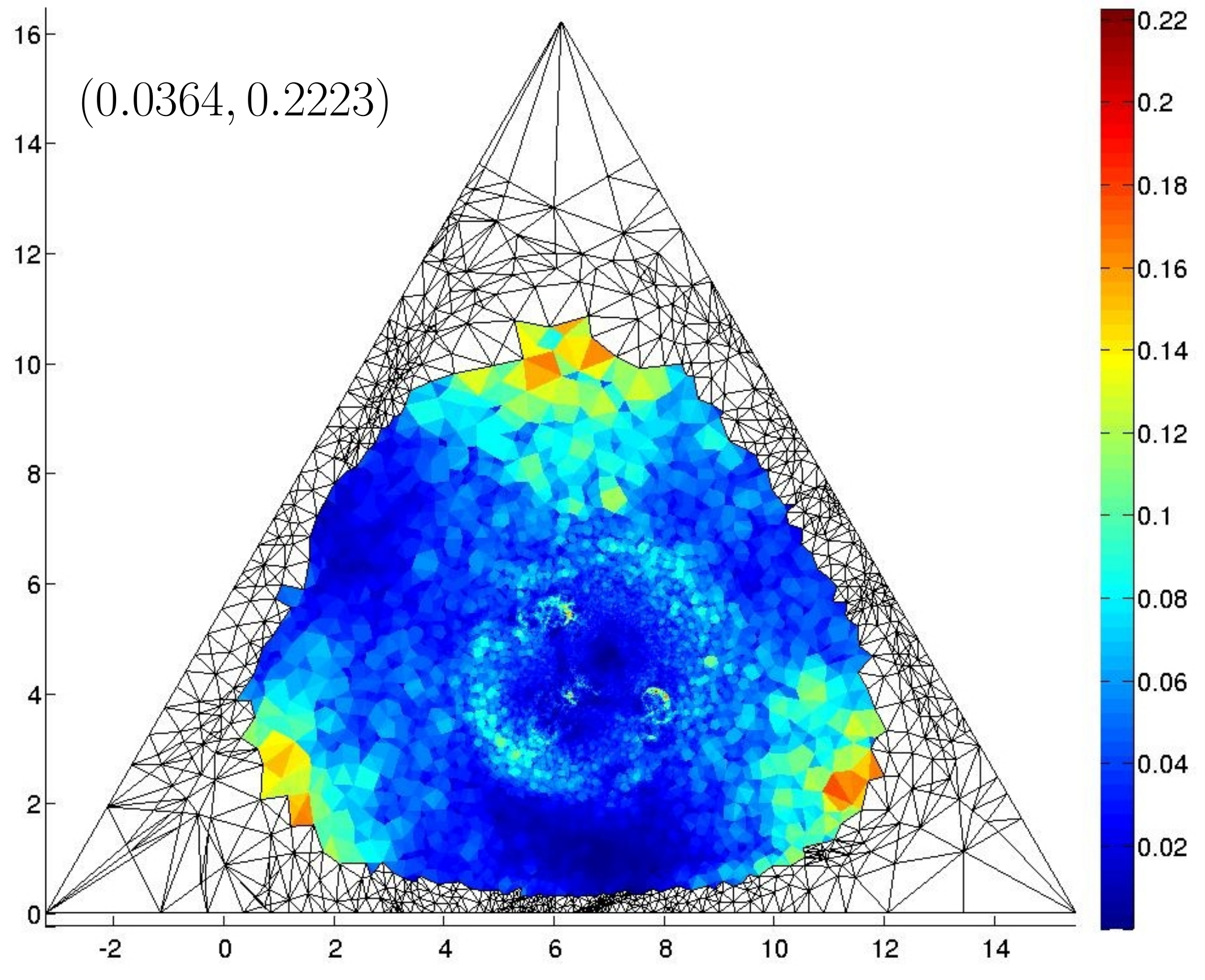} \\
\includegraphics[width=0.28\textwidth]{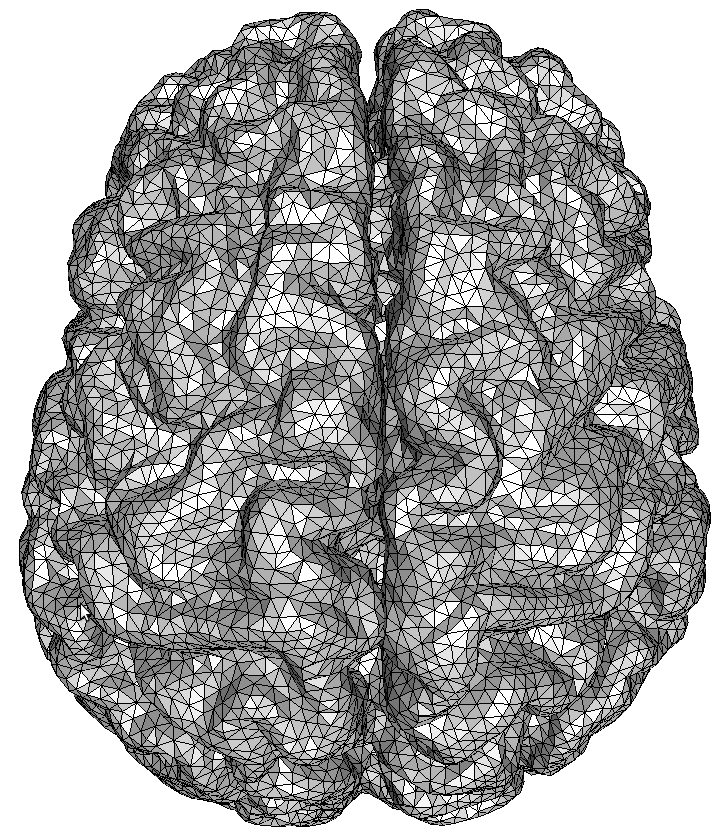} & 
\includegraphics[width=0.38\textwidth]{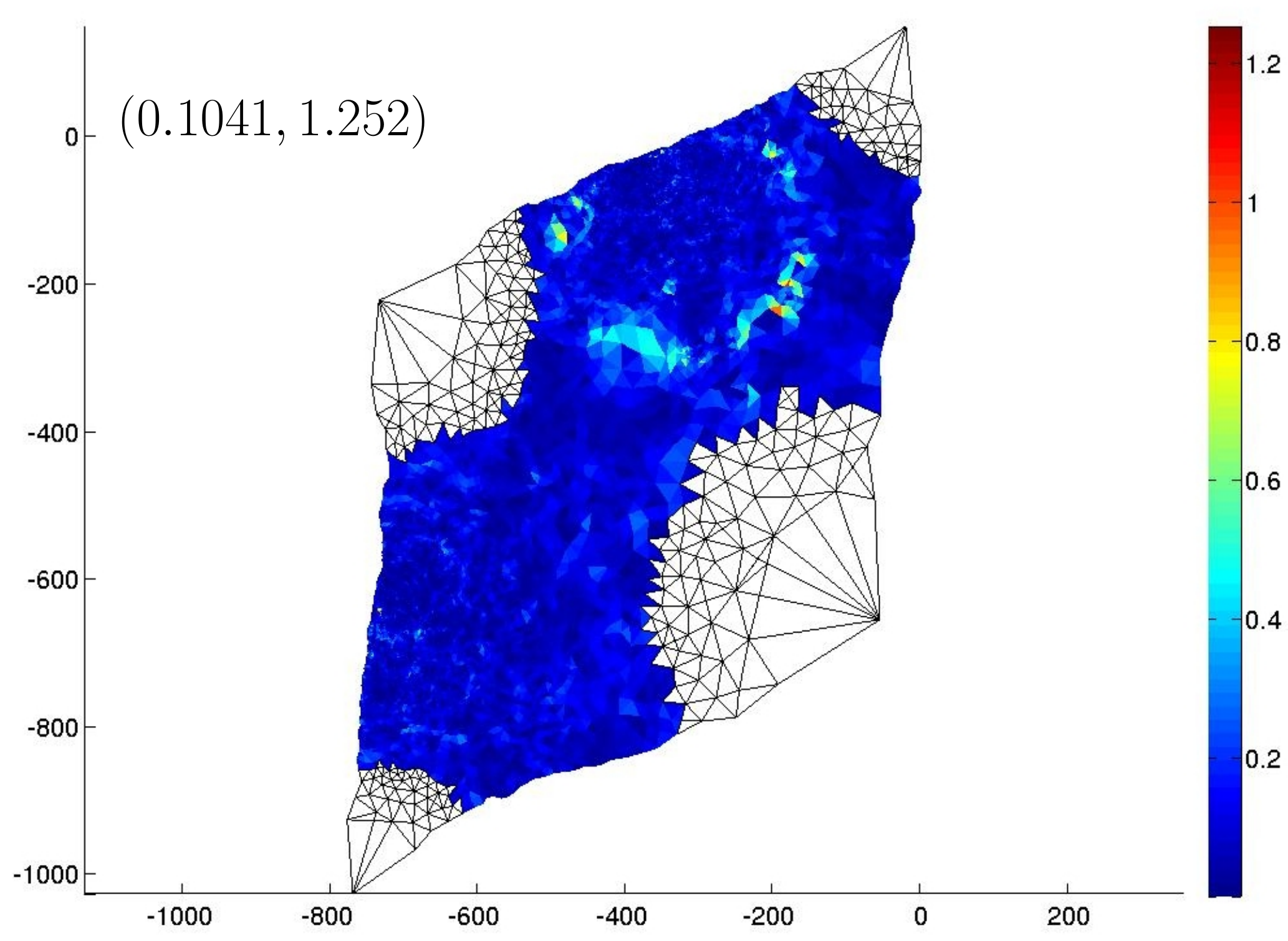} \\
\includegraphics[width=0.4\textwidth]{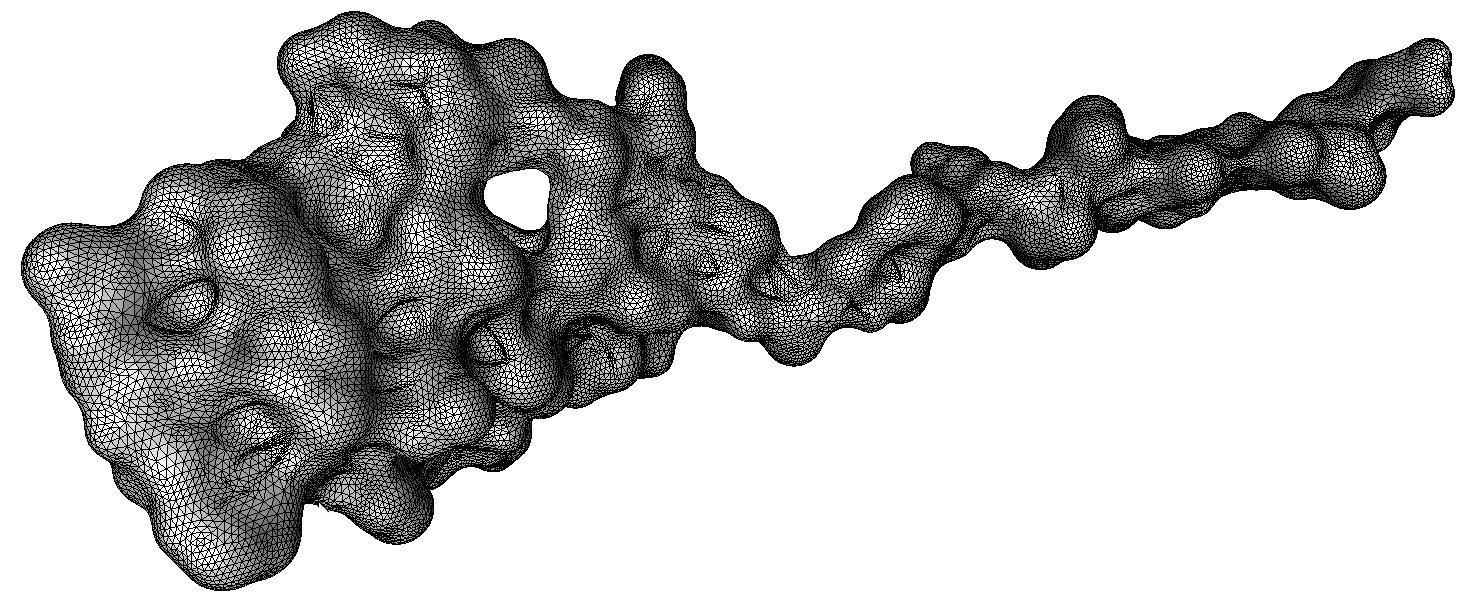} & 
\includegraphics[width=0.38\textwidth]{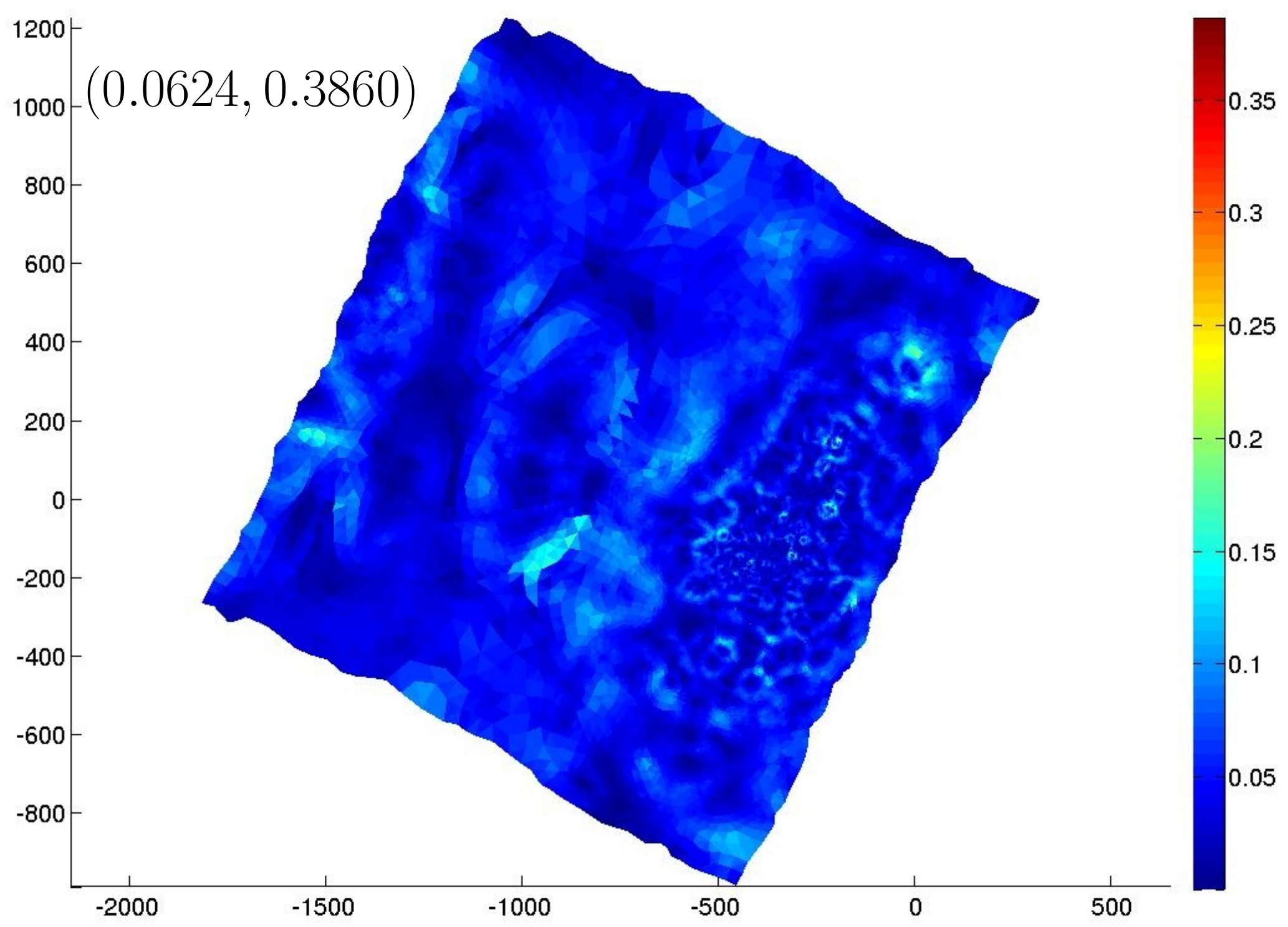} \\
\includegraphics[width=0.35\textwidth]{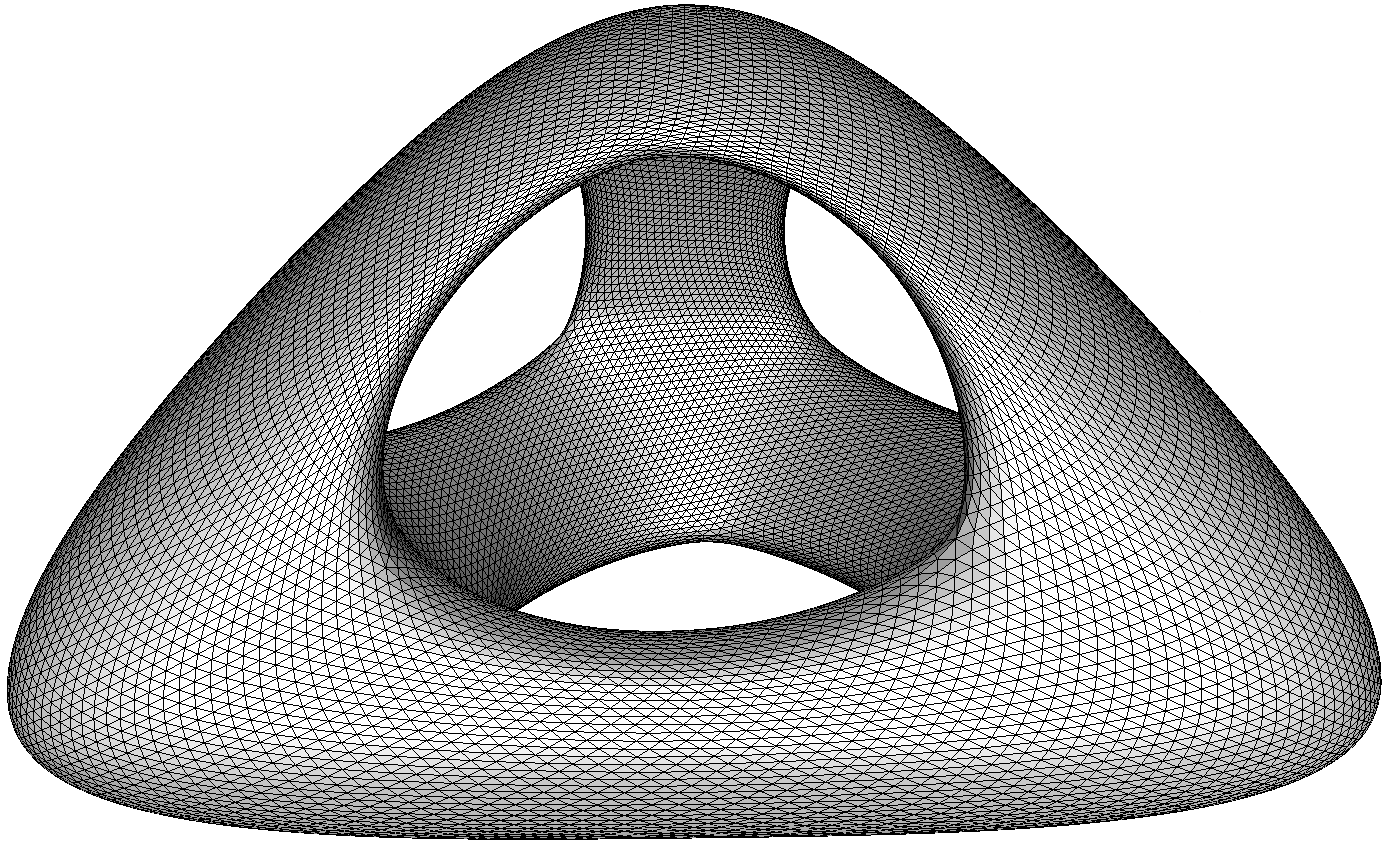} & 
\includegraphics[width=0.39\textwidth]{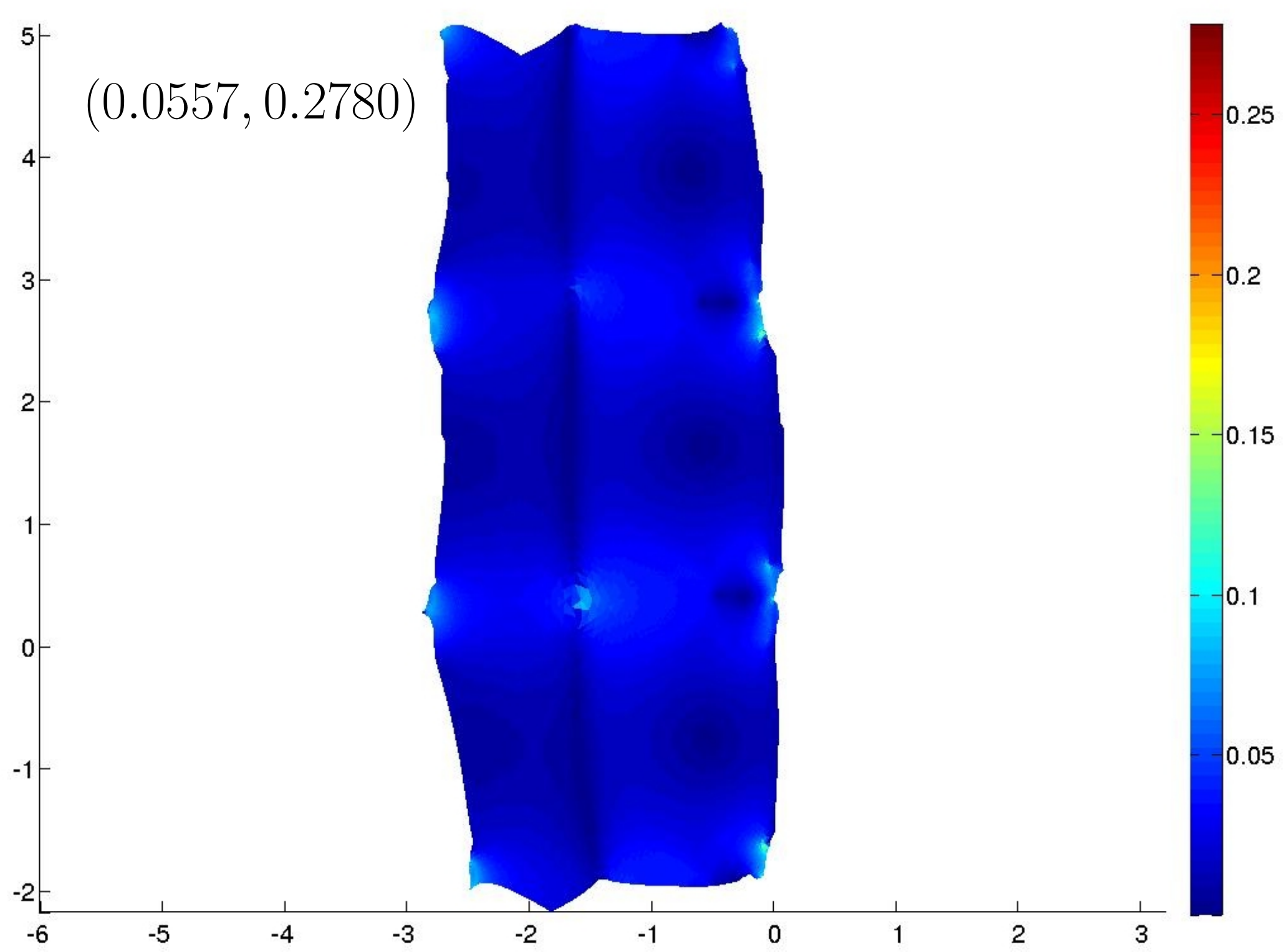} \\
\end{tabular}
\end{center}
\caption{Left: the input triangulated surfaces; Right: the planar embedding 
computed by our algorithm. The color maps plot $D(h) - 1$ and the pairs of numbers
are the $(d_2, d_\infty)$ errors. 
\label{fig:more_examples}}
\end{figure}

\subsection{Remark}
From the above experiments, we observe that our conformality numerically converges
to the classical one as the triangle size goes to $0$, in a linear rate. In particular, 
this convergence behavior is independent of the quality of the triangles in the
input triangulated surfaces. It is also very efficient comparing to the state of the art. 
One disadvantage of our method is that we may subdivide the input triangles into the
polygonal pieces to accurately represent the discrete conformal map. This may increase
the complexity of the output by our algorithm. However, from the statistics we show,
the increase of the complexity is very little for most of the prescribed curvatures.

\section{Conclusion}
We have introduced a new discrete conformality for triangulated surfaces possibly
with boundary, and showed a discrete uniformization theorem with this conformality, 
and described an algorithm for solving the problem of prescribing curvature, and 
explicitly constructed a discrete conformal map between the input triangulated surface
and the deformed triangulated surface.  
In addition, we have presented the numerical examples to show the convergence of 
our discrete conformality and to demonstrate the efficiency and the robustness
of our algorithm.

We point out a few possible few directions for future research. In~\cite{gglsw}, we 
have presented a similar discrete conformality for hyperbolic triangulated surfaces and 
a discrete uniformization theorem associated to it. We
plan to develop an algorithm based on this discrete conformality to solve the problem of 
prescribing curvature for hyperbolic triangulated surfaces. Hyperbolic triangulation is 
more natural for surfaces with genus bigger than $1$ as their fundamental domains can be 
flattened into hyperbolic plane without choosing any singular vertices.  
It remains open whether our discrete conformality and uniformization theorem can be extended
to spherical triangulated surfaces, which is definitely worth investigating in the future.  
Another interesting future research is to see if diagonal switches can be used in 
inversive distance circle patterns.

\section{Acknowledgment} 
We would like to thank Yaron Lipman for the implementation of the method BD, 
and Xianfeng Gu and S-T. Yau for the implementation of the method HF. 
The work is supported in part by the NSF of USA and the NSF of China.

\bibliographystyle{abbrv}
\bibliography{discrete_conformality}

\begin{thebibliography}{10}

\bibitem{meshlab}
{\em {MeshLab}}.
\newblock http://meshlab.sourceforge.net/.

\bibitem{Amenta00asimple}
N.~Amenta, S.~Choi, T.~K. Dey, and N.~Leekha.
\newblock A simple algorithm for homeomorphic surface reconstruction.
\newblock In {\em International Journal of Computational Geometry and
  Applications}, pages 213--222, 2000.

\bibitem{mosek}
E.~Andersen and K.~Andersen.
\newblock The mosek interior point optimizer for linear programming: An
  implementation of the homogeneous algorithm.
\newblock In H.~Frenk, K.~Roos, T.~Terlaky, and S.~Zhang, editors, {\em High
  Performance Optimization}, volume~33 of {\em Applied Optimization}, pages
  197--232. Springer US, 2000.

\bibitem{bps}
A.~Bobenko, U.~Pinkall, and B.~Springborn.
\newblock Discrete conformal maps and ideal hyperbolic polyhedra.
\newblock {\em arXiv:1005.2698}.

\bibitem{Bobenko07}
A.~Bobenko and B.~Springborn.
\newblock A discrete laplace-beltrami operator for simplicial surfaces.
\newblock {\em Discrete Comput. Geom.}, 38(4):740--756, 2007.

\bibitem{Bowers}
P.~L. Bowers and K.~Stephenson.
\newblock Uniformizing dessins and {B}ely\u\i\ maps via circle packing.
\newblock {\em Mem. Amer. Math. Soc.}, 170(805):xii+97, 2004.

\bibitem{chow}
B.~Chow and F.~Luo.
\newblock Combinatorial {R}icci flows on surfaces.
\newblock {\em J. Differential Geom.}, 63(1):97--129, 2003.

\bibitem{verdiere}
Y.~C. de~Verdi\'{e}re.
\newblock Un principe variationnel pour les empilements de cercles.
\newblock {\em Invent. Math.}, 104:655--669, 1991.

\bibitem{driscoll}
T.~A. Driscoll.
\newblock {\em The Schwarz-Christoffel Toolbox for MATLAB}.
\newblock http://www.tobydriscoll.net/SC/index.html.

\bibitem{Fisher:2006}
M.~Fisher, B.~Springborn, A.~I. Bobenko, and P.~Schroder.
\newblock An algorithm for the construction of intrinsic delaunay
  triangulations with applications to digital geometry processing.
\newblock In {\em ACM SIGGRAPH 2006 Courses}, SIGGRAPH '06, pages 69--74, New
  York, NY, USA, 2006. ACM.

\bibitem{Glickenstein1}
D.~Glickenstein.
\newblock A combinatorial yamabe flow in three dimensions.
\newblock {\em Topology}, 44:791--808, 2005.

\bibitem{Glickenstein2}
D.~Glickenstein.
\newblock A maximum principle for combinatorial yamabe flow.
\newblock {\em Topology}, 44:809--825, 2005.

\bibitem{gglsw}
X.~Gu, R.~Guo, F.~Luo, J.~Sun, and T.~Wu.
\newblock A discrete uniformization theorem for polyhedral surfaces ii.
\newblock {\em arXiv:1401.4594}.

\bibitem{glsw2}
X.~Gu, F.~Luo, J.~Sun, and T.~Wu.
\newblock Discrete conformal mapping for polyhedral surfaces.
\newblock {\em manuscript available}.

\bibitem{glsw1}
X.~Gu, F.~Luo, J.~Sun, and T.~Wu.
\newblock A discrete uniformization theorem for polyhedral surfaces.
\newblock {\em arXiv:1309.4175}.

\bibitem{Gu:2003}
X.~Gu and S.-T. Yau.
\newblock Global conformal surface parameterization.
\newblock In {\em Proceedings of the 2003 Eurographics/ACM SIGGRAPH Symposium
  on Geometry Processing}, SGP '03, pages 127--137, Aire-la-Ville, Switzerland,
  Switzerland, 2003. Eurographics Association.

\bibitem{gu:2008}
X.~Gu and S.-T. Yau.
\newblock {\em Computational Conformal Geometry (volume 3 of the Advanced
  Lectures in Mathematics series)}.
\newblock International Press of Boston; Har/Cdr edition (July 1, 2008), 2008.

\bibitem{Guo}
R.~Guo.
\newblock Local rigidity of inversive distance circle packing.
\newblock {\em Trans. Amer. Math. Soc.}, 363(9):4757--4776, 2011.

\bibitem{Hatcher}
A.~Hatcher.
\newblock On triangulations of surfaces.
\newblock {\em Topology Appl}, 40:189--194, 1991.

\bibitem{he98}
Z.-X. He and O.~Schramm.
\newblock The {$C^\infty$}-convergence of hexagonal disk packings to the
  {R}iemann map.
\newblock {\em Acta Math.}, 180(2):219--245, 1998.

\bibitem{Jin}
M.~Jin, J.~Kim, F.~Luo, and X.~Gu.
\newblock Discrete surface ricci flow.
\newblock {\em {IEEE} Trans. Vis. Comput. Graph.}, 14(5):1030--1043, 2008.

\bibitem{Kharevych}
L.~Kharevych, B.~Springborn, and P.~Schr\"{o}der.
\newblock Discrete conformal mappings via circle patterns.
\newblock {\em ACM Trans. Graph.}, 25(2):412--438, Apr. 2006.

\bibitem{Levy}
B.~L{\'e}vy, S.~Petitjean, N.~Ray, and J.~Maillot.
\newblock Least squares conformal maps for automatic texture atlas generation.
\newblock {\em ACM Trans. Graph.}, 21(3):362--371, July 2002.

\bibitem{Lipman12}
Y.~Lipman.
\newblock Bounded distortion mapping spaces for triangular meshes.
\newblock {\em {ACM} Trans. Graph.}, 31(4):108, 2012.

\bibitem{Lui}
L.~M. Lui, K.~C. Lam, S.~Yau, and X.~Gu.
\newblock Teichmuller mapping (t-map) and its applications to landmark matching
  registration.
\newblock {\em {SIAM} J. Imaging Sciences}, 7(1):391--426, 2014.

\bibitem{luo}
F.~Luo.
\newblock Combinatorial yamabe flow on surfaces.
\newblock {\em Commun. Contemp. Math.}, 6(5):765--780, 2004.

\bibitem{Luo2011}
F.~Luo.
\newblock Rigidity of polyhedral surfaces, iii.
\newblock {\em Geometry \& Topology}, 15:2299–2319, 2011.

\bibitem{pinkall}
U.~Pinkall and K.~Polthier.
\newblock Computing discrete minimal surfaces and their conjugates.
\newblock {\em Experimental Mathematics}, 2(1):15--36, 1993.

\bibitem{rivin}
I.~Rivin.
\newblock Euclidean structures on simplicial surfaces and hyperbolic volume.
\newblock {\em Ann. of Math}, 139(3):553--580, 1994.

\bibitem{Rocek}
M.~Ro$\check{c}$ek and R.~M. Williams.
\newblock The quantization of regge calculus.
\newblock {\em Zeitschrift f\"ur Physik C Particles and Fields},
  21(4):371--381, 1984.

\bibitem{RS}
B.~Rodin and D.~Sullivan.
\newblock The convergence of circle packings to the {R}iemann mapping.
\newblock {\em J. Differential Geom.}, 26(2):349--360, 1987.

\bibitem{ssp}
B.~Springborn, P.~Schr\"{o}der, and U.~Pinkall.
\newblock Conformal equivalence of triangle meshes.
\newblock {\em ACM Trans. Graph.}, 27(3):77:1--77:11, Aug. 2008.

\bibitem{Stephenson03}
K.~Stephenson.
\newblock Circle packing: a mathematical tale.
\newblock {\em Notices Amer. Math. Soc}, 50:1376--1388, 2003.

\end{thebibliography}

\vspace{8mm}
\noindent\textbf{Appendix: Proof of Theorem~\ref{theorem:dc_boundary}}

\begin{proof}
Denote $(\tilde{S}, \tilde{V})$ the doubled surface of $(S, V)$ and $\tilde{d}$ the doubled metric
of $d$. Prescribe the curvature $\tilde{K}^*$ for $(\tilde{S}, \tilde{V})$ by setting 
$\tilde{K}^*(v) = 2 * K^*(v)$ for a vertex $v$ on the boundary, and $\tilde{K}^*(v) =  K^*(v)$ for a vertex $v$ in the interior.  
It is easy to verify that the curvature $\tilde{K}^*$ satisfies the hypotheses in Theorem~\ref{thm:main} imposed 
on a prescribed curvature on $ (\tilde{S}, \tilde{V})$. Thus there exists a PL metric $\tilde{d}'$ discrete
conformal to $\tilde{d}$ and the discrete curvature of $\tilde{d}'$ is the curvature of $\tilde{K}^*$.
We will show that $\tilde{d}'$ respects the doubling structure and the restriction of $\tilde{d}'$ onto
$S$ is the PL metric $d'$ with the property stated in the theorem. 

We first show for a PL metric $dd$ on $(\tilde{S}, \tilde{V})$ respecting the doubling structure, 
there is a Delaunay triangulation $T$ in $dd$ which has the following symmetric property: 
(1) Any triangle $f$ in $T$ not crossing the boundary has an identical mirror triangle $f'=h(f)$ in $T$; 
(2) Let $F_s(T)$ be set of triangles in $T$ crossing a segment $s\in B$. Then any triangle $f = uu'v \in F_s(T)$ 
must have a pair of vertices $u, u'$ with $u' = h(u)$, and moreover, if the third vertex $v$ of $f$ is not
the endpoints of the segment $s$, the neighboring triangle $f' = v'vu' \in F_s(T)$ has the property that $v'=h(v)$. 
Note that the quadrilateral $f\cup f'$ must be cocircular as the segment $s$ is the common bisector 
of the edge $uu'$ and $vv'$. See Figure~\ref{fig:boundary_triangles} for an illustration. 

Since one can reach a Delaunay triangulation starting from any triangulation by diagonally
switching the edges which fail to be Delaunay finite many times, we can prove this by induction
on the number of diagonal switches. We start with the triangulation $T_0$ on $(\tilde{S}, \tilde{V})$ 
so that the restrictions of $T_0$ onto both copies of $S$ are identical triangulations. Note a segment 
$s\in B$ must be an edge in $T_0$. Thus $F_s(T_0)$ is empty and the symmetric property trivially holds.
Assume by diagonally switching a set of edges which fails to be Delaunay, we reach a triangulation $T_k$ 
satisfying the symmetric property. Assume there is an edge $e\in T_k$ which fails to be Delaunay. If $e$ 
is not a side of any triangle in $F_s(T_k)$ for any segment $s\in B$, then its mirror $h(e)$ also fails 
to be Delaunay. Note that if the edge $e$ itself is a segment on $B$, then $h(e) = e$. Switch both $e$ 
and $h(e)$ and reach a triangulation $T_{k+1}$ which satisfies the symmetric property. If $e$ is a side of a 
triangle in $F_s(T_k)$ for some segment $s\in B$, there are two cases: (i) $e$ crosses $s$; 
and (ii) $e$ does not cross $e$. In the case (i), the endpoints $z, z'$ of $e$ must satisfy
$z'=h(z)$ and any edge in the triangles incident to $e$ which crosses $s$ must also fail to be Delaunay.  
For example, as shown in Figure~\ref{fig:switch_boundary_edge}, the edges $zu'$ and $z'v$ must also fail 
to be Delaunay. Switch these edges and reach a triangulation $T_{k+1}$ which satisfies the symmetric property. 
In the second case, switch both $e$ and $e'=h(e)$. If the endpoints of $e$ contain no endpoints of the segment 
$s$, as shown in Figure~\ref{fig:switch_boundary_edge}, switch the diagonal $vu'$ as it must also fail to be Delaunay. 
The resulting triangulation $T_{k+1}$ satisfies the symmetric property. This proves that there is a 
Delaunay triangulation $T$ in $dd$ satisfying the above symmetric property. 

\begin{figure}[!t]
\begin{center}
\begin{tabular}{c}
\includegraphics[width=0.6\textwidth]{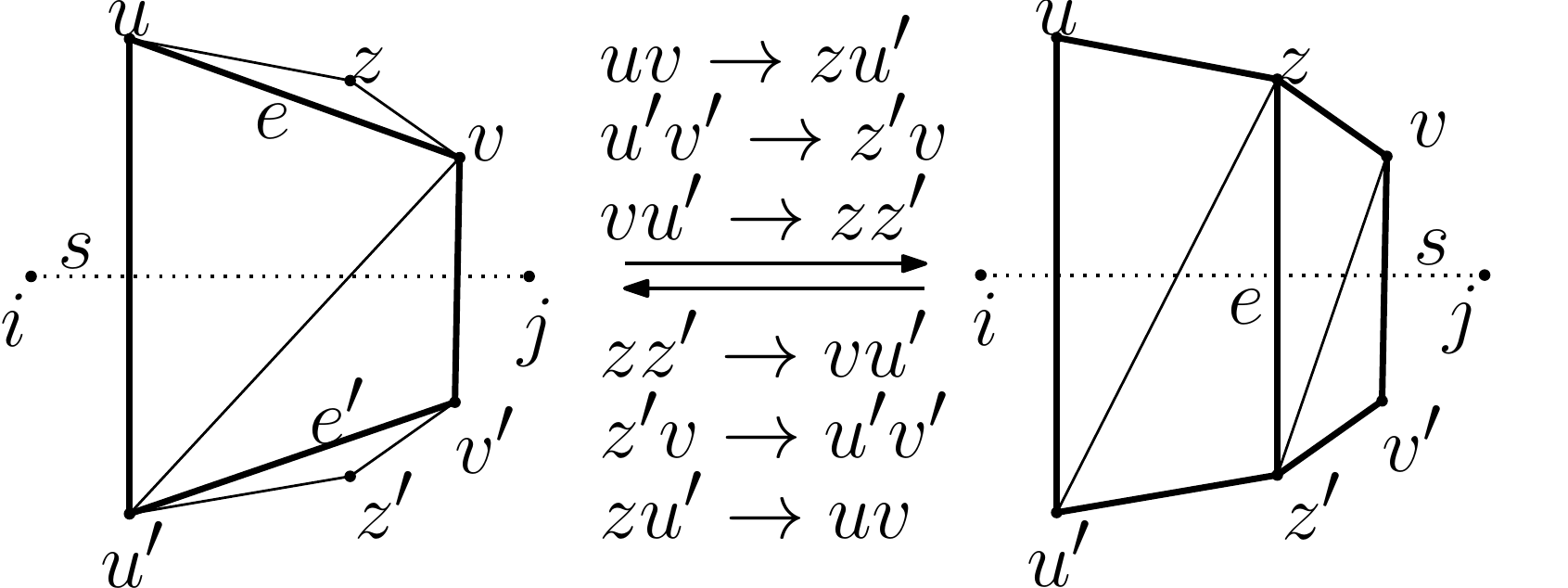}\\
\end{tabular}
\end{center}
\vspace{-0.1in}
\caption{Diagonal switches for the edges of the triangles crossing a segment on the boundary.
From left to right: the edges $e =uv, e'=u'v'$ always fail to be Delaunay at the same time. 
After their switches, the edges $u'v$ fails to be Delaunay. From right to left: the edges $e=zz', 
zu'$ and $vz'$ always fail to be Delaunay at the same time.  
\label{fig:switch_boundary_edge}}
\end{figure}

Let $w: \tilde{V} \rightarrow \tilde{V}$ with $\sum_{v\in \tilde{V}} w(v) = 0$ be the conformal factor 
so that $\tilde{d}' = w*\tilde{d}$. We claim $w$ respects the doubling structure, i.e., $w(v) = w(h(v))$ 
for any vertex $v\in \tilde{V}$. Otherwise, let us define a new conformal factor $w'$ so that $w'(v) = w(h(v))$ 
for any vertex $v\in \tilde{V}$, and then $w'\neq w$, which from Lemma~\ref{lem:conformalmetricspace} 
implies the metric $w'* \tilde{d}$ is different from  $\tilde{d}'$. However, it is easy to verify that the 
curvature of the metric $w'*\tilde{d}$ is also equal to $\tilde{K}^*$.
This contradicts to the uniqueness of $\tilde{d}'$. 

Now let $w(t) = tw$ for $t\in [0, 1]$ be a path from $0$ to $w$, and we have $w(t)$ respects the doubling 
structure for any $t$. As discussed in Section~\ref{sec:convexenergy}, $\tilde{d}(t) = w(t)*\tilde{d}$ for $t\in [0, 1]$
is a path in $C(\tilde{d})$. Let $0=t_0<t_1<t_2<\cdots<t_m=1$ is a partition of $[0, 1]$ so that for any $0\leq i\leq m-1$, 
$\tilde{d}(t)$ with $t\in [t_i, t_{i+1}]$ is a path inside the cell $\mathcal{M}_D(T_i)$ for some triangulation $T_i$.
If $\tilde{d}(t_i)$ respects the doubling structure and $T_i$ satisfies the symmetric property in the metric $\tilde{d}(t_i)$, 
then $T_i$ remains so in any PL metric $\tilde{d}(t)$ for any $t\in [t_i, t_{i + 1}]$. Indeed, the symmetric 
property (1) obviously holds as $w(t)$ respects the doubling structure. To show the symmetric property (2), it suffices 
to show is that the quadrilateral $f\cup f'$ remains cocircular.  This can be done by verifying that the sum of 
the cosines of the angles opposite to the diagonal remains $0$ along the path $w(t)$. 
Furthermore, consider the region $\cup_{f\in F_s(T_i)} f$, 
as shown in Figure~\ref{fig:boundary_triangles}. One can cut it into two geometrically identical subregions using a 
straight line connecting the endpoints of the segment $s$ and passing through the midpoints of the edges
in $F_s(T_i)$ of the form $uu'$ with $u' = h(u)$. This shows that $\tilde{d}(t)$ respects the doubling structure for 
any $t\in [t_i, t_{i + 1}]$, in particular, so is $\tilde{d}(t_{i+1})$. 
Now by construction, $\tilde{d}(t_0)$ respects the doubling structure.  From the previous discussion, $\tilde{d}(t_0)$
lies in the cell $\mathcal{M}_D(T_0)$ where $T_0$ satisfies the symmetric property. 
Then using induction, we show that $\tilde{d}'= \tilde{d}(1)$ respects the double structure. 
The restriction of $\tilde{d}'$ onto $S$ is the PL metric $d'$ on $(S, V)$. Finally, it is easy to verify 
that the curvature on $S$ in $d'$ equals $K^*$. This proves the theorem.  
\end{proof}

\end{document}